\documentclass[11pt]{article}
\usepackage[letterpaper, margin=1in]{geometry}
\usepackage{amsfonts,amsmath,amssymb,amsthm}
\usepackage{cite,comment,enumerate,hyperref}
\usepackage{cleveref}
\usepackage[ruled,vlined,linesnumbered]{algorithm2e}
\usepackage{thmtools}
\usepackage{thm-restate}
\usepackage[utf8]{inputenc}
\usepackage[T1]{fontenc}
\usepackage{xcolor}
\usepackage[shortlabels]{enumitem}

\newtheorem{theorem}{Theorem}[section]

\newtheorem{observation}{Observation}

\newtheorem{lemma}[theorem]{Lemma}
\newtheorem{claim}[theorem]{Claim}
\theoremstyle{definition}
\newtheorem{definition}{Definition}
\theoremstyle{remark}	

\usepackage{todonotes}
\newcommand{\ytodo}[1]{\todo[color=blue!20]{Y: #1}}

\newcommand{\ttodo}[1]{\todo[color=orange!20]{T: #1}}

\newcommand{\sebastian}[1]{\todo[color=green!20]{S: #1}}
\newcommand{\mtodo}[1]{\todo[color=red!20]{M: #1}}

\newcommand{\dist}{\operatorname{d}}
\newcommand{\remove}[1]{}

\newcommand{\lev}{\ell}

\DeclareMathOperator{\E}{\mathbb E}
\DeclareMathOperator*{\argmin}{arg\,min}

\DeclareMathOperator{\dLNbunch}{\delta_{\tilde{B}}}

\SetKwProg{procedure}{Procedure}{}{}
\SetKwFunction{DynAPSPtwo}{APSP$_2$}
\SetKwFunction{DynAPSPtwoW}{APSP$_{2,W}$}
\SetKwFunction{BDynAPSPtwo}{BoundedAPSP$_2$}
\SetKwFunction{StaticAPSPtwo}{StaticAPSP$_2$}
\SetKwFunction{Delete}{Delete}
\SetKwFunction{Increase}{Increase}
\SetKwFunction{Update}{Update}
\SetKwFunction{Distance}{Distance}
\SetKwFunction{UpdateOverlapDS}{UpdateOverlapDS}
\SetKwFunction{Update}{Update}
\SetKwFunction{Initialization}{Initialization}
\SetKwFunction{UpdateBunchRadius}{UpdateBunchRadius}
\SetKwFunction{UpdateAdjacentDS}{UpdateAdjacentDS}
\SetKwFunction{Query}{Query}
\SetKwFunction{MaintainBunches}{MaintainBunches}
\SetKwFunction{UpdateNbrBunchDSBunchChange}{UpdateNbrBunchDSBunchChange}
\SetKwFunction{UpdateNbrBunchDSEdgeChange}{UpdateNbrBunchDSEdgeChange}
\SetKwFunction{UpdateAdjacentDSBunchChange}{UpdateAdjacantDSBunchChange}
\SetKwFunction{UpdateAdjacentDSMinHeapChange}{UpdateAdjacentDSMinHeapChange}
\SetKwFunction{UpdateEstimate}{UpdateEstimate}
\SetKwFunction{Scale}{Scale}
\SetKwFunction{PreProcessing}{PreProcessing}

\makeatletter
\renewcommand{\paragraph}{%
  \@startsection{paragraph}{4}%
  {\z@}{1.4ex \@plus 1ex \@minus .2ex}{-1em}%
  {\normalfont\normalsize\bfseries}%
}
\makeatother

\usepackage{titlesec}
\titlespacing*
    {\subsection}
    {0pt}
    {1.4ex}
    {1ex}

\titlespacing*
    {\subsubsection}
    {0pt}
    {1.4ex}
    {.5ex}

\newcommand*\samethanks[1][\value{footnote}]{\footnotemark[#1]}

\title{New Tradeoffs for Decremental Approximate All-Pairs Shortest~Paths}
\author{Michal Dory\thanks{University of Haifa. This work was partially conducted while the author was a postdoc at ETH Zurich. This work was supported in part by funding from the European Research Council (ERC) under the European Union’s Horizon 2020 research and innovation programme (grant agreement No. 853109), and the Swiss National Foundation (project grant 200021-184735).} \and  Sebastian Forster\thanks{Department of Computer Science, University of Salzburg. This work is supported by the Austrian Science Fund (FWF): P 32863-N. This project has received funding from the European Research Council (ERC) under the European Union's Horizon 2020 research and innovation programme (grant agreement No~947702).} \and Yasamin Nazari\thanks{VU Amsterdam. This work was partially conducted while the author was a postdoc at University of Salzburg.} \and Tijn de Vos\samethanks[2]}
\date{}

\usepackage{lineno}

\begin{document}
\begin{titlepage}
\maketitle
\begin{abstract}
    We provide new tradeoffs between approximation and running time for the decremental all-pairs shortest paths (APSP) problem. For undirected graphs with $m$ edges and $n$ nodes undergoing edge deletions, we provide four new approximate decremental APSP algorithms, two for weighted and two for unweighted graphs. 
Our first result is $(2+ \epsilon)$-APSP with total update time $\tilde{O}(m^{1/2}n^{3/2})$ (when $m= n^{1+c}$ for any constant $0<c<1$).
Prior to our work the fastest algorithm for weighted graphs with approximation at most $3$ had total $\tilde O(mn)$ update time for $(1+\epsilon)$-APSP [Bernstein, SICOMP 2016]. 
Our second result is  $(2+\epsilon, W_{u,v})$-APSP with total update time $\tilde{O}(nm^{3/4})$, where the second term is an additive stretch with respect to $W_{u,v}$, the maximum weight on the shortest path from $u$ to $v$.  

Our third result is $(2+ \epsilon)$-APSP for unweighted graphs in $\tilde O(m^{7/4})$ update time, which for sparse graphs ($m=o(n^{8/7})$) is the first subquadratic $(2+\epsilon)$-approximation.
Our last result for unweighted graphs is $(1+\epsilon, 2(k-1))$-APSP, for $k \geq 2 $, with $\tilde{O}(n^{2-1/k}m^{1/k})$ total update time (when $m=n^{1+c}$ for any constant $c >0$). 
For comparison, in the special case of $(1+\epsilon, 2)$-approximation, this improves over the state-of-the-art algorithm by [Henzinger, Krinninger, Nanongkai, SICOMP 2016] with total update time of $\tilde{O}(n^{2.5})$.
All of our results are randomized, work against an oblivious adversary, and have constant query time. 

     
\end{abstract}

\thispagestyle{empty}
\newpage

\tableofcontents
\thispagestyle{empty}
\end{titlepage}

\newpage

\section{Introduction}
The dynamic algorithms paradigm is becoming increasingly popular for studying algorithmic questions in the presence of gradually changing inputs.
A natural goal in this area is to design algorithms that process each change to the input as fast as possible to adapt the algorithm's output (or a data structure for querying the output) to the current state of the input.
The time spent after an update to perform these computations is called the \emph{update time} of the algorithm.
In many cases, bounds on the update time are obtained in an \emph{amortized} sense as an average over a long enough sequence of updates. Among dynamic graph problems, the question of maintaining exact or approximate shortest paths has received considerable attention in the past two decades. The main focus usually lies on maintaining a distance oracle that answers queries for the distance between a pair of nodes.
For this problem, we call an algorithm \emph{fully dynamic} if it supports both insertions and deletions of edges, and \emph{partially dynamic} if it supports only one type of updates; in particular we call it \emph{decremental} if it only supports edge deletions  (which is the focus of this paper), and \emph{incremental} if it only supports edge insertions. 

The running times of partially dynamic algorithms are usually characterized by their bounds on the \emph{total update time}, which is the accumulated time for processing all updates in a sequence of at most $ m $ deletions (where $ m $ is the maximum number of edges ever contained the graph).
A typical design choice, which we also impose in this paper, is small (say polylogarithmic) query time. In particular, our algorithms will have constant query time.
While fully dynamic algorithms are more general, the restriction to only one type of updates in partially dynamic algorithms often admits much faster update times.
In particular, some partially dynamic algorithms have a total update time that almost matches the running time of the fastest static algorithm, i.e., computing \emph{all} updates does not take significantly more time than processing the graph once.

\paragraph{Decremental shortest paths.} For the decremental single-source shortest paths problem, conditional lower bounds~\cite{RodittyZ11,HenzingerKNS15} suggest that \emph{exact} decremental algorithms have an $\Omega(mn)$ bottleneck in their total update time (up to subpolynomial factors).
On the other hand, this problem admits a $ (1 + \epsilon) $-approximation (also called \emph{stretch}) with total update time $ m^{1 + o(1)} $ in weighted, undirected graphs~\cite{BR11,henzinger2014decremental,BernsteinC16,Bernstein17,BernsteinGS21}, which exceeds the running time of the state-of-the art static algorithm by only a subpolynomial factor. Hence for the single-source shortest paths problem on undirected graphs we can fully characterize for which multiplicative stretches the total update time of the fastest decremental algorithm matches (up to subpolynomial factors) the running time of the fastest static algorithm.

Obtaining a similar characterization for the decremental all-pairs shortest paths (APSP) problem is an intriguing open question.
Conditional lower bounds~\cite{DHZ00,HenzingerKNS15} suggest an $\Omega(mn)$ bottleneck in the total update time of decremental APSP algorithms with (a) any finite stretch on \textit{directed} graphs and (b)  with any stretch guarantee with a multiplicative term of $ \alpha \geq 1 $ and an additive term of $ \beta \geq 0 $ with $ 2 \alpha + \beta < 4 $ on \textit{undirected} graphs. This motivates the study of decremental $ (\alpha, \beta) $-approximate APSP algorithms such that $ 2 \alpha + \beta \geq 4 $, i.e., with multiplicative stretch $ \alpha \geq 2 $ or additive stretch $ \beta \geq 2 $.

Apart from two notable exceptions~\cite{HKN2016,AbrahamC13}, all known decremental APSP algorithms fall into one of two categories.
They either (a) maintain exact distances or have a relatively small multiplicative stretch of $ (1 + \epsilon) $ or (b) they have a stretch of at least~$3$.
The space in between is largely unexplored in the decremental setting.
This stands in sharp contrast to the static setting where for undirected graphs approximations guarantees different from multiplicative $ (1 + \epsilon) $ or ``$ 3 $ and above'' have been the focus of a large body of works~\cite{AingworthCIM99,DHZ00,CohenZ01,BaswanaK10,BermanK07,Kavitha12,BaswanaGS09,PatrascuR14,Sommer16,Knudsen17,AkavR20,AkavR21}. 
The aforementioned exceptions in dynamic algorithms~\cite{HKN2016,AbrahamC13}, concern unweighted undirected graphs and maintain $ (1 + \epsilon, 2) $-approximations that simultaneously have a multiplicative error of~$ 1 + \epsilon $ and an additive error of~$ 2 $, which implies a purely multiplicative $ (2 + \epsilon) $-approximation. The algorithms run in $\tilde O(n^{5/2})$\footnote{In the introduction, we make two simplifying assumptions: (a) $ \epsilon $ is a constant and (b) the ratio $ W $ between the maximum and minimum weight is polynomial in $ n $. Unless otherwise noted, the cited algorithms have constant or polylogarithmic query time. Throughout we use $ \tilde O (\cdot) $ notation to omit factors that are polylogarithmic in $ n $.} and $n^{5/2+o(1)}$ total update time respectively. In the next section, we will detail how our algorithms improve over this total update time and generalize to weighted graphs. 

Concerning the fully dynamic setting, Bernstein~\cite{Bernstein09} provided an algorithm for $(2+\epsilon)$-APSP which takes $m^{1+o(1)}$ time per update. Although we are not aware of any lower bounds, it seems to be hard to beat this: no improvements have been made since. The lack of progress in the fully-dynamic setting motivates the study in a partially dynamic setting, where we obtain improvements for the decremental case.

\subsection{Our Results}
In this paper, we provide novel decremental APSP algorithms with approximation guarantees that previously were mostly unexplored in the decremental setting.
Our algorithms are randomized and we assume an oblivious adversary. 
For each pair of nodes, we can not only provide the distance estimate, but we can also report a shortest path~$\pi$ of this length in $\tilde O(|\pi|)$ time using standard techniques, see e.g.~\cite{ForsterGNS23}. 

\paragraph{$\mathbf{(2+\epsilon)}$-APSP for weighted graphs.} Our first contribution is an algorithm for maintaining a $ (2 + \epsilon) $-approximation in weighted, undirected graphs.

\begin{restatable}{theorem}{mainthm}\label{thm:main_mult}
    Given a weighted graph $G$ and a constant $0 < \epsilon < 1$, there is a decremental data structure that maintains a $(2+\epsilon)$-approximation of APSP. The algorithm has constant query time and the total update time is w.h.p.
    \begin{itemize}
        \item $\tilde O(m^{1/2}n^{3/2}\log^2(nW))$ if $m=n^{1+\Omega(1)}$ and $m\leq n^{2-\rho}$ for an arbitrary small constant $\rho$,
        \item $\tilde O(m^{1/2}n^{3/2+\rho}\log^2(nW))$ if $m\geq n^{2-\rho}$ for any constant $\rho$,
        \item $\tilde O(m^{1/2}n^{3/2+o(1)}\log^2(nW))$ otherwise,
    \end{itemize} where $W$ is the ratio between the maximum and minimum weight.
\end{restatable}


The fastest known algorithm with an approximation ratio at least as good as ours is the $ (1 + \epsilon) $-approximate decremental APSP algorithm by Bernstein~\cite{Bernstein16} with total update time $ \tilde O (m n) $. 
With our $ (2 + \epsilon) $-approximation, we improve upon this total update time when $ m =n^{1+\eta} $, for any $\eta>0$.
Furthermore, the fastest known algorithm with a larger approximation ratio than our algorithm is the $ (3 + \epsilon) $-approximate decremental distance oracle by Łącki and Nazari~\cite{LN2020} with total update time $ \tilde O (m \sqrt{n}) $. 
Our result also has to be compared to the \emph{fully dynamic} algorithm of Bernstein~\cite{Bernstein09} for maintaining a $ (2 + \epsilon) $-approximation that takes amortized time $ m^{1 + o(1)} $ per update.
Note that in \emph{unweighted} graphs, $ (2 + \epsilon) $-approximate decremental APSP algorithm can be maintained with total update time $ \tilde O (n^{2.5}) $ (which is implied by the results of Abraham and Chechik~\cite{AbrahamC13} and Henzinger, Krinninger, and Nanongkai~\cite{HKN2016}). Our approach improves upon this bound for most densities, and matches it for $m=n^2$.

\paragraph{$\mathbf{(2+\epsilon, 1)}$-APSP for unweighted and $\mathbf{(2+\epsilon, W_{u,v})}$-APSP for weighted graphs.} 
Our second contribution is a faster algorithm with an additional additive error term of $1$ for unweighted graphs, which in turn can be used to obtain our unweighted $(2+\epsilon)$-APSP result.
The corresponding generalization to weighted graphs can be formulated as follows.

\begin{restatable}{theorem}{thmtwoW}\label{thm:(2_1)approx}
    Given a weighted graph $G$ and a constant $0 < \epsilon < 1$, there is a decremental data structure that maintains a $(2+\epsilon,W_{u,v})$-approximation for APSP, where $W_{u,v}$ is the maximum weight on a shortest path from $u$ to $v$. The algorithm has constant query time and the total update time is w.h.p.~$\tilde O( nm^{3/4}\log^2(nW))$ if $m=n^{1+\Omega(1)}$, and $\tilde O( n^{1+o(1)}m^{3/4}\log^2(nW))$ otherwise, where $W$ is the ratio between the maximum and minimum weight. 
\end{restatable}

\paragraph{$\mathbf{(2+\epsilon)}$-APSP for unweighted graphs.}
We obtain this result by a general reduction from mixed approximations to purely multiplicative approximations, which might be of independent interest beyond the dynamic setting. See Theorem~\ref{thm:reduction} for the statement. We can then combine this with Theorem~\ref{thm:(2_1)approx} to obtain a fast algorithm for non-dense unweighted graphs. 

\begin{restatable}{theorem}{thmtwoUnw}\label{thm:unweighted_2}
    Given an undirected unweighted graph $G$, there is a decremental data structure that maintains a $(2+\epsilon)$-approximation for APSP with constant query time. W.h.p.\ the total update time is bounded by $\tilde O( m^{7/4})$ if $m=n^{1+\Omega(1)}$, and $\tilde O( m^{7/4+o(1)})$ otherwise.
\end{restatable}

Note that for $m=o(n^{8/7})$ this gives the first subquadratic decremental $(2+\epsilon)$-approximate APSP algorithm. 
For $m\leq n^{6/5}$, this approach is faster than our weighted result, \Cref{thm:main_mult}, which was already beating the unweighted state-of-the-art~\cite{AbrahamC13,HKN2016}.

\paragraph{$\mathbf{(1+\epsilon, 2(k-1))}$-APSP for unweighted graphs.}Our fourth contribution is an algorithm for unweighted, undirected graphs that maintains, for any $ k \geq 2 $, a ($1 + \epsilon$,$2(k-1)$)-approximation, i.e., a distance estimate that has a multiplicative error of~$ 1 + \epsilon $ and an additive error of~$ 2(k-1) $.
\begin{restatable}{theorem}{thmadd}\label{thm:main_add}
Given an undirected unweighted graph $G$, a constant $0 < \epsilon < 1$ and an integer $2 \leq k \leq \log{n}$, there is a decremental data structure that maintains $(1+\epsilon, 2(k-1))$-approximation for APSP with constant query time. The expected total update time is bounded by $\min\{O(n^{2-1/k+o(1)}m^{1/k}), \tilde{O}((n^{2-1/k}m^{1/k}) O(1/\epsilon)^{k/\rho})\}$ where $m=n^{1+\rho}$.
\end{restatable}

Note that for small values of $k$ and if $m=n^{1+\rho}$ for a constant $\rho>0$, we get total update time of $\tilde{O}(n^{2-1/k}m^{1/k})$, and otherwise we have an extra $n^{o(1)}$ factor. In addition, in the special case of $k=\log{n}$, we get a near-quadratic update time of $O(n^{2+o(1)})$.
The state-of-the-art for a purely multiplicative $ (1 + \epsilon) $-approximation is the algorithm of Roditty and Zwick with total update time $ \tilde O (m n) $.\footnote{Note that the algorithms of Roditty and Zwick~\cite{RZ12} for unweighted, undirected graphs precedes the more general algorithm of Bernstein~\cite{Bernstein16} for weighted, directed graphs.}
It was shown independently by Abraham and Chechik~\cite{AbrahamC13} and by Henzinger, Krinninger, and Nanongkai~\cite{HKN2016} how to improve upon this total update time bound at the cost of an additional small additive error term: a $ (1 + \epsilon, 2) $-approximation can be maintained with total update time $ \tilde O (n^{2.5}) $.
This has been generalized by Henzinger, Krinninger, and Nanongkai~\cite{HKN14} to an additive error term of $ 2 (1 + \tfrac{2}{\epsilon})^{k-2} $ and total update time $ \tilde O (n^{2 + 1/k} O(\tfrac{1}{\epsilon})^{k-1}) $.
We improve upon this tradeoff in two ways: (1) our additive term is independent of $ 1/\epsilon $ and \emph{linear} in $k$ and (2) our algorithm profits from graphs being sparse.

\paragraph{Static follow-up work.} Interestingly, our techniques inspired a new algorithm for a static $2$-approximate distance oracle~\cite{DoryFKNVV23}, giving the first such distance oracle algorithm with subquadratic construction time for sparse graphs. Moreover, for $m=n$ they match conditional lower bounds~\cite{PatrascuRT12, aboud2022stronger}.


\subsection{Related Work}

\paragraph{Static algorithms.}
The baseline for the static APSP problem are the exact textbook algorithms with running times of $ O (n^3) $ and $ \tilde O (m n) $, respectively.
There are several works obtaining improvements upon these running times by either shaving subpolynomial factors or by employing fast matrix multiplication, which sometimes comes at the cost of a $ (1 + \epsilon) $-approximation instead of an exact result.
See~\cite{Zwick01} and~\cite{Williams18}, and the references therein, for details on these approaches.
For the regime of stretch $ 3 $ and more in undirected graphs, a multitude of algorithms has been developed with the distance oracle of Thorup and Zwick~\cite{TZ2005} arguably being the most well-known constructions.
In the following, we focus on summarizing the state of affairs for approximate APSP with stretch between $ 1 + \epsilon $ and $ 3 $.

In \emph{weighted}, undirected graphs, Cohen and Zwick~\cite{CohenZ01} obtained a $2$-approximation with running time $ \tilde O (m^{1/2} n^{3/2}) $.
This running time has been improved to $ \tilde O (m \sqrt{n} + n^2) $ by Baswana and Kavitha~\cite{BaswanaK10} and, employing fast matrix multiplication, to $ \tilde O (n^{2.25}) $ by Kavitha~\cite{Kavitha12}.
For $(2+\epsilon)$-APSP, Dory et al.~\cite{DoryFKNVV23} provide two algorithms, for dense graphs: $ O(n^{2.214})$ and sparse graphs: $\tilde O(mn^{2/3})$.
In addition, efficient approximation algorithms for stretches of $ \tfrac{7}{3} $~\cite{CohenZ01} and $ \tfrac{5}{2} $~\cite{Kavitha12} have been obtained.
These results have recently been generalized by Akav and Roditty~\cite{AkavR21} who presented an algorithm with stretch $ 2 + \tfrac{k-2}{k} $ and running time $ \tilde O (m^{2/k} n^{2 - 3/k} + n^2) $ for any $ k \geq 2 $.

In \emph{unweighted}, undirected graphs, a $ (1, 2) $-approximation algorithm with running time $ \tilde O (n^{2.5}) $ has been presented by Aingworth, Chekuri, Indyk and Motwani~\cite{AingworthCIM99}. This has been improved by Dor, Halperin, and Zwick~\cite{DHZ00} that showed a $(1,2)$-approximation with running time $\tilde{O}(\min\{n^{3/2}m^{1/2}, n^{7/3} \})$. They also show a generalized version of the algorithm that gives stretch $ (1, k) $ and running time ${\tilde{O}}(\min\{n^{2-{2}/{(k+2)}}m^{{2}/{(k+2)}}, n^{2+{2}/{(3k-2)}}\})$ for every even $k>2$.
Recently faster $(1,2)$-approximation algorithms based on fast matrix multiplication techniques were developed \cite{DengKRWZ22,durr2023improved}, the fastest of them runs in $O(n^{2.260})$ time \cite{durr2023improved}. This was extended to a $(1+\epsilon,2)$-approximation in $O(n^{2.152})$~\cite{DoryFKNVV23}. In addition, recently Roditty \cite{Roditty23} extended the approach of \cite{DHZ00} to obtain a combinatorial $(2,0)$-approximation for APSP in $\tilde{O}(n^{2.25})$ time in unweighted undirected graphs.  Using fast matrix multiplication, this running time was improved to $O(n^{2.032})$~\cite{DoryFKNVV23,saha2024faster}.

Berman and Kasiviswanathan~\cite{BermanK07} showed how to compute a $ (2, 1) $-approximation in time $ \tilde O (n^2) $.
Subsequent works~\cite{BaswanaGS09,BaswanaK10,PatrascuR14,Sommer16,Knudsen17} have improved the polylogarithmic factors in the running time and the space requirements for such nearly $2$-approximations.
Recently, slightly subquadratic algorithms have been given: an algorithm with stretch $ (2 (1 + \epsilon), 5)$ by Akav and Roditty~\cite{AkavR20}, and an algorithm with stretch $(2,3)$ by Chechik and Zang~\cite{ChechikZ22}.

\paragraph{Decremental algorithms.}

The fastest algorithms for maintaining exact APSP under edge deletions have total update time $ \tilde O (n^3) $~\cite{DemetrescuI06,BaswanaHS07,EvaldFGW21}.
There are several algorithms that are more efficient at the cost of returning only an approximate solution.
In particular, a $ (1+\epsilon) $-approximation can be maintained in total time $ \tilde O (m n) $~\cite{RZ12,Bernstein16,KarczmarzL19}.
If additionally, an additive error of $ 2 $ is tolerable, then a $ (1+\epsilon, 2) $-approximation can be maintained in total time $ \tilde O (n^{2.5}) $ in unweighted, undirected graphs~\cite{HKN2016,AbrahamC13}.
Note that such a $ (1+\epsilon, 2) $-approximation directly implies a $ (2 + \epsilon) $-approximation because the only paths of length $ 1 $ are edges between neighboring nodes.
All of these decremental approximation algorithms are randomized and assume an oblivious adversary.
Deterministic algorithms with stretch $ 1 + \epsilon $ exist for unweighted, undirected graphs with running time $ \tilde O (m n) $~\cite{HKN2016}, for weighted, undirected graphs with running time $ \tilde O (m n^{1+o(1)})$~\cite{BernsteinGS21} and for weighted, directed graphs with running time $ \tilde O (n^3) $~\cite{KarczmarzL20}.

Decremental approximate APSP algorithms of larger stretch, namely at least~$ 3 $, have first been studied by Baswana, Hariharan, and Sen~\cite{BaswanaHS03}.
After a series of improvements~\cite{RZ12, BR11, AbrahamCT14, henzinger2014decremental}, the state-of-the-art algorithms of~\cite{Chechik18,LN2020} maintain $ (2k-1)(1 + \epsilon) $-approximate all-pairs shortest paths for any integer $ k \geq 2 $ and $ 0 < \epsilon \leq 1 $ in total update time $ \tilde O ( (m+n^{1+o(1)})n^{1/k}) $ with query time  $O(\log \log (nW))$ and $O(k)$ respectively.
All of these ``larger stretch'' algorithms are randomized and assume an oblivious adversary.

Recently, deterministic algorithms have been developed by Chuzhoy and Saranurak~\cite{ChuzhoyS21} and by Chuzhoy~\cite{Chuzhoy21}.
One tradeoff in the algorithm of Chuzhoy~\cite{Chuzhoy21} for example provides total update time $ \tilde O (m^{1 + \mu}) $ for any constant $ \mu $ and polylogarithmic stretch.
As observed by M\k{a}dry~\cite{Madry10}, decremental approximate APSP algorithms that are deterministic -- or more generally work against an adaptive adversary -- can lead to fast static approximation algorithms for the maximum multicommodity flow problem via the Garg-Könemann-Fleischer framework~\cite{GargK07,Fleischer00}.
The above upper bounds for decremental APSP have recently been contrasted by conditional lower bounds~\cite{AbboudBKZ22} stating that constant stretch cannot be achieved with subpolynomial update and query time under certain hardness assumptions on 3SUM or (static) APSP.

\paragraph{Fully dynamic algorithms.}
The reference point in fully dynamic APSP with subpolynomial query time is the \emph{exact} algorithm of Demetrescu and Italiano~\cite{DemetrescuI04} with update time $ \tilde O (n^2) $ (with log-factor improvements by Thorup~\cite{Thorup04}).
For undirected graphs, several fully dynamic distance oracles have been developed.
In particular, Bernstein~\cite{Bernstein09} developed a distance oracle of stretch $ 2 + \epsilon $ (for any given constant $ 0 < \epsilon \leq 1 $) and update time $ O (m^{1 + o(1)}) $.
In the regime of stretch at least~$ 3 $, tradeoffs between stretch and update time have been developed by Abraham, Chechik, and Talwar~\cite{AbrahamCT14}, and by Forster, Goranci, and Henzinger~\cite{ForsterGH21}.
Finally, most fully dynamic algorithm with update time sensitive to the edge density can be combined with a fully dynamic spanner algorithm leading to faster update time at the cost of a multiplicative increase in the stretch, see~\cite{BaswanaKS12} for the seminal work on fully dynamic spanners.

\section{High-Level Overview}
 In this section we provide a high-level overview of our algorithms. First we describe our $(2+\epsilon)$-APSP algorithms for weighted graphs, by giving a simpler static version first. Then we describe our $(2+\epsilon,W_{u,v})$-APSP algorithm, where we introduce a notion of \emph{bunch overlap threshold} to overcome the challenge of dynamically maintaining a well-known static adaptive sampling technique~\cite{TZ01} (see \Cref{sc:overview2W}). 
 We give a reduction from mixed approximations to purely multiplicative approximations, which together with the previous result gives our $(2+\epsilon)$-APSP algorithm for unweighted graphs. 
Finally, we describe our $(1+\epsilon, 2(k-1))$-APSP algorithm for any $k \geq 2$ in unweighted graphs.

\subsection{\texorpdfstring{$(2+\epsilon)$}{(2+epsilon)}-APSP for Weighted Graphs}\label{sc:high-level_2approx} 
We first present a warm-up static algorithm, that we later turn into a dynamic algorithm. 

\subsubsection{Static \texorpdfstring{$2$}{2}-APSP} 
We start by reviewing the concepts of bunch and cluster as defined in the seminal distance oracle construction of \cite{TZ2005}.

Let $0<p <1$ be a parameter and $A$ a set of nodes sampled with probability $p$. For each node $u$, the pivot $p(u)$ is the closest node in $A$ to $u$. The \textit{bunch} $B(v)$ of a node $v$ is the set $B(v):=\{ u \in V : d(u,v) < d(v,p(v)) \}$ and the \textit{clusters} are $C(u):= \{ v \in V: d(u,v) < d(v,p(v)) \}=\{v \in V: u\in B(v)\}$. Hence bunches and clusters are reverse of each other: $u \in B(v)$ if any only if $v \in C(u)$.

 These structures have been widely used in various settings, including the decremental model (e.g.~\cite{HKN2016, RodittyZ11, LN2020, Chechik18}). We can use these existing decremental algorithms to maintain bunches and clusters, but will need to develop new decremental tools to get our desired tradeoffs that we will describe in \Cref{sec:overview_dynamic}.

Using the well-known distance oracles of \cite{TZ2005}, we know that by storing the clusters and bunches and the corresponding distances inside them, we can query $3$-approximate distances for any pair. At a high-level, for querying distance between a pair $u, v$ we either have $u \in B(v)$ or $v \in B(u)$ and so we explicitly have stored their distance, or we can get a $3$-approximate estimate by computing $\min \{d(u,p(u)+d(p(u),v), d(v, p(v))+d(p(v), u)\}$. In the following we explain that in some special cases we can use these estimates to a obtain $2$-approximation for a pair $u,v$, and in other cases by storing more information depending on how $B(u)$ and $B(v)$ overlap we can also obtain a $2$-approximation to $d(u,v)$. 

Let us assume that for all nodes $v$ we have computed $B(v)$ and have access to all the distances from $v$ to each node $u \in B(v)$. Also assume that we have computed the distances from the set~$A$ to all nodes.
We now describe how using this information we can get a $2$-approximate estimate of \textit{any} pair of nodes $u,v$ in one case. In another case we will discuss what other estimates we need to precompute. Let $\pi$ be the shortest path between $u$ and $v$.
 We consider two cases depending on how the bunches $B(u)$ and $B(v)$ interact with $\pi$. A similar case analysis was used in \cite{censor2021fast, dory2022exponentially} in distributed approximate shortest paths algorithms. 
\begin{enumerate}
    \setlength\itemsep{0em}
    \item \label{overview:outside_bunch} There exists $w \in \pi$ such that $w \not \in B(u) \cup B(v)$ (left case in Figure~\ref{fig:pivots}): Since $w$ is on the shortest path, we either have $d(w,u) \leq d(u,v)/2$ or $d (v,w) \leq d(u,v)/2$. Suppose $d(w,u) \leq d(u,v)/2$. We observe that by definition $d(u,p(u))\leq d(u,w)$, hence we obtain $d(u,p(u))+d(p(u),v)\leq d(u,p(u))+d(u,p(u))+d(u,v) \leq 2d(u,v)$. The case that $d (v,w) \leq d(u,v)/2$ is analogous, and hence by computing $\min \{ d(u,p(u))+d(p(u),v), d(v,p(v))+d(p(v),u)\}$ we get a $2$-approximation.
 
    \item \label{overview:adj_bunches} There exists no $w \in \pi$ such that $w \not \in B(u) \cup B(v)$. In other words, $B(u) \cap B(v) \cap \pi = \emptyset$ and there exists at least one edge $\{u',v'\}$ on $\pi$ where $u' \in B(u)$ and $v' \in B(v)$ (right case in Figure~\ref{fig:pivots}\footnote{Although the picture implies $B(u)$ and $B(v)$ are disjoint, this also includes the case where they overlap, or even where $u\in B(v)$ or $v\in B(u)$.}): in this case we can find the minimum over estimates obtained through all such pairs $u', v'$, i.e. by computing $d(u,u')+ w(u',v')+d(v',v)$. 
\end{enumerate}

\begin{figure}[htbp!]
\centering
\includegraphics[width=\textwidth]{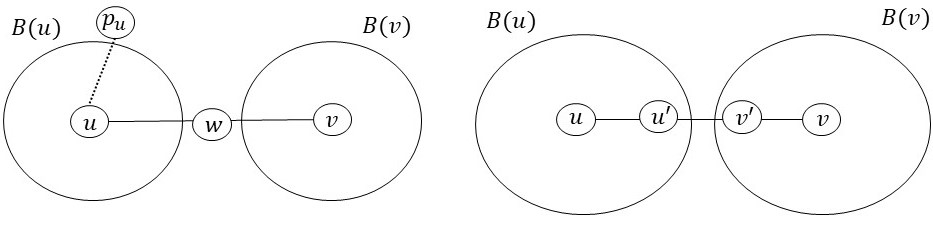}
\caption{\footnotesize Possible scenarios for the overlap between the bunches $B(u)$, $B(v)$ and the shortest path $\pi$ from $u$ to $v$.}
\label{fig:pivots}
\end{figure}    

If we had access to all the above distances for the bunches and the pivots, we could then use them to query $2$-approximate distances between any pair $u,v \in V$. As we will see, in our algorithms we only have \textit{approximate bunches and pivots} which lead to $(2+\epsilon)$-approximate queries.

\paragraph{Running time.} 
We can compute bunches and clusters in $\tilde O(\tfrac{m}{p})$ time~\cite{TZ2005}. We need to compute shortest paths from $A$ for the estimates in Case~\ref{overview:outside_bunch}. This takes $O(|A|m)=\tilde O(pnm)$ time using Dijkstra's algorithm. Finally, we consider Case~\ref{overview:adj_bunches}. A straight-forward approach is to consider each pair of vertices $u,v$, and all $u'\in B(u)$ and $v'\in B(v)$. This takes $O(\tfrac{n^2}{p^2})$ time. 

We use a sophisticated intermediate step that computes the distance between $u'$ and $v$ if there is an edge $\{u',v'\}\in E$ such that $v'\in B(v)$. For each node $v \in V$, we consider all $v'\in B(v)$, for which we consider all neighbors $u'\in N(v')$ and compute the estimate $w(u',v')+d(v',v)$ (and replace current $u'-v$ minimum if it is smaller). Since $|B(v)|=\tilde O(\tfrac{1}{p})$, we can compute these estimates for $d(u',v)$ in $\tilde O(\tfrac{n^2}{p})$ time. 

Next, for each pair $u,v\in V$, we take the minimum over the distances $d(u,u')+d(u',v)$ for all $u'$ in the bunch $B(u)$ using the precomputed distances from the previous step in time $\tilde O(n^2/p)$, by iterating over all pairs $u,v\in V$, and all $u'\in B(u)$. 
So in total the adjacent case takes time $\tilde O(\tfrac{n^2}{p})$. Hence the total running time of the algorithm will be $\tilde O(pnm+m/p+n^2/p)=\tilde O(pnm+n^2/p)$. By setting $p=\sqrt{\tfrac{n}{m}}$ we obtain the total update time $\tilde O(m^{1/2}n^{3/2})$ as stated in Theorem~\ref{thm:main_mult}, for details see \Cref{sc:mult_result}.


There are several subtleties that make it more difficult to maintain these estimates in decremental settings since the bunches and which bunches are adjacent keep on changing over the updates. We next discuss how these can be handled. We obtain a total update time that matches the stated static running time (up to subpolynomial factors). 

\subsubsection{Dynamic Challenges}\label{sec:overview_dynamic}
\paragraph{Maintaining bunches and pivots.}
First, we need to dynamically maintain bunches efficiently as nodes may join and leave a bunch throughout the updates using an adaptation of prior work. One option would be maintaining the clusters and bunches using the \cite{RZ12} framework, however using this directly is slow for our purposes. Hence we maintain \textit{approximate} clusters and bunches using hopsets of \cite{LN2020}. This algorithm also maintains approximate pivots, i.e.~pivots that are within $(1+\epsilon)$-approximate distance of the true closest sampled node. These estimates let us handle the first case (up to $(1+\epsilon)$ approximation).
One subtlety is that the type of approximate bunches used in \cite{LN2020} is slightly different with the type of approximate bunches we need for other parts of our algorithm. Specifically we need to bound the number of times nodes in a bunch can change. Roughly speaking, we can show this since the graph is decremental and the estimates obtained from the algorithm of \cite{LN2020} are monotone. Hence we perform \emph{lazy bunch updates}, where we only let a bunch grow if the distance to the pivot grows by at least a factor $1+\epsilon$. This means a bunch grows at most $O(\log_{1+\epsilon}(nW))$ times. In addition, these approximate bunches need to be taken into account into our stretch analysis. 

We note that $(1+\epsilon)$-approximate decremental SSSP takes $\tilde O(m^{1+o(1)})$ total update time. However, we use the multi-source shortest path result from~\cite{LN2020}, where the $n^{o(1)}$-term vanishes when applied with polynomially many sources, giving us a total update time of $\tilde O(|A|m)$.

\paragraph{Maintaining estimates for Case~\ref{overview:adj_bunches}.}
When we are in Case~\ref{overview:adj_bunches}, we need more tools to keep track of estimates going through these nodes since both the bunches and the distances involved are changing. In particular, we explain how by using heaps in different parts of our algorithm we get efficient update and query times.
Moreover, note that we have an additional complication: there are two data structures in this step, where one impacts the other. We need to ensure that an update in the original graph does not lead to many changes in the first data structure, which all need to be processed for the second data structure. 

To be precise, we need intermediate data structures $Q_{u',v}$ that store approximate distances for $u'$ and $v$ of the form $w(u',v')+\delta(v',v)$ for each $v'\in B(v)$ such that $\{u',v'\}\in E$, as explained for the static algorithm\footnote{In this section, $\delta(\cdot,\cdot)$ denotes $(1+\epsilon)$-approximate distances in our dynamic data structures.}. As opposed to the static algorithm, we do not only keep the minimum, but instead maintain a min-heap, from which the minimum is easily extracted. 

However, this approach has a problem: $\delta(v',v)$ might change many times, and for each such change we need to update at most $|N(v')|\leq n$ heaps. To overcome this problem, we use another notion of \textit{lazy distance update}: we only update a bunch if the estimate of a node changes by at least a factor $1+\epsilon$. Combining this with the fact that, due to the lazy \textit{bunch} update, nodes only join a cluster at most $O(\log_{1+\epsilon}(nW))$ times, and thus there is only an $ O(\log_{1+\epsilon}(nW)\log(nW))$ overhead in maintaining these min-heaps $Q_{u',v}$ dynamically due to bunch updates. 

Similarly, $w(u',v')$ might change many times, which has an impact for every $v\in C(v')$. Since this can be many nodes, we need to use an additional trick so an adversary cannot increase our update time to $mn$. Again the solution is a lazy update scheme: instead of $w(u',v')$ we use $\tilde w(u', v')=(1+\epsilon)^{\lceil \log_{1+\epsilon} w(u',v')\rceil}$, which can only change $\log_{1+\epsilon} W$ times. This comes at the cost of a $(1+\epsilon)$ approximation factor.

We note that our bunches are approximate bunches in three different ways. We have one notion of approximation due to the fact that we are using hopsets (that implicitly maintain bunches on scaled graphs). We have a second notion of approximation due to our lazy bunch update, which only lets nodes join a bunch a bounded number of times. We have a third notion of approximation due to the lazy distance update, which only propagates distance changes $O(\log(nW))$ times. We need to carefully consider how these different notions of approximate bunch interact with each other and with the stretch of the algorithm.

Next, we want to use the min-heaps $Q_{u',v}$ to maintain the distance estimates for $u,v\in V$ (see again the right case in Figure~\ref{fig:pivots}). We construct min-heaps $Q'_{u,v}$, with for each pair $u,v$ entries $\delta(u,u')+\delta(u',v)$ for $u' \in B(u)$. Here $\delta(u',v)$ is the minimum from the intermediate data structure $Q_{u',v}$. 

An entry $\delta(u,u')+\delta(u',v)$ in $Q'_{u,v}$ has to be updated when either $\delta(u,u')$ or $\delta(u',v)$ change. We have to make sure we do not have to update these entries too often. For the first part of each entry, we bound updates to the distance estimates $\delta(u,u')$ by the lazy distance updates. One challenge here is that when the distance estimates in a bunch change we need to update all the impacted heaps i.e.~if $\delta(u,u')$ inside a bunch $B(u)$ changes, we must update all the heaps $Q'_{u,v}$ such that there exists $\{u',v'\} \in E$ with $v'\in B(v)$. For this purpose one approach is the following: for each $u \in V,u' \in B(u)$ we keep an additional data structure $\text{Set}_{u,u'}$ that stores for which $v \in V$ we have $v' \in B(v)$ with $\{u',v'\}\in E$. This means that given a change in $\delta(u,u')$, we can update each of the heaps in logarithmic time.

For the second part of the entry, which is the distance estimate $\delta(u',v)$ from the first data structure, this is more complicated. Another subtlety for our decremental data structures is the fact that $\delta(u',v)$ may also \text{decrease} due to a new node being added to $B(v)$. However, our lazy bunch update ensures that $\delta(u',v)$ can only decrease $O(\log_{1+\epsilon}(nW))$ times. Further, we use lazy distance increases \textit{in between} such decreases. Hence in total we propagate changes to $\delta(u',v)$ at most $O(\log_{1+\epsilon}(nW)\log(nW))$ times, and thus our decremental algorithm has $\tilde O(\log_{1+\epsilon}(nW)\log(nW))$ overhead compared to the static algorithm.

\subsection{\texorpdfstring{$(2+\epsilon, W_{u,v})$}{(2+epsilon,W{u,v})}-APSP for Weighted Graphs}\label{sc:overview2W}
For a moment, let us go back to the first static algorithm. It turns out, that if instead of computing an estimate for two `adjacent' bunches, we keep an estimate for bunches that overlap in at least one node, we obtain a $(2, W_{u,v})$-approximation, where $W_{u,v}$ is the maximum weight on the shortest path from $u$ to $v$. 

\begin{figure}[htbp!]
    \centering
    \includegraphics[width=\textwidth]{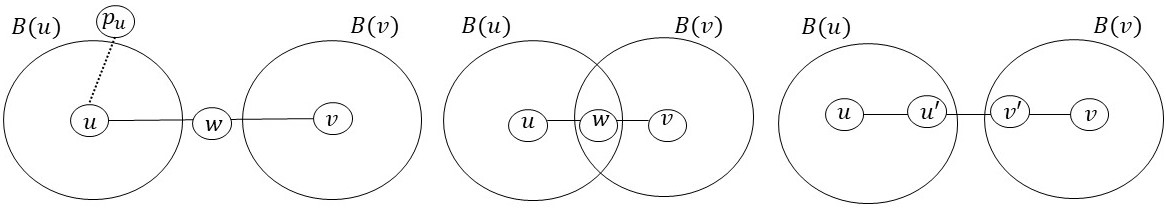}
    \caption{Possible scenarios for the overlap between the bunches $B(u)$ and $B(v)$ and the shortest path $\pi$ from $u$ to $v$.}
    \label{fig:pivots_allcases_overview}
\end{figure}

The stretch analysis comes down to the following cases:
\begin{itemize}
    \item There exists no $w\in \pi\cap B(u)\cap B(v)$ (the left and right case in \Cref{fig:pivots_allcases_overview}). 
    \item There exists $w\in \pi$ such that $w\in B(u)\cap B(v)$ (the middle case in \Cref{fig:pivots_allcases_overview}). \label{overview:overlap}
\end{itemize}

The crucial observation here is that the distance estimate via the pivot (Case~\ref{overview:outside_bunch} in the previous section) already gives a $(2, W_{u,v})$-approximation for the adjacent case (the right case in \Cref{fig:pivots_allcases_overview}), see \Cref{sec:2_W} for the details. So if we maintain an estimate for $d(u,v)$ if there exists $w\in \pi$ such that $w\in B(u)\cap B(v)$ (the `overlap case'), then we obtain a $(2, W_{u,v})$-approximation in total.

Here we need to create a min-heap for the overlap case, later we show we can maintain in $\tilde O(nC_{\max}/p)$ time. First we focus on the challenge of bounding $C_{\max}$.

\paragraph{Efficiency challenge.}
While the bunches are small, there is no bound on the size of the clusters. To overcome this issue in the static case, Thorup and Zwick developed an alternative way to build the bunches and clusters, that guarantees that the clusters are also small \cite{TZ01}. At a high-level, their approach is to adaptively change (grow) $A$ by sampling big clusters into smaller clusters using a sampling rate proportional to the cluster size and adding these nodes to the set of pivots $A$. In total, in $\tilde O(m/p)$ time we obtain bunches and clusters, \emph{both} of size $\tilde O(1/p)$.
This idea was developed in the context of obtaining compact routing schemes, and later found numerous applications in static algorithms for approximating shortest paths and distances (see e.g., \cite{BaswanaK10,AkavR20,backurs2021toward,AkavR21}).

Using this trick, the static algorithm has running time $\tilde {O}(\tfrac{m}{p}+pnm+\tfrac{n}{p^2})=\tilde O(nm^{3/4}+m^{4/3})$, for $p=m^{-1/3}$ or $\tilde O(mn^{1/2}+n^2)$ for $p=n^{-1/2}$.

However, the adaptive sampling procedure does not seem to be well-suited to a dynamic setting. The reason is that the sampling procedure is adaptive in such a way that when the bunches/clusters change, the set of sampled nodes can also change. Although (a dynamic adaptation of) the algorithm can guarantee that the set of sampled nodes is small at any particular moment, there is no guarantee on the \emph{monotonicity} of the set $A$ of sources that we need to maintain distances from. This would add a significant amount of computation necessary to propagate these changes to other parts of the algorithm. A possible solution is to enforce monotonicity by never letting nodes leave $A$. This would require us to bound the \emph{total} number of nodes that would ever be sampled, which seems impossible with the current approach. 
To overcome this, we next suggest a new approach that gives us the monotonicity needed for our dynamic algorithm. For simplicity, we first present this in a static setting.


\paragraph{Introducing bunch overlap thresholds.} Our proposed approach is to divide the nodes into two types depending on the number of bunches they appear in.
In particular, we consider a threshold parameter $\tau$ and we define a node to be a \textit{heavy node} if it is in more than $\tau$ bunches (equivalently its cluster contains more than $ \tau $ nodes), and otherwise we call it a \textit{light node}. In other words, for the set of light nodes we have $C_{\max}\leq \tau$. This threshold introduces a second type of pivot for each node $u$, which we denote by $q(u)$ that is defined to be the closest heavy node to $u$. In particular, instead of computing all distance estimates going through nodes $w \in B(u) \cap B(v)$, we only compute them for the case that $w$ is light, and otherwise it is enough to compute the minimum estimates going through the heavy pivots, i.e.~$\min \{ d(u,q(u))+d(q(u),v), d(v,q(v))+d(q(v),u)\}$. We emphasize that these heavy pivots are \textit{not a subset} of our original pivots $p(v), v \in V$, and they have a different behavior in bounding the running time than the \textit{bunch pivots}.
 As we will see later, introducing these thresholds is crucial in efficiently obtaining the estimates required for Case~\ref{overview:overlap} as we can handle the light and heavy case separately.

\paragraph{Stretch of the algorithm with bunch overlap thresholds.} As discussed above, the stretch analysis was divided into two cases. In the new variant of the algorithm we add a third case, where we look at the distance estimates going through heavy pivots, i.e.~$\min \{ d(u,q(u))+d(q(u),v), d(v,q(v))+d(q(v),u)\}$. The main idea is that now we only need to consider the difficult case, Case~\ref{overview:overlap}, if the relevant node is light, which allows us to implement the algorithm efficiently. 

If the relevant node $w\in B(u)\cap B(v)$ in Case~\ref{overview:adj_bunches} is heavy, we obtain a 2-approximation through the heavy pivot: we get that the $\min\{d(u,w),d(v,w)\} \leq d(u,v)/2$, since $d(u,v)= d(u,w)+d(w,v)$. W.l.o.g.\ say $d(u,w)\leq d(u,v)/2$. Then $d(u,q(u))\leq d(u,w)\leq d(u,v)/2$, hence $d(u,q(u))+d(q(u),v)\leq d(u,q(u))+d(q(u),u)+d(u,v) \leq 2d(u,v)$. 

\paragraph{Maintaining estimates through heavy nodes.}
For the stretch analysis above, we need to maintain a shortest path tree from each heavy node. We ensure a bunch grows at most $O(\log_{1+\epsilon}(nW))$ times. Hence by a total load argument we get that the \emph{total} number of nodes in all bunches over all updates is at most $O(\tfrac{n}{p}\log_{1+\epsilon}(nW))$. This implies there cannot be too many heavy nodes in total, hence we can enforce monotonicity by keeping in $V_{\rm{heavy}}$ all nodes that were once heavy (for details see Lemma~\ref{lm:size_heavy}). As a consequence, we can maintain multi-source approximate distances from all heavy nodes efficiently: we can compute shortest paths from this set to all other nodes in $O(|V_{\rm{heavy}}|m)=O(\tfrac{nm}{p\tau})$ time.

\paragraph{Running time of the algorithm with bunch overlap thresholds.}
Summarizing, we compute bunches and clusters in $O(m/p)$ time, compute distances from $A$ to $V$ in $O(|A|m)=\tilde O(pnm)$ time, and we can compute shortest paths from ththe heavy nodes in $O(|V_{\rm{heavy}}|m)=O(\tfrac{nm}{p\tau})$ time.

Finally, we consider Case~\ref{overview:overlap}. For each node $u\in V$, we consider every light node $w\in B(u)$, and then all nodes $v\in C(W)$. Since $w$ is light, we have $|C(w)|\leq \tau$, so this takes time $\tilde O(n\cdot \tfrac{1}{p}\cdot \tau)$.  

Hence, in the total update time, we obtain $\tilde O((pnm+\tfrac{nm}{p\tau}+n\tau/p)\log^2(nW))$. Setting $p=m^{-1/4}$ and $\tau = m^{1/2}$ gives a total update time of $\tilde O(nm^{3/4}\log(nW))$. We provide the details in \Cref{sc:2_W}.

Again for a decremental algorithm, we need to combine this new idea with the subtleties we had in the previous section: the bunches, their overlaps, and the heavy nodes keep on changing over the updates. We obtain a total update time that matches the stated static running time (up to subpolynomial factors) for the algorithm using bunch overlap thresholds.

\subsubsection{\texorpdfstring{$(2+\epsilon)$}{(2+epsilon)}-APSP for Unweighted Graphs} 
We observe that in unweighted graphs mixed approximation can always be turned into a multiplicative approximation -- at the cost of a blow-up in the number of vertices. More precisely, we prove the following reduction. 
\begin{restatable}{theorem}{thmreduction}\label{thm:reduction}
   Let $\mathcal A$ be an algorithm that provides a $(a+\epsilon,k)$-approximation for APSP in $\tau_{\mathcal A}(n',m')$ time, for any unweighted $n'$-node $m'$-edge graph $G'$, and any constants  $a,k\in \mathbb N_{\geq 1}$ and $\epsilon\in[0,1)$. Given an unweighted graph $G=(V,E)$ on $n$ nodes and $m$ vertices, we can compute $(a+(k+2)\epsilon,0)$-APSP in $\tau_{\mathcal A}(n+km,(k+1)m)$ time. 

    If $\mathcal A$ is a dynamic algorithm, then each update takes $(k+1)\cdot [\text{update time of }\mathcal A]$ time. The query time remains the same (up to a constant factor). 
\end{restatable}
The result is obtained by splitting every edge into $k$ edges, by introducing new nodes on the edge. Now a path corresponding to a $+k$ approximation cannot take a non-trivial detour. 

We combine this result with \Cref{thm:(2_1)approx} to obtain a $(2+\epsilon)$-approximation for unweighted graphs in $\tilde O(m^{7/4})$ total update time. For details, see \Cref{sec:reduction}.

\subsection{\texorpdfstring{$(1+\epsilon, 2(k-1))$}{(1+epsilon,2(k-1))}-APSP for Unweighted Graphs} \label{sec:add_overview}


In this section, we describe our near-additive APSP algorithm.
Our work is inspired by a classic result of Dor, Halperin, and Zwick \cite{DHZ00}, that presented a static algorithm that computes purely additive $+2(k-1)$ approximation for APSP in $\tilde{O}(n^{2-1/k}m^{1/k})$ time in unweighted graphs. Our goal is to obtain similar results dynamically. 
More concretely, we obtain decremental $(1+\epsilon,2(k-1))$-approximate APSP in $O(n^{2-1/k+o(1)}m^{1/k})$ total update time in unweighted graphs.

\subsubsection{\texorpdfstring{$\mathbf{(1+\epsilon,2)}$}{(1+epsilon)}-APSP}
To explain the high-level idea of the algorithm, we first focus on the special case that $k=2$, and that we only want to approximate the distances between pairs of nodes at distance at most $d$ from each other, for a parameter $d$. In this case, we obtain a $+2$-additive approximation in $\tilde{O}(n^{3/2}m^{1/2} d )$ time. This case already allows to present many of the high-level ideas of the algorithm. Later we explain how to extend the results to the more general case.
 We start by describing the static algorithm from \cite{DHZ00}, and then explain the dynamic version. 
The static version of the data structure is as follows:

\begin{itemize}
    \setlength\itemsep{0em}
    \item Let $s_1= (\frac{m}{n})^{1/2}$. Let $V_1$ be the set of \textit{dense nodes}: $V_1:= \{ v \in V: \deg(v) \geq s_1\}$. Also let $E_2$ be the set of \textit{sparse edges}, i.e. edges with at least one endpoint with degree less than $s_1$.
    \item \textbf{Node set $D_1$.} Construct a hitting set $D_1$ of nodes in $V_1$. 
    This means that every node $v \in V_1$ has a neighbor in $D_1$. The size of $D_1$ is $O(n \log{n}/s_1)$.
    \item \textbf{Edge set $E^*$.} Let $E^*$ be a set of size $O(n)$ such that for each $v \in V_1$, there exists $u \in D_1 \cap N(v)$ such that $\{u,v\} \in E^*$.
    \item \textbf{Computing distances from $D_1$.} Store distances $D_1 \times V$ by running a BFS from each $u \in D_1$ on the input graph $G$. 
    \item \textbf{Computing distances from $V \setminus D_1$. } For each $u \in V \setminus D_1$ store a shortest path tree, denoted by $T_u$  rooted at $u$ by running Dijkstra on $(V, E_2 \cup E^* \cup E^{D_1}_u)$, where $E^{D_1}_u$ is the set of weighted edges corresponding to distances in $(\{u\} \times D_1)$ computed in the previous step. 
\end{itemize}

\paragraph{Dynamic data structure.} 
The static algorithm computes distances in two steps. First, it computes distances from $D_1$ by computing BFS trees in the graph $G$. This step can be maintained dynamically by using \emph{Even-Shiloach trees (ES-trees)} \cite{ES81}, a data structure that maintains distances in a decremental graph. Maintaining the distances up to distance $O(d)$ takes $O(|D_1|md)$ time.
The more challenging part is computing distances from $V \setminus D_1$. Here the static algorithm computes distances in the graphs $H_u=(V, E_2 \cup E^* \cup E^{D_1}_u)$. Note that the graphs $H_u$ change dynamically in several ways. First, in the static algorithm, the set $E_2$ is the set of light edges. To keep the correctness of the dynamic algorithm, every time that a node $v$ no longer has a neighbor in $D_1$, we should add all its adjacent edges to $E_2$. Second, when an edge in $E^*$ is deleted from $G$, we should replace it by an alternative edge if such an edge exists. Finally, the weights of the edges in $E^{D_1}_u$ can change over time, when the estimates change. This means that other than deletions of edges, we can also add new edges to the graphs $H_u$ or change the weights of their edges. Because of these edge insertions we can no longer use the standard ES-tree data structure to maintain the distances, because that only works in a decremental setting. 

To overcome this, we use \emph{monotone ES-trees}, a generalization of ES-trees proposed by \cite{HKN2016,HKN14}. In this data structure, to keep the algorithm efficient, when new edges are inserted, the distance estimates do not change. In particular, when edges are inserted, some distances may decrease, however, in such cases the data structure keeps an old larger estimate of the distance. The main challenge in using this data structure is to show that the stretch analysis still holds. In particular, in our case, the stretch analysis of the static algorithm was based on the fact the distances from the nodes $V \setminus D_1$ are computed in the graphs $H_u$, where in the dynamic setting, we do not have this guarantee anymore.  

\paragraph{Stretch analysis.}

The high-level idea of the stretch analysis uses the special structure of the graphs $H_u$. To prove that the distance estimate between $u$ and $v$ is at most $d(u,v)+2$, we distinguish between two cases. First, we maintain the property that as long as $v$ has a neighbor $x \in D_1$, then an edge of the form $\{v,x\}$ is in $E^*$ for $x \in D_1 \cap N(v)$, and since we maintain correctly the distances from $D_1$, we can prove that we get an additive stretch of at most 2 in this case. 

In the second case, all the edges adjacent to $v$ are in $E_2 \subseteq H_u$, and here we can use an inductive argument on the length of the path to prove that we get an additive stretch of at most 2. Note that the estimate that we get can be larger compared to the distance between $u$ and $v$ in the graph $H_u$, but we can still show that the additive approximation is at most 2 as needed. For a detailed stretch analysis, see \Cref{sec:add_warmup_stretch}.

\paragraph{Update time.}

Our update time depends on the size of the graphs $H_u$. While in the static setting it is easy to bound the number of edges in $E_2$ and $E^*$, in the decremental setting, as new edges are added to these sets, we need a more careful analysis. For example, in the static setting $|E^*|=O(n)$, as any node only adds one adjacent edge to the set $E^*$, where in the decremental setting, we may need to add many different edges to this set because the previous ones got deleted. However, we can still prove that even if a node needs to add to $E^*$ edges to all its adjacent neighbors in $D_1$ during the algorithm, the size of $H_u$ is small enough as needed. A detailed analysis of the update time appears in \Cref{sec:warmup_add_time}, where we show that the total update time is $\tilde{O}(m^{1/2}n^{3/2}d)$ in expectation, this matches the static complexity up to the factor $d$.

\paragraph{Handling large distances.} The algorithm described above is efficient if $d$ is small, to obtain an efficient algorithm for the general case, we combine our approach with a decremental near-additive APSP algorithm. More concretely, we use an algorithm from \cite{HKN14} that allows to compute $(1+\epsilon,\beta)$-approximation for APSP in $O(n^{2+o(1)})$ total update time, for $\beta=n^{o(1)}$. Now we set $d=\Theta(\beta/\epsilon)$, and distinguish between two cases. For pairs of nodes at distance at most $d$ from each other, we get a $+2$-additive approximation as discussed above, in $O(m^{1/2}n^{3/2+o(1)})$ time. For pairs of nodes at distance larger than $d$, the near-additive approximation is already a $(1+O(\epsilon))$-approximation by the choice of the parameter $d$. For such pairs, the additive $\beta$ term becomes negligible, as their distance is $\Omega(\beta/\epsilon)$. 
Overall, we get a $(1+\epsilon,2)$-approximation for APSP in $O(m^{1/2}n^{3/2+o(1)})$ time. In fact, if the graph is dense enough ($m=n^{1+\Omega(1)}$), the $n^{o(1)}$ term is replaced by a poly-logarithmic term. For more details see \Cref{sec_summary_additive}. 

\subsubsection{General \texorpdfstring{$k$}{k}}
We can extend the algorithm to obtain a $(1+\epsilon,2(k-1))$-approximate APSP in $O(n^{2-1/k+o(1)}m^{1/k})$ total update time, by adapting the general algorithm of \cite{DHZ00} to the dynamic setting.
At a high-level, instead of having one hitting set $D_1$, we have a series of hitting sets $D_1,D_2,...,D_k$, such that the set $D_i$ hits nodes of degree $s_i$. We compute distances from $D_i$ in appropriate graphs $H_u^i$, that are sparser when the set $D_i$ is larger. Balancing the parameters of the algorithm leads to the desired total update time. For full details and proofs see \Cref{sec:add_general}.

\subsubsection{Limitations of previous approaches} 
We next explain why approaches used in previous decremental algorithms cannot be generalized to obtain our near-additive results. First, the unweighted $(1+\epsilon,2)$-approximate APSP algorithm that takes $\tilde{O}(n^{5/2}/\epsilon)$ time \cite{HKN2016}, is inspired by static algorithms for $+2$-additive emulators, sparse graphs that preserve the distances up to a $+2$-additive stretch. The main idea is to maintain a sparse emulator of size $\tilde{O}(n^{3/2})$, and exploit its sparsity to obtain a fast algorithm. This construction however is specific for the $+2$ case, and cannot be generalized to a general $k$. Note that the update time of the algorithm crucially depends on the size of the emulator, and to obtain a better algorithm for a general $k$, we need to be able to construct a sparser emulator for a general $k$.
However, a beautiful lower bound result by Abboud and Bodwin \cite{abboud20174} shows that for any constant $k$, a purely additive $+k$ emulator should have $\Omega(n^{4/3-\epsilon})$ edges. 
This implies that dynamic algorithms that are based on purely additive emulators seem to require $\Omega(n^{7/3-\epsilon})$ time, and we cannot get running time arbitrarily close to $O(n^2)$. An alternative approach is to build a \emph{near-additive} emulator. This is done in \cite{HKN14}, where the authors show a $(1+\epsilon, 2(1+2/\epsilon)^{k-2})$-approximation algorithm for decremental APSP with expected total update time of $\tilde{O}(n^{2+1/k} (37/\epsilon)^{k-1})$. Here the running time indeed gets closer to $O(n^2)$, but this comes at a price of a much worse additive term of $2(1+2/\epsilon)^{k-2}$. While in the static setting there are also other bounds obtained for near-additive emulators, such as $O(1+\epsilon,O(k/\epsilon)^{k-1})$ emulators of size $O(kn^{1+\frac{1}{2^{k+1}-1}})$ \cite{elkin2004,thorup2006spanners,abboud2018hierarchy,elkin2020near}, in all of the constructions the additive term depends on $\epsilon$, and the dependence on $\epsilon$ is known to be nearly tight for small $k$ by lower bounds from \cite{abboud2018hierarchy}. Hence this approach cannot lead to small constant additive terms that do not depend on $\epsilon$.

\paragraph{Outline.}
The detailed description of our $(2+\epsilon)$-APSP algorithms for weighted graphs can be found in \Cref{sec:TZ_2+eps}. In \Cref{sc:2_W}, we provide the algorithm for $(2+\epsilon,W_{u,v})$-APSP. In \Cref{sec:reduction}, we provide our reduction for purely multiplicative APSP and the result for $(2+\epsilon)$-APSP for unweighted graphs. And finally in \Cref{sc:additive}, we describe our $(1+\epsilon, 2(k-1))$-APSP algorithm for any $k \geq 2$ in unweighted graphs.

\section{Preliminaries}
\subsection{Notation and Terminology}\label{sc:notation}
Throughout this paper, neighbors of a node $v$ are denoted by $N(v)$. We denote by $d_G(u,v)$ the distance between $u$ and $v$ in the graph $G$. If $G$ is clear from the context, we use the notation $d(u,v)$. We denote $W$ for the ratio between the maximum and minimum weight of the graph.  

\paragraph{Dynamic graph algorithms.}
We are interested in designing dynamic graph algorithms for shortest path problems, which allow both for update and query operations. In particular, we look at \emph{decremental} algorithms, which allow for weight increases and edge deletions. This is opposed by \emph{fully dynamic} algorithms, which also allow for weight decreases and edge insertions. To be concrete, we allow for the following three operations:
\begin{itemize}
    \item \Delete{$u,v$}: delete the edge $\{u,v\}$ from the graph;
    \item \Update{$u,v,\Delta$}: increase the weight of an edge $\{u,v\}$ to $\Delta$, i.e., $w(u,v)\leftarrow \Delta $;
    \item \Distance{$u,v$}: return a distance approximation $\hat d_G(u,v)$ between the nodes $u$ and $v$ in the current graph $G$.
\end{itemize}

We assume that we have an \emph{oblivious adversary}, i.e., that the sequence of updates is fixed from the start. This is as opposed to the \emph{adaptive adversary}, which can use the random choices made by the algorithm to determine the next update. 

Edge weights can be positive integers bounded by $W$, for some parameter $W\in \mathbb N$. We say that $\Distance(u,v)$ returns \emph{$(\alpha,\beta)$-approximate distances}, denoted by $\hat d_G$, if $d_G(u,v) \leq \hat d_G(u,v) \leq \alpha d_G(u,v)+\beta$ for all $u,v\in V$. Moreover we denote $(\alpha,0)$-approximate by $\alpha$-approximate. 

In all our algorithms, we maintain a distance oracle with constant \emph{query time}, i.e., the time needed to answer a distance query $\Distance{u,v}$.
Further, we give the \emph{total update time}, which is the time needed to process a sequence of up to $m$ edge deletions or weight increases, where $m$ is the number of edges before the first deletion. 

For any object, value, or data structure $x$ that changes over time, we use a superscript $x^{(t)}$ to specify its state at time $t$, i.e., $x^{(t)}$ is $x$ after processing $t$ updates.

\subsection{Decremental Tools}\label{sc:decr_tools}
We start by describing several tools from previous work that we need to use in our new decremental algorithm.
\paragraph{Cluster and bunch maintenance.} 
Recall Thorup and Zwick~\cite{TZ2005} clusters as follows: Given a parameter $p\in [0,1]$, called the \emph{cluster sampling rate}, we let $A$ be a set of sampled nodes, where each node is sampled uniformly at random with probability $p$. Then we define the \emph{pivot} of a node $v$ to be the closest node in $A$, i.e., $p(v)=\argmin_{s\in A} d(s,v)$. 

The following lemma is implicitly proven in \cite{LN2020}, but for completeness we sketch a proof in Appendix~\ref{ap:bunches_and_clusters}.

\begin{restatable}{lemma}{lmbunch} \label{lm:approx_bunches_on_G}
Given a decremental graph $G=(V,E,w)$, and parameters $0< \epsilon <1$, $0<\rho<1/2$ and the sampling probability $0< p <1$ we can maintain for each $v\in V$ an approximate pivot $\tilde p(v)$, such that $d_G(v,p(v)) \leq \dLNbunch(v,\tilde{p}(v)) \leq (1+\epsilon) d_G(v,p(v))$.
Moreover, we maintain approximate bunches $\tilde{B}(v)$ that contains
$\hat{B}(v):= \{ w \in V : (1+\epsilon)d_G(v,w) < d_G(v, p(v)) \}$ and we have monotone (non-decreasing over updates) distance estimates $\dLNbunch$ such that for any $u \in \tilde{B}(v)$ we have $d_G(u,v) \leq \dLNbunch(u,v) \leq (1+\epsilon) d_G(u,v)$, and with high probability, we have $|\tilde{B}(v)| = O(\frac{\log n\log nW}{p})$. The algorithm has total update time $\tilde O(\tfrac{mn^\rho}{p}) \cdot O(\frac{1}{\rho \epsilon})^{1+2/\rho}$.
\end{restatable}

In our algorithms for balancing out different terms in the running time we set the parameter~$\rho$ depending on the graph density so that $m= n^{1+\rho}$. Hence for graphs with $m= n^{1+\Omega(1)}$ edges, we can set $\rho$ to be a constant and for sparser graphs we can set $\rho = \frac{\sqrt{\log n}}{\log \log n}$. For simplicity of our statements in the rest of this section we assume $0< \epsilon <1$ is a constant, and thus we hide $(1/\epsilon)^c$ factors in $\tilde{O}$ when $c$ is constant, or in $n^{o(1)}$ when $c=\omega(1)$.

\paragraph{Multi-source shortest paths.}
We will use the following lemma (from \cite{LN2020}) repeatedly to maintain $(1+\epsilon)$ approximate shortest paths from a given set of sources. 

\begin{lemma}[\cite{LN2020}]\label{lm:decr_MSSP}
    There is a data structure which given a weighted undirected graph $G = (V, E)$ explicitly maintains $(1+\epsilon)$-approximate distances from a set of sources $S$ in $G$ under edge deletions, where $0<\epsilon< 1/2$ is a constant. After each update, the algorithm returns all nodes for which the distance estimate has changed. 
    During the execution of the algorithm, nodes can be added to $S$. We write $s$ for the final size of $S$. 
    Assuming that $|E| = n^{1+\Omega(1)}$ and $s = n^{\Omega(1)}$, the total update time is w.h.p.~$\tilde O(sm)$. If $|E| = n^{1+o(1)}$ or $s= n^{o(1)}$, then the total update time is $\tilde O(smn^{o(1)})$. 
\end{lemma}

\section{\texorpdfstring{$(2+\epsilon)$}{(2+epsilon)}-Approximate APSP for Weighted Graphs}\label{sec:TZ_2+eps}
Given an undirected, weighted graph $G=(V,E,w)$, our goal is to maintain a $(2+\epsilon)$-approximate all-pair distances in $G$, denoted by $\hat{d}$. 
We provide an algorithm \DynAPSPtwo with parameters $ 0< p,\epsilon\leq 1 $, and $0< \rho \leq 1$ that gives a $(2+\epsilon)$-approximation of APSP for weighted graphs, with ratio between the maximum and minimum weight $W$. Overall, the total update time is $\tilde{O}((pnm + \tfrac{n^2}{p}) \log^2 (nW))$, for an exact statement see Lemma~\ref{lm:mult_time}. The algorithm has constant query time. We provide a high-level algorithm in \Cref{sc:mult_alg}. Then we give further details and parametrized running time analysis of each step in \Cref{sc:mult_time} and the stretch analysis in \Cref{sc:mult_stretch}. In \Cref{sc:mult_result}, we combine these and set the parameter $p$ appropriately to obtain our final result.

\subsection{Decremental Data Structures}\label{sc:mult_alg}

We will make use of the $(1+\epsilon_2)$-approximate bunches and clusters that are introduced in \Cref{sc:decr_tools}. However, for efficiency purposes, we do not use these bunches themselves, but instead look at bunches with some additional approximation. The goal is to ensure that the bunches do not increase too often. 

Let $A$ be a set of sampled nodes with a probability $p$. Recall that, for each node $v$, the \emph{pivot} $p(v)$ is the closest node in $A$ to $v$.

Lemma~\ref{lm:approx_bunches_on_G} gives us an approximate pivot $\tilde p(v)$, together with a distance estimate $\dLNbunch$, such that $d_G(v,p(v)) \leq \dLNbunch(v,\tilde{p}(v)) \leq (1+\epsilon_2) d_G(v,p(v))$. Further, we have $(1+\epsilon_2)$-approximate bunches: $\tilde B(v) := \{w\in V : (1+\epsilon_2)d(v,w) < \dLNbunch(v,\tilde p(v))\}$, and distance estimate $\dLNbunch$ such that $d(u,v)\leq \dLNbunch(u,v)\leq (1+\epsilon_2)d(u,v)$ for all $u\in \tilde B(v)$.


We now define our final approximate bunch, which we denote by $B(v)$\footnote{In the high-level overview $B(v)$ has been used for the exact bunches. Since exact bunches do not occur in our actual algorithm, $B(v)$ will denote these approximate bunches from here on.}.
\begin{definition}\label{def:app_app_bunch}
We need the following notions related to bunches and clusters:
\begin{itemize}
    \item For all $v \in V$, we keep track of a $(1+\epsilon_3)$-approximate \emph{bunch radius} $r(v)$, which is a $(1+\epsilon_3)$-approximation of the approximate distance to the approximate pivot, i.e., of $\dLNbunch(v,\tilde{p}(v))$. We denote by $r^{(t)}$ the bunch radius at time $t$. Initially, we set $r^{(0)}(v):=\dLNbunch(v,\tilde{p}(v))$. We update $r^{(t)}(v)$ whenever $\dLNbunch(v,\tilde{p}(v)) > (1+\epsilon_3)r^{(t-1)}(v)$, if so then we set $r^{(t)}(v):=\dLNbunch(v,\tilde{p}(v))$.
    \item Initially we set $B^{(0)}(v):= \tilde B^{(0)}(v)$ and whenever $r^{(t)}(v)\neq r^{(t-1)}(v)$ we set $B^{(t)}(v):=\tilde B^{(t)}(v)$.
    \item The \textit{cluster} of a node $u\in V$ is the set $C(u):= \{ v \in V: u \in B(v) \}$.
\end{itemize}    
\end{definition}

 Note that these bunches $B(v)$ increase at most $\log_{1+\epsilon_3}(nW)$ times by construction.
Next, let us investigate the sizes of these sets. 

\begin{lemma}\label{lm:size_heavy}
 At any given time
\begin{enumerate}[(i)]
    \item With high probability, $ A $ has size $\tilde O (p n) $; \label{item:size_A}
    \item With high probability, each bunch $B(v)$ has size $\tilde O (\tfrac{1}{p}) $;  \label{item:size_bunch}
\end{enumerate}
\end{lemma}
\begin{proof}
\begin{enumerate}[(i)]
    \item This follows immediately from the definition of $A$. 
    \item This holds trivially as $B(v)\subseteq \tilde B(v)$, and $|\tilde{B}(v)|=\tilde O(\tfrac{1}{p})$ at any moment in time due to Lemma~\ref{lm:approx_bunches_on_G}.\qedhere
\end{enumerate}
\end{proof}

Now, we give a description of the algorithm, which is a dynamic version of the static algorithm given in the high-level overview, \Cref{sc:high-level_2approx}. Next, we give a proof of correctness and show how to implement each of the steps. We provide a separate initialization procedure, as for some steps it will turn out that the initialization step and the update steps are quite different.\\
Throughout the algorithm, we set $\tilde\delta := (1+\epsilon_4)^{\lceil \log_{1+\epsilon_4}(\delta)\rceil}$ whenever we want a distance estimate $\delta$ to update at most $\log_{1+\epsilon_4}(nW)$ times. Note that we can always maintain such a $\tilde \delta$ with only a constant multiplicative factor in the total update time. Moreover, it is clear that $\delta \leq \tilde \delta \leq (1+\epsilon_4)\delta$.
Below we define the functions \Initialization{}, and functions used after each edge update $\{u,v\}\in E$: \Delete{$u,v$}, \Update{$u,v, w(u,v)$}. Finally we have the \Query{} function that returns a final estimate over different data structures maintained. 
\\ 

\noindent
\textbf{Algorithm} \DynAPSPtwo{$V,E,w,p,\epsilon$}\\
\Initialization{}:

Set $\epsilon_1=\epsilon_2=\epsilon_3=\epsilon_4=\epsilon/3$.
\begin{enumerate}
    \item \textbf{Sampling.} Sample each node with probability $p$ to construct $A$. \label{step:sampling}
    \item \textbf{Distances from $A$.} Compute $(1+\epsilon_1)$-approximate distances $\delta_A(s,v)$ for $s\in A$ and $v\in V$. 
    \item \textbf{Bunches and clusters.} Initialize the algorithm for maintaining the approximate bunches $\tilde B$ and pivots $\tilde p$ according to Lemma~\ref{lm:approx_bunches_on_G}, denoting with $\dLNbunch(v,w)$ the $(1+\epsilon_2)$-approximate distances for $w\in \tilde B(v)$.\\
    Initialize $(1+\epsilon_3)$-approximate bunch radii $r(v)$ with respect to these approximate pivots $\tilde p(v)$.  \label{step:A_MSSP}
    We initialize another approximation to the bunches, denoted by $B(v)$, which will be defined with respect to the bunch radius $r(v)$. Initially, we simply set $B(v)=\tilde{B}(v)$. We let $\delta_B(u,v):=\dLNbunch(u,v)$, and set $\tilde \delta_B(u,v) := (1+\epsilon_4)^{\lceil \log_{1+\epsilon_4}(\delta_B(u,v))\rceil}$. \label{step:bunches}
    \item \textbf{Neighboring bunch data structure.} For $x\in V$ and $v\in V$ such that there exists $y\in N(x)\cap B(v)$ initialize a Min-Heap $Q^{\rm{nbr-bunch}}_{x,v}$, with an entry for \emph{each} $y\in N(x)\cap B(v)$. The key for the Min-Heap is the distance estimate $ \tilde w (x, y) + \tilde\delta_B (y, v) $. Moreover, to make sure the minimal entry, denoted by $\delta^{\rm{nbr-bunch}}(x,v)$, does not change too often, we maintain an additional approximation: $\tilde \delta^{\rm{nbr-bunch}}(x,v):= (1+\epsilon_4)^{\lceil \log_{1+\epsilon_4}(\delta^{\rm{nbr-bunch}}(x,v))\rceil}$. \label{step:init_nbr_bunch}
    \item \textbf{Adjacent data structure.} For $u,v\in V$ initialize a Min-Heap $Q^{\rm{adjacent}}_{u,v}$, with an entry for each $\{u',v'\}\in E$ such that $u'\in B(u)$ and $v' \in B(v)$. The key for the Min-Heap is the distance estimate $\tilde \delta_B(u,u')+\tilde \delta^{\rm{nbr-bunch}}(u',v)$. \label{step:min_heap_nbr} 
\end{enumerate}

\noindent\Delete{$u,v$}:\\
\indent \Update{$u,v,\infty$}\\

\noindent\Update{$u,v,w(u,v)$}:
\begin{enumerate}[(1)]
     \item \textbf{Distances from $A$.} For all $s\in A$, run \Update{$u,v,w(u,v)$} on the SSSP approximation denoted by $\delta_A(s,\cdot)$. 
    \item \textbf{Update bunches and heaps.} Maintain the bunches and clusters, and update the heaps accordingly. See below for the details.  \label{step:update3}
\end{enumerate}

 \noindent\Query{$u,v$}:\\
    Output $\hat d(u, v)$ to be the minimum of 
    \begin{enumerate}[(a)]\setcounter{enumi}{1}
         \item $\min\{ \dLNbunch(u,\tilde p(u))+\delta_A(v,\tilde p(u)),\delta_A(u,\tilde p(v))+\dLNbunch(v,\tilde p(v))\}$;\label{step:pivot}
        \item The minimum entry of $Q^{\rm{adjacent}}_{u,v}$, which gives $\min\{\tilde \delta_B(u,u')+\tilde \delta^{\rm{nbr-bunch}}(u',v): u'\in B(u)
        \}$. \label{step:nbr_light1}
    \end{enumerate}

\paragraph{Update bunches and heaps (Update Step~\ref{step:update3}).}
Now let us look more closely on how we maintain the bunches and clusters, and how these effect the heaps. We maintain the approximate bunches $\tilde B(v)$ for all $v\in V$ by Lemma~\ref{lm:approx_bunches_on_G}. From here maintain the additional approximate bunches $B(v)$ for all $v\in V$ as given in Definition~\ref{def:app_app_bunch}, i.e., 
If there is a change to $\dLNbunch(v',\tilde p(v'))$, do the following two things:\label{step:update_pivots}
    \begin{itemize}
        \item We check if $\dLNbunch^{(t)}(v',\tilde p^{(t)}(v'))> (1+\epsilon_3)r^{(t-1)}(v')$, and if so we set $r^{(t)}(v')=\dLNbunch^{(t)}(v',p^{(t)}(v'))$ and we set $B^{(t)}(v):=\tilde B^{(t)}(v)$. If not, we set $r^{(t)}(v')=r^{(t-1)}(v')$ and $B^{(t)}(v)=B^{(t-1)}(v)$.
    \end{itemize}
 Let $u'\in V$ be a node for which there is a change to the bunch involving a node $w\in V$. Such a change can be: \ref{case:smaller bunch} $w$ can leave the bunch of $u'$ \ref{case:distance in bunch changed} the distance estimate of $w$ to $u'$ can change while $w$ stays within the bunch, or \ref{case:bigger bunch} $w$ can join the bunch of $u'$. We formalize this as follows, using $\Delta$ to denote which case we are in.
    \begin{enumerate}[(i)]
        \item $w\in B^{(t-1)}(u')$ and $w\notin B^{(t)}(u')$, we set $\Delta=\infty$ to indicate that $w$ is no longer in $B(u')$. \label{case:smaller bunch}
        \item $w\in B^{(t-1)}(u')$ and $w\in B^{(t)}(u')$ with an increased distance estimate $\tilde \delta_B^{(t)}(u',w)$, we set $\Delta=\tilde \delta_B^{(t)}(u',w)$ to be the new distance. \label{case:distance in bunch changed}
        \item \label{case:bigger bunch} $w\notin B^{(t-1)}(u')$ and $w\in B^{(t)}(u')$, we set $\Delta=\tilde \delta_B^{(t)}(u',w)$ to be the distance. 
    \end{enumerate}

    Next, we do the following. 
    \begin{itemize}
        \item \UpdateNbrBunchDSBunchChange{$w,u', \infty$} to update the corresponding entries $\tilde w(v',w)+\tilde\delta_B(w,u')$ in the data structures $Q^{\rm{nbr-bunch}}_{v',u'}$ for every $v'\in \text{Set}^{\rm{nbr-bunch}}_{w,u'}$. If this changes the minimal entry $\delta^{\rm{nbr-bunch}}(x,y)$ in $Q^{\rm{nbr-bunch}}_{x,y}$ for some $x,y\in V$, we check whether we need to update $\tilde \delta^{\rm{nbr-bunch}}(x,y)$, if so we run \UpdateAdjacentDSMinHeapChange{$x,y,\tilde \delta^{\rm{nbr-bunch}}(x,y)$}. 
        \item \UpdateAdjacentDS{$w,u',\Delta$} to update the corresponding entries $\tilde\delta_B(u',w)+\tilde\delta^{\rm{nbr-bunch}}(w,v')$ in the data structures $Q^{\rm{adjacent}}_{u',v'}$ for every $v'\in \text{Set}^{\rm{adj-bunch}}_{u',w}$.
    \end{itemize}

\subsection{Further Details and Time Analysis}\label{sc:mult_time}
In this section, we add further details of the procedures described and analyze the total update time. 

\paragraph{Sampling.} Note that $A$ remains unchanged throughout the algorithm, hence it is only sampled once at the initialization. This takes $O(n)$ time. 

\paragraph{Distances from $A$.} We use the decremental shortest path algorithm from Lemma~\ref{lm:decr_MSSP} together with Lemma~\ref{lm:size_heavy}, which bounds the size of $A$, to obtain the following lemma. 

\begin{lemma}\label{lm:A_MSSP}
    With high probability, we can maintain $(1+\epsilon_2)$-approximate SSSP from each node in $A$, providing us with $\delta_A(s,v)$ for $s\in A$ and $v\in V$. We can do this in time $\tilde{O}(pnm)$ if $m=n^{1+\Omega(1)}$ and $pn=n^{\Omega(1)}$, and in time $\tilde O(pn^{1+o(1)}m)$ otherwise.
\end{lemma}

\paragraph{Bunches and clusters.}
We have approximate bunches $\tilde B(v)$ and pivots $\tilde p(v)$ by Lemma~\ref{lm:approx_bunches_on_G}, now we maintain the bunch radii $r(v)$ and approximate bunches $B(v)$ and clusters $C(v)$ according to Definition~\ref{def:app_app_bunch}. 
As before, we simply set $\tilde \delta_B(u,v):= (1+\epsilon_4)^{\lceil\log_{1+\epsilon_4}(\delta_B(u,v))\rceil}$. It is immediate that these extra steps do not incur more than a constant factor in the running time with respect to Lemma~\ref{lm:approx_bunches_on_G}, so we obtain the following result. 

\begin{lemma}\label{lm:updatebunches}
    We can maintain $(1+\epsilon_3)$-approximate bunch radii $r(v)$ with respect to the approximate pivots $\tilde p$.
    With high probability, we can maintain approximate bunches $B(v)$ and clusters $C(v)$, and distance estimates $\delta_B(v,u)$ for $u\in B(v)$, for all nodes $v\in V$. Moreover, we can maintain $\tilde \delta_B(u,v)$, a $(1+\epsilon_4)$-approximation of $\delta_B(u,v)$ that updates at most $O(\log_{1+\epsilon_4}(nW))$ times. This takes total update time $ \tilde O (\tfrac{mn^\rho}{p}) $, for some arbitrarily small constant $\rho$, if $m=n^{1+\Omega(1)}$, and $\tilde O(\tfrac{mn^{o(1)}}{p})$ otherwise. 
\end{lemma}

\paragraph{Neighboring bunch data structure.} 
To maintain $Q^{\rm{adjacent}}_{u,v}$, we maintain an auxiliary data structure, which is another Min-Heap $Q^{\rm{nbr-bunch}}_{x,v}$. This data structure consists of the following ingredients. We say $B(v)$ is a \emph{neighboring (approximate) bunch} of $x$ if there exists an edge from $x$ to $B(v)$. For every node $ x $, and every neighboring bunch $ B (v) $ our goal is to store  all neighbors $ y $ such that $ y \in B (v) \cap N(x) $ together with $ \tilde w (x, y) + \tilde\delta_B(y, v) $, which allows us to find the neighbor~$ y $ minimizing this sum for each neighboring bunch\footnote{Recall that we denote $\tilde w$ for $(1+\epsilon_4)^{\lceil\log_{1+\epsilon_4}{w}\rceil}$.}. Later, we will combine this with the distance estimate $\tilde\delta_B(u,x)$ for $u\in B(x)$ to get a distance estimate from $u$ to $v$, which will be stored in $Q^{\rm{adjacent}}_{u,v}$. To make sure we do not have to update the entries in $Q^{\rm{adjacent}}_{u,v}$, we do not maintain the minimum of $Q^{\rm{nbr-bunch}}_{x,v}$ exactly, but instead we maintain an additional approximation, denoted $\tilde\delta^{\rm{nbr-bunch}}(x,v)$ that only changes at most a polylogarithmic number of times.

\begin{lemma}\label{lm:NbrBunchDS}
    With high probability, we can maintain the data structures $Q^{\rm{nbr-bunch}}_{x,v}$ for all $x,v\in V$ in total update time $ \tilde O (\tfrac{n^2}{p}\cdot (\log_{1+\epsilon_4}(nW))+\log_{1+\epsilon_3}(nW))) $. Moreover, in the same total update time, we can maintain an approximation $\tilde\delta^{\rm{nbr-bunch}}(x,v)$ of the minimum entry of $Q^{\rm{nbr-bunch}}_{x,v}$ that only changes at most $\log_{1+\epsilon_3}(nW)\log_{1+\epsilon_4}(nW)$ times for each pair $(x,v)$.
\end{lemma}
\begin{proof}
    This data structure is initialized as follows:
    \begin{itemize}
        \item Iterate over every vertex $v\in V$, and over all $y\in B(v)$. Then iterate over all neighbors $x\in N(y)$ of $y$ and add an entry to the data structure $Q^{\rm{nbr-bunch}}_{x,v}$ with key $\tilde{w}(x,y) + \tilde \delta_B(y,v)$.
    
        This takes $ \tilde O (\tfrac{n^2}{p}) $ initialization time, since there are $n$ options for $v$, it has bunch size $\tilde O(\tfrac{1}{p})$, so $\tilde O(\tfrac{1}{p})$ options for $y$, which has at most $n$ neighbors. 
    \end{itemize}
    We need to facilitate two types of updates for an entry of $Q^{\rm{nbr-bunch}}_{x,v}$ corresponding to an edge $\{x,y\}$: increases in the weight $\tilde w(x,y)$ and increases in the distance estimates $\tilde\delta_B(y,v)$. We define the functions \UpdateNbrBunchDSEdgeChange and \UpdateNbrBunchDSBunchChange respectively. 
    
    Next, we consider how to update this data structure. Suppose a value $\tilde\delta_B(y,v)$ changed for some $y\in B(v)$ due to an update to the bunches (Lemma~\ref{lm:updatebunches}), then we would need to change the corresponding entries in $Q^{\rm{nbr-bunch}}_{x,v}$. Lemma~\ref{lm:updatebunches} gives us $y$, $v$, and the new distance $\tilde \delta_B(y,v)$, which we need to update in all data structures $Q^{\rm{nbr-bunch}}_{x,v}$ for which $\{x,y\}\in E$. To do this efficiently, we maintain a set $\text{Set}_{x,y}^{\rm{nbr-bunch-edge}}$ for each edge $\{x,y\}\in E$ that stores for which $v\in V$ there is an entry in $Q^{\rm{nbr-bunch}}_{x,v}$ and a pointer to where in that data structure it appears. We equip this set with the structure of a self-balancing binary search tree, so we have polylogarithmic insertion and removal times. We update this every time the appearance or location in $Q^{\rm{nbr-bunch}}_{x,v}$ changes. Now if $\tilde w(x,y)$ changes due to an update, we can efficiently change the corresponding places in $Q^{\rm{nbr-bunch}}_{x,v}$ for all relevant $v\in V$.
    
    Moreover, for the changes in the distances from the bunches, we maintain a set $\text{Set}_{y,v}^{\rm{nbr-bunch}}$ for each $v\in V, y\in B(v)$ that stores for which $x\in V$ there is an entry in $Q^{\rm{nbr-bunch}}_{x,v}$ and a pointer to where in that data structures it appears. We equip this set with the structure of a self-balancing binary search tree, so we have polylogarithmic insertion and removal times. We update this every time the appearance or location in $Q^{\rm{nbr-bunch}}_{x,v}$ changes. Now if $\tilde\delta_B(y,v)$ changes due to an update, we can efficiently change the corresponding places in $Q^{\rm{nbr-bunch}}_{x,v}$ for all relevant $x\in V$. 
    
    We define the function \UpdateNbrBunchDSEdgeChange{$x,y,\Delta$}, where $\{x,y\}\in E$, and $\Delta$ is the new value $\tilde w^{(t)}(x,y)$, or $\infty$ if the edge $\{x,y\}$ is deleted. Depending on $\Delta$, we do one of the following:
    \begin{enumerate}[(i)]
        \item If $\Delta<\infty$, then we update the corresponding estimates:
        \begin{itemize}
            \item For each entry $v\in \text{Set}_{x,y}^{\rm{nbr-bunch-edge}}$, update the corresponding entry $\tilde w(x,y)+\tilde\delta_B(y,v)$ in $Q^{\rm{nbr-bunch}}_{x,v}$ and update the entry in $\text{Set}_{x,y}^{\rm{nbr-bunch-edge}}$ accordingly.
        \end{itemize}
        \item If $\Delta=\infty$, then we remove the corresponding estimates:
        \begin{itemize}
            \item For each entry $v\in \text{Set}_{x,v}^{\rm{nbr-bunch-edge}}$, remove the corresponding entry in $Q^{\rm{nbr-bunch}}_{x,v}$. 
        \end{itemize}
    \end{enumerate}
    Concerning the total update time of \UpdateNbrBunchDSEdgeChange, we recall that $\tilde w(x,y) = (1+\epsilon_4)^{\lceil \log_{1+\epsilon_4}(w(x,y))\rceil}$. This means that each weight $\tilde w(x,y)$ gets updated at most $\log_{1+\epsilon_4}(W)$ times, giving total update time $\tilde O(\tfrac{n^2}{p}\log_{1+\epsilon_4}(W)$.
    
    We define the function \UpdateNbrBunchDSBunchChange{$y,v,\Delta$}, where $y$ is an element of $B^{(t-1)}(v)$ and/or $B^{(t)}(v)$, and $\Delta$ is the new value $\tilde\delta_B^{(t)}(y,v)$, or $\infty$ if $y$ is no longer in $B(v)$. \UpdateNbrBunchDSBunchChange does one of three things:
    \begin{enumerate}[(i)]
        \item If $\Delta<\infty$, $y\notin B^{(t-1)}(v)$, and $y\in B^{(t)}(v)$, then we add the corresponding new estimates:\label{step:new1}
        \begin{itemize}
            \item We initialize $\text{Set}_{y,v}^{\rm{nbr-bunch}}$ to be empty.
            \item For all neighbors $x\in N(y)$ of $y$, we add $\tilde w(x,y)+\tilde\delta_B(y,v)$ to $Q^{\rm{nbr-bunch}}_{x,v}$ and store $x$ together with its location in $Q^{\rm{nbr-bunch}}_{x,v}$ in $\text{Set}_{y,v}^{\rm{nbr-bunch}}$.
        \end{itemize}
        \item If $\Delta<\infty$, $y\in B^{(t-1)}(v)$, and $y\in B^{(t)}(v)$, then we update the corresponding estimates:\label{step:existing1}
        \begin{itemize}
            \item For each entry $x\in \text{Set}_{y,v}^{\rm{nbr-bunch}}$, update the corresponding entry $\tilde w(x,y)+\tilde\delta_B(y,v)$ in $Q^{\rm{nbr-bunch}}_{x,v}$ and update the entry in $\text{Set}_{y,v}^{\rm{nbr-bunch}}$ accordingly. 
        \end{itemize}
        \item If $\Delta=\infty$, $y\in B^{(t-1)}(v)$, and $y\notin B^{(t)}(v)$, then we remove the corresponding estimates:\label{step:del1}
        \begin{itemize}
            \item For each entry $x\in \text{Set}_{y,v}^{\rm{nbr-bunch}}$, remove the corresponding entries from $Q^{\rm{nbr-bunch}}_{x,v}$ and from $\text{Set}_{y,v}^{\rm{nbr-bunch}}$. 
        \end{itemize}
    \end{enumerate}
    Moreover, if one of these three actions change the minimum of $Q^{\rm{nbr-bunch}}_{x,v}$, this has an impact on the data structures $Q^{\rm{adjacent}}_{\cdot,v}$. We run \UpdateAdjacentDSMinHeapChange{$x,v,\Delta$}, which is defined and analyzed in Lemma~\ref{lm:adjacentDS}, to propagate this change. Here $\Delta$ is the minimum of $Q^{\rm{nbr-bunch}}_{x,v}$. The total update time of \UpdateAdjacentDSMinHeapChange is included in Lemma~\ref{lm:adjacentDS}, and not in this lemma.

    As $\tilde\delta_B(y,v)$ can change at most $O(\log_{1+\epsilon_4}(nW))+\log_{1+\epsilon_3}(nW))$ times, and for each time we need at most $\tilde O(\deg(y))=O(n)$ time to update the data structure, we get at most $\tilde O(\deg(y)(\log_{1+\epsilon_4}(nW))+\log_{1+\epsilon_3}(nW)))$ total update time for the pair $(y,v)$. Now note that there are at most $\tilde O(\tfrac{n}{p})$ such pairs, since there are $n$ options for $v$, with $\tilde O(\tfrac{1}{p})$ options for $y\in B(v)$. So we get total update time  $ \tilde O (\tfrac{n^2}{p} \cdot (\log_{1+\epsilon_4}(nW))+\log_{1+\epsilon_3}(nW))) $.
     
    Further, we maintain $\tilde\delta^{\rm{nbr-bunch}}(x,v)$, a $(1+\epsilon_4)$-approximation of the minimal entry $\delta^{\rm{nbr-bunch}}(x,v)$ of $Q^{\rm{nbr-bunch}}_{x,v}$, defined by $\tilde\delta^{\rm{nbr-bunch}}(x,v):= (1+\epsilon_4)^{\lceil\log_{1+\epsilon_4}(\delta^{\rm{nbr-bunch}}(x,v))\rceil}$. 
    Note that $\tilde\delta^{\rm{nbr-bunch}}(x,v)$ is \emph{not} monotone over time, so it is not trivial how often the estimate changes. However, note that $\tilde\delta^{\rm{nbr-bunch}}(x,v)$ increases at most $\log_{1+\epsilon_4}(nW)$ times in a row. It can only decrease $\log_{1+\epsilon_3}(nW)$ times, since the values $\tilde w(x,y)+\tilde\delta_B(y,v)$ are non-decreasing, except if new vertices $y\in B^{(t)}(v)$ are considered. This is only the case when the bunch $B^{(t)}(v)$ grows, which happens at most $\log_{1+\epsilon_3}(nW)$ times by Lemma~\ref{lm:updatebunches}.
\end{proof}


\paragraph{Adjacent data structure.}
Using $Q^{\rm{nbr-bunch}}_{x,v}$, we construct the required Min-Heap $Q^{\rm{adjacent}}_{u,v}$ for each pair $u,v\in V$. For every pair of nodes $ u, v \in V $ our goal is to store all edges $ \{x, y\} $ such that $ x \in B (u) $ and $ y \in B (v) $, together with the distance $\tilde\delta_B(u,x)+\tilde w(x,y)+\tilde\delta_B(y,v)$.

\begin{lemma}\label{lm:adjacentDS}
    With high probability, we can maintain the data structures $Q^{\rm{adjacent}}_{u,v}$ for all $u,v\in V$ in total update time $ \tilde O (\tfrac{n^2}{p}\log_{1+\epsilon_4}(nW)\log_{1+\epsilon_3}(nW)) $.
\end{lemma}
\begin{proof}

    This data structure is initialized as follows:    
    \begin{itemize}
        \item Iterate over all nodes $ u $, all $ x \in B (u) $, and all neighboring bunches $ B (v) $ of $ x $ and add an entry for $x$ to $Q^{\rm{adjacent}}_{u,v}$ with key $\tilde\delta_B(u,x)+\tilde\delta^{\rm{nbr-bunch}}(x,v)$.
    
    As there are at most $ n $ bunches in total, this takes time $ \tilde O (n \cdot \tfrac{1}{p} \cdot n) = \tilde O (\tfrac{n^2}{ p}) $. 
    \end{itemize}
    Again, we maintain additional data structures to facilitate updates: for each $x,v\in V$, we maintain $\text{Set}_{x,v}^{\rm{adjacent}}$ which shows for which $u\in V$ there is an entry in $Q^{\rm{adjacent}}_{u,v}$ and a pointer to where in that data structure it appears. We update this every time the appearance or location in $Q^{\rm{adjacent}}_{u,v}$ changes. Now if $\tilde\delta_B(u,w)$ changes due to an update, we can efficiently change the corresponding places in $Q^{\rm{adjacent}}_{u,v}$. 
    
    We define the update function \UpdateAdjacentDSBunchChange analogously to the function \UpdateNbrBunchDSBunchChange. \UpdateAdjacentDSBunchChange{$u,x,\Delta$} does one of the following three things:
    \begin{enumerate}[(i)]
        \item If $\Delta<\infty$, $x\notin B^{(t-1)}(u)$, and $x\in B^{(t)}(u)$, then we add the corresponding new estimates:
        \begin{itemize}
            \item We initialize $\text{Set}_{u,x}^{\rm{adj-bunch}}$ to be empty.
            \item For all nodes with non-empty $\tilde\delta^{\rm{nbr-bunch}}(x,v)$, we add $\tilde\delta_B(u,x)+\tilde\delta^{\rm{nbr-bunch}}(x,v)$ to $Q^{\rm{adjacent}}_{u,v}$ and store $v$ together with its location in $Q^{\rm{adjacent}}_{u,v}$ in $\text{Set}_{u,x}^{\rm{adj-bunch}}$.
        \end{itemize}
        \item If $\Delta<\infty$, $x\in B^{(t-1)}(u)$, and $x\in B^{(t)}(u)$, then we update the corresponding estimates:
        \begin{itemize}
            \item For each entry $v\in \text{Set}_{u,x}^{\rm{adj-bunch}}$, update the corresponding entry $\tilde\delta_B(u,x)+\tilde\delta^{\rm{nbr-bunch}}(x,v)$ in $Q^{\rm{adjacent}}_{u,v}$ and update the entry in $\text{Set}_{u,x}^{\rm{adj-bunch}}$ accordingly. 
        \end{itemize}
        \item If $\Delta=\infty$, $x\in B^{(t-1)}(u)$, and $x\notin B^{(t)}(u)$, then we remove the corresponding estimates:
        \begin{itemize}
            \item For each entry $v\in \text{Set}_{u,x}^{\rm{adj-bunch}}$, remove the corresponding entries from $Q^{\rm{adjacent}}_{u,v}$ and from $\text{Set}_{u,x}^{\rm{adj-bunch}}$. 
        \end{itemize}
    \end{enumerate}
    As updates in bunches/distance estimates in bunches happen at most $O(\log_{1+\epsilon_4}(nW)+\log_{1+\epsilon_3}(nW))$ times, see Lemma~\ref{lm:size_heavy}, we get total update time  $ \tilde O (\tfrac{n^2}{p}\cdot (\log_{1+\epsilon_4}(nW)+\log_{1+\epsilon_3}(nW))) $.
    
    We also define the function \UpdateAdjacentDSMinHeapChange to process changes in $\tilde\delta^{\rm{nbr-bunch}}(\cdot,\cdot)$. \UpdateAdjacentDSMinHeapChange{$x,v,\Delta$} does one of the following three things:
     \begin{enumerate}[(i)]
        \item If $\Delta<\infty$ and the distance estimate from $x$ to $v$ is new, then we add the corresponding new estimates:\label{step:new2}
        \begin{itemize}
            \item We initialize $\text{Set}_{x,v}^{\rm{adj-heap}}$ to be empty.
            \item For all $u\in C(x)$, we add $\tilde\delta_B(u,x)+\tilde\delta^{\rm{nbr-bunch}}(x,v)$ to $Q^{\rm{adjacent}}_{u,v}$ and store $u$ together with its location in $Q^{\rm{adjacent}}_{x,v}$ in $\text{Set}_{x,v}^{\rm{adj-heap}}$.
        \end{itemize}
        \item If $\Delta<\infty$ and there already was a distance estimate from $x$ to $v$, then we update the corresponding estimates:\label{step:existing2}
        \begin{itemize}
            \item For each entry $x\in \text{Set}_{x,v}^{\rm{adj-heap}}$, update the corresponding entry $\tilde\delta_B(u,x)+\tilde\delta^{\rm{nbr-bunch}}(x,v)$ in $Q^{\rm{adjacent}}_{u,v}$ and update the entry in $\text{Set}_{x,v}^{\rm{adj-heap}}$ accordingly. 
        \end{itemize}
        \item If $\Delta=\infty$ and there is no longer a distance estimate from $x$ to $v$, then we remove the corresponding estimates:\label{step:del2}
        \begin{itemize}
            \item For each entry $x\in \text{Set}_{x,v}^{\rm{adj-heap}}$, remove the corresponding entries from $Q^{\rm{adjacent}}_{u,v}$ and from $\text{Set}_{x,v}^{\rm{adj-heap}}$. 
        \end{itemize}
    \end{enumerate}
    Since each $\tilde\delta^{\rm{nbr-bunch}}(\cdot,\cdot)$ updates at most $\log_{1+\epsilon_3}(nW)\log_{1+\epsilon_4}(nW)$ times by Lemma~\ref{lm:NbrBunchDS}, we get total update time $\tilde O(\tfrac{n^2}{p}\log_{1+\epsilon_4}(nW))$ for this operation. 

   Together, \UpdateNbrBunchDSBunchChange and \UpdateAdjacentDSBunchChange form the function \UpdateAdjacentDS. To be precise, we define \UpdateAdjacentDS as follows: \\
\UpdateAdjacentDS{$u,v,\Delta$}:
\begin{enumerate}
    \item \UpdateNbrBunchDSBunchChange{$u,v,\Delta$}
    \item \UpdateAdjacentDSBunchChange{$u,v,\Delta$}
\end{enumerate}
This means we obtain total update time $$\tilde O(\tfrac{n^2}{p}(\log_{1+\epsilon_4}(nW)+\log_{1+\epsilon_3}(nW)+\log_{1+\epsilon_3}(nW)\log_{1+\epsilon_4}(nW)))=\tilde O(\tfrac{n^2}{p}\log_{1+\epsilon_4}(nW)\log_{1+\epsilon_3}(nW)).$$
\end{proof}

\paragraph{Query time.}
Finally, we observe that the query is relatively straightforward and can be done in constant time. 

\begin{lemma}
    \DynAPSPtwo has constant query time. 
\end{lemma}
\begin{proof}
    A query $\Query{u,v}$ consists of the minimum over four values. The first two can be computed in constant time, as the components are all maintained explicitly. The last two are maintained explicitly. This gives us constant query time in total. 
\end{proof}

We summarize this section in the following lemma. 

\begin{lemma}\label{lm:mult_time}
    With high probability, the algorithm \DynAPSPtwo has total update time $ \tilde O ((pnm + \tfrac{n^2}{ p})\log^2(nW)) $ if $m=n^{1+\Omega(1)}$, $p=n^{\Omega(1)-1}$, and $m\leq n^{2-\rho}$ for an arbitrarily small constant $\rho$. If $m\geq n^{2-\rho}$, we have total update time $ \tilde O ((pnm + \tfrac{mn^ \rho+n^2}{ p})\log^2(nW)) $. Otherwise we have total running time $ \tilde O ((pnm + \tfrac{n^2}{ p})n^{o(1)}\log^2(nW)) $.
\end{lemma}
\begin{proof}
    We list the total update times as follows:
\begin{itemize}
    \item \textbf{Sampling.} $O(n)$
    \item \textbf{Distances from $A$.}
    \begin{itemize}
	\item $\tilde{O}(pnm)$ if $m=n^{1+\Omega(1)}$ and $pn=n^{\Omega(1)}$
	\item $\tilde O(pn^{1+o(1)}m)$ otherwise.
    \end{itemize}
    \item \textbf{Bunches and clusters.} 
    \begin{itemize}
	\item $ \tilde O (\tfrac{mn^\rho}{p}) $, for some arbitrarily small constant $\rho$, if $m=n^{1+\Omega(1)}$, 
	\item $\tilde O(\tfrac{mn^{o(1)}}{p})$ otherwise.
    \end{itemize}
   \item \textbf{Neighboring bunch data structure.} $ \tilde O (\tfrac{n^2}{p} (\log_{1+\epsilon_3}(nW)+\log_{1+\epsilon_4}(nW))) $
   \item \textbf{Adjacent data structure.} $ \tilde O (\tfrac{n^2}{p}\log_{1+\epsilon_3}(nW)\log_{1+\epsilon_4}(nW)) $.
\end{itemize}
This means for the case that $m=n^{1+\Omega(1)}$ and $p=n^{\Omega(1)-1}$, we have total update time
$$ \tilde O ((pnm + \tfrac{mn^{\rho}}{p}  + \tfrac{n^2}{ p})\log_{1+\epsilon_3}(nW)\log_{1+\epsilon_4}(nW)) = \tilde O ((pnm + \tfrac{mn^{\rho}}{p} + \tfrac{n^2}{ p})\log^2(nW)), $$
using that we hide constant factors of $1/\epsilon$ in the $\tilde O$-notation. If we furthermore assume that $n^\rho = O(n/p)$, i.e., $p = O(n^{1-\rho})$ for some arbitrarily small $\rho$, then we obtain the simplified version
$$ \tilde O ((pnm +   \tfrac{n^2}{ p})\log^2(nW)).$$
For other choices of $m,n,$ and $p$, we obtain 
$$ \tilde O ((pnm + \tfrac{n^2}{ p})n^{o(1)}\log^2(nW)). $$
\end{proof}

\subsection{Stretch Analysis}\label{sc:mult_stretch}
We can now follow the same structure as outlined in the high-level overview to prove the stretch in different case for each pair $u,v \in V$. In our analysis we also discuss how maintaining approximate bunches and distances impacts the final stretch.

\begin{lemma}\label{lm:mult_stretch}
    For all $u,v\in V$, the distance estimate $\hat{d}(u,v)$ returned by a query to \DynAPSPtwo is a $(2+\epsilon)$-approximation of the distance $d(u,v)$.
\end{lemma}
\begin{proof}
    Let $u,v \in V$. We will show that at any time $t$, if \Query{$u,v$} returns $\hat d^{(t)}(u,v)$, then we have $d^{(t)}(u,v)\leq \hat d^{(t)}(u,v)\leq (2+\epsilon)d^{(t)}(u,v)$. Since we look at all data structures at the same time~$t$, we omit the superscripts $\cdot^{(t)}$ in the remainder of this proof.
    
    First, let us note that $\hat{d}(u,v)$ consists of different distance estimates that all correspond to actual paths in $G$, hence we clearly have $d(u,v)\leq \hat d(u,v)$. Heaps do not contain `old' entries, entries that correspond to paths that are no longer in $G$ at the time of query, since the update functions \UpdateNbrBunchDSBunchChange, \UpdateNbrBunchDSEdgeChange, \UpdateAdjacentDSMinHeapChange, and \UpdateAdjacentDSBunchChange take such entries out as soon as they are no longer relevant. 
    
    Next, we do a case distinction as follows. 
    Let $\pi$ denote the shortest path from $u$ to $v$.
    
    \begin{figure}[htbp!]
    \centering
    \includegraphics[width=\textwidth]{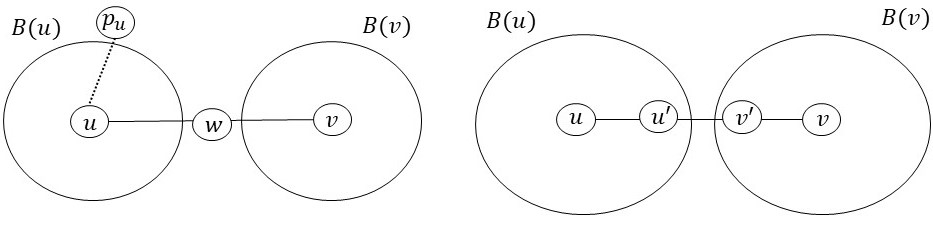}
    \caption{Possible scenarios for the overlap between the bunches $B(u)$ and $B(v)$ and the shortest path $\pi$ from $u$ to $v$.}
    \label{fig:pivots2}
    \end{figure}
    
    \begin{itemize}
        \item[Case 1:] There exists $w\in \pi$ such that $w\notin (B(u)\cup B(v))$ (the left case in Figure~\ref{fig:pivots2}).\\
        Since $w\notin B(u)$, we have $d(u,w)\geq   r(u) \geq \tfrac{\dLNbunch(u,\tilde p(u))}{1+\epsilon_3}$, so $\dLNbunch(u,\tilde p(u))\leq (1+\epsilon_3)d(u,w)$. 
        Similarly $\dLNbunch(v,\tilde p(v)) \leq (1+\epsilon_3)d(w,v)$. Since $d(u,w)+d(w,v)=d(u,v)$, we have either
        $d(u,w)\leq \tfrac{d(u,v)}{2}$ or $d(v,w)\leq \tfrac{d(u,v)}{2}$. 
        Without loss of generality, assume $d(u,w) \leq \tfrac{d(u,v)}{2}$, then by Query Step~\ref{step:pivot} in the algorithm we have:
        \begin{align*} 
        \hat{d}(u,v) &\leq \dLNbunch(u,\tilde p(u))+\delta_A(v,\tilde p(u)) \\
        &\leq \dLNbunch(u,\tilde p(u)) +(1+\epsilon_1)d(v,\tilde p(u))\\
        &\leq \dLNbunch(u,\tilde p(u))+(1+\epsilon_1)(d(u,\tilde p(u))+d(u,v))\\
        &\leq \dLNbunch(u,\tilde p(u))+(1+\epsilon_1)(\dLNbunch(u,\tilde p(u))+d(u,v))\\
        &\leq  (2+\epsilon_1)(1+\epsilon_3)d(u,w)+(1+\epsilon_1)d(u,v)\\
        &\leq (2+\tfrac{3}{2}\epsilon_1+\tfrac{1}{2}\epsilon_1\epsilon_3+\epsilon_3)d(u,v)\\
        &\leq (2+\epsilon)d(u,v),
        \end{align*}
        where in the last step we use $\epsilon_1=\epsilon_3=\epsilon/3$.
        \item[Case 2:] 
        There is no $w\in \pi$ such that $w\notin (B(u)\cup B(v))$ (the right case in Figure~\ref{fig:pivots2}).
        \\ 
        This means there has to be a direct edge from the cluster of $u$ to the cluster of $v$: the shortest path $\pi$ consists of a path from $u$ to $u'$, for some $u'\in B(u)$, an edge $\{u',v'\}\in E$, for some $v'\in B(v)$, and a path from $v'$ to $v$ (see the right case in Figure~\ref{fig:pivots2}). Note that in particular this includes the cases where $u\in B(v)$ or $v\in B(u)$, or more general when $B(u)\cap B(v)\neq \emptyset$, since we do not require $u\neq u'$, $v\neq v'$, $u'\notin B(v)$, or $v'\notin B(u)$.\\
        By Query Step~\ref{step:nbr_light1} we have
        \begin{align*}
            \hat d(u,v) &\leq \tilde\delta_B(u,u')+\tilde\delta^{\rm{nbr-bunch}}(u',v)\\
            &\leq (1+\epsilon_4)\delta_B(u,u')+(1+\epsilon_4)\delta^{\rm{nbr-bunch}}(u',v)\\
            &\leq (1+\epsilon_4)(1+\epsilon_2)d(u,u')+(1+\epsilon_4)(\tilde w(u',v')+\tilde\delta_B(v',v))\\ 
            &\leq (1+\epsilon_4)(1+\epsilon_2)d(u,u')+(1+\epsilon_4)^2(w(u',v')+\delta_B(v',v))\\
            &\leq (1+\epsilon_4)(1+\epsilon_2)d(u,u')+(1+\epsilon_4)^2(w(u',v')+(1+\epsilon_2)d(v',v))\\
            &\leq (1+\epsilon_4)^2(1+\epsilon_2)d(u,v)\\
            &\leq (2+\epsilon)d(u,v),
        \end{align*} 
        where in the last step we use $\epsilon_2=\epsilon_4=\epsilon/3$, and $\epsilon<1$. \qedhere
        \end{itemize}

\end{proof}

\subsection{Putting Everything Together}\label{sc:mult_result}
In this section we combine our lemmas in the following theorem. To obtain this result, we pick $p$ to balance the different terms in the total update time.

\mainthm*
\begin{proof}
    This follows from Lemma~\ref{lm:mult_time} and~\ref{lm:mult_stretch}, by picking $p$ to minimize $pnm + \tfrac{n^2}{ p}$. This gives $p=\sqrt{\tfrac{n}{m}}$, implying $pnm=\tfrac{n^2}{p}=m^{1/2}n^{3/2}$.
Note that whenever $m=n^{1+\Omega(1)}$ and $m\leq n^{2-\rho}$, this choice of $p$ indeed satisfies the other requirements of Lemma~\ref{lm:mult_time}: $p=n^{\Omega(1)-1}$. When $m=n^{1+o(1)}$, obtain an extra factor $n^{o(1)}$ in the total update time. And when $m\geq n^{2-\rho}$, we obtain the extra factor $n^{\rho}$.
\end{proof}

\section{\texorpdfstring{$(2+\epsilon,W_{u,v})$}{(2+epsilon,W{u,v})}-APSP for Weighted Graphs and Unweighted \texorpdfstring{$2$}{2}-APSP}\label{sc:2_W}
In this section, we show that an adapted version of our algorithm returns $(2+\epsilon,W_{u,v})$-approximate distances with improved update time, where $W_{u,v}$ is the maximum weight on a shortest path from $u$ to $v$. 
This result follows by the fact that in this case we do not need to handle the adjacent bunch case, but instead can consider an \emph{overlap} bunch case (see the middle case in Figure~\ref{fig:pivots_allcases}), which leads to an improved running time. 

To make this efficient, there are some additional complications. We introduce a parameter~$\tau$, the \emph{bunch overlap threshold}, which bounds the cluster size. For nodes $w\in B(u)\cap B(v)$ with cluster smaller than $\tau$, we can compute the overlap case. If all nodes in the overlap $B(u)\cap B(v)$ have larger clusters, we obtain an approximate by computing shortest paths from the heavy nodes. More details can be found in \Cref{sec:2_W}.

Moreover, in \Cref{sec:reduction}, we show a general reduction from mixed APSP to purely multiplicative APSP for unweighted graphs. We use this reduction to obtain our unweighted $(2+\epsilon,0)$-APSP result. 

\begin{figure}[htbp!]
    \centering
    \includegraphics[width=\textwidth]{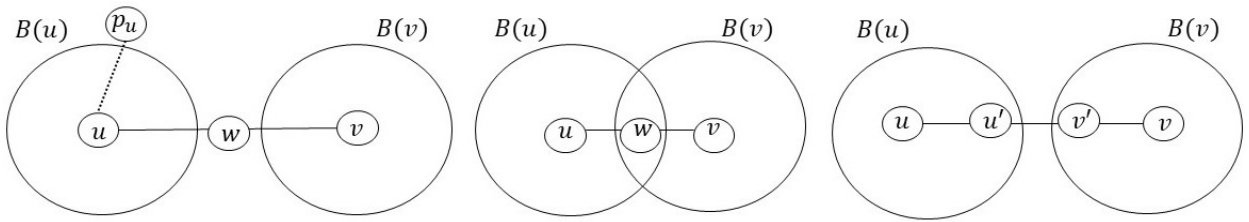}
    \caption{Possible scenarios for the overlap between the bunches $B(u)$ and $B(v)$ and the shortest path $\pi$ from $u$ to $v$.}
    \label{fig:pivots_allcases}
\end{figure}

\subsection{\texorpdfstring{$(2+\epsilon,W_{u,v})$}{(2+epsilon,W{u,v})}-APSP for Weighted Graphs}\label{sec:2_W}
We use the same structure as in \Cref{sec:TZ_2+eps}: we have bunches, and clusters exactly as before. We introduce the parameter $1\leq \tau\leq n$ to be set later. Now we define nodes to be heavy or light based on their cluster size. 
\begin{itemize}
    \item We call a node $ u $ \emph{heavy} if it is or ever was contained in at least $\tau$ approximate bunches $B(v)$, i.e., $u$ is heavy at time $t$ if there is a $t'\leq t$ such that $|C^{(t
     )}(u)|\geq \tau$. If a node is not heavy, we call it \emph{light}. 
    \item Initially, we denoted $V_{\rm{heavy}}$ for the set of heavy nodes, and $V_{\rm{light}}=V\setminus V_{\rm{heavy}}$ for the light nodes. By definition of heavy nodes, a light node can become heavy, but a heavy node can never become light. 
    \item Suppose we have a distance estimate $\delta_H(s,v)$ for all $s\in V_{\rm{heavy}},v\in V$. For each node $u \in V$, the \emph{heavy pivot} $q(u)$ is the closest node in $V_{\rm{heavy}}$ to $u$ according to the distance estimate $\delta_H$
\end{itemize}

We can show that the size of $ V_{\rm{heavy}} $ is bounded. 

\begin{lemma}\label{lm:size_heavy_moved}
 At any given time, with high probability, $ V_{\rm{heavy}} $ has size $\tilde O (\tfrac{n }{p\tau}\log_{1+\epsilon_3}(nW)) $.\label{lm:size_H}
\end{lemma}
\begin{proof}
    Let $B^{(t)}(v)$ and $C^{(t)}(w)$ denote the bunch of $v$  and the cluster of $w$ respectively at time $t$. Since bunch radius $r(v)$ increases by at least a factor $1+\epsilon_3$ if the bunch changes, we have at most $\log_{1+\epsilon_3}(nW)$ different bunch radii. Until the bunch radius increases, the bunch can only lose vertices, which does not contribute to the total size. Since each bunch is of size $\tilde{O}(\tfrac{1}{p})$ with high probability, we obtain that the total size of a bunch over the course of the algorithm is bounded by $\left|\bigcup_{t} B^{(t)}(v)\right|=\tilde{O}\left(\tfrac{\log_{1+\epsilon_3}(nW)}{p}\right)$. Hence we have 
    \begin{align*}
        |V_{\rm{heavy}}|\tau &\leq \sum_{w\in V_{\rm{heavy}}} \tau \\
        &\leq \sum_{w\in V_{\rm{heavy}}} \left| \bigcup_t C^{(t)}(w) \right| \\
        &\leq \sum_{w\in V} \left| \bigcup_t C^{(t)}(w) \right| \\
        &= \sum_{v\in V} \left| \bigcup_t B^{(t)}(v) \right| \\
        &= \tilde O ( \tfrac{n}{p}\log_{1+\epsilon_3}(nW)).
    \end{align*}
    Thus $|V_{\rm{heavy}}| = \tilde O (\tfrac{n }{p\tau}\log_{1+\epsilon_3}(nW))$.\qedhere
\end{proof}

Now we are ready to present our new algorithm. 
\textbf{Algorithm} \DynAPSPtwoW{$V,E,w,p,\tau,\epsilon$}\\
\Initialization{}:

Set $\epsilon_1=\epsilon_2=\epsilon_3=\epsilon_4=\epsilon/3$.
\begin{enumerate}
    \item \textbf{Sampling.} Sample each node with probability $p$ to construct $A$. \label{step:sampling_moved}
    \item \textbf{Distances from $A$.} Compute $(1+\epsilon_1)$-approximate distances $\delta_A(s,v)$ for $s\in A$ and $v\in V$. 
    \item \textbf{Bunches and clusters.} Initialize the algorithm for maintaining the approximate bunches $\tilde B$ and pivots $\tilde p$ according to Lemma~\ref{lm:approx_bunches_on_G}, denoting with $\dLNbunch(v,w)$ the $(1+\epsilon_2)$-approximate distances for $w\in \tilde B(v)$.\\
    Initialize $(1+\epsilon_3)$-approximate bunch radii $r(v)$ with respect to these approximate pivots $\tilde p(v)$.  \label{step:A_MSSP_moved}
    We initialize another approximation to the bunches, denoted by $B(v)$, which will be defined with respect to the bunch radius $r(v)$. Initially, we simply set $B(v)=\tilde{B}(v)$. We let $\delta_B(u,v):=\dLNbunch(u,v)$, and set $\tilde \delta_B(u,v) := (1+\epsilon_4)^{\lceil \log_{1+\epsilon_4}(\delta_B(u,v))\rceil}$. Initialize $V_{\rm{heavy}}:= \{w\in V : |C(w)|>\tau\}$. \label{step:bunches_moved}
    \item \textbf{Distances from $V_{\rm{heavy}}$.} Compute $(1+\epsilon_1)$-approximate distances $\delta_H(s,v)$ for $s\in V_{\rm{heavy}}$ and $v\in V$. Further, for all $v\in V$, we initialize a Min-Heap $Q^{\rm{heavy-pivot}}_v$ that stores $\delta_H(s,v)$ for $s\in V_{\rm{heavy}}$. The heavy pivot $q(v)$ is the minimum entry of $Q^{\rm{heavy-pivot}}_v$. \label{step:H_MSSP_moved}
    \item \textbf{Overlap data structure.} For $u,v\in V$, initialize a Min-Heap $Q^{\rm{overlap}}_{u,v}$, with an entry for each $w\in B(u)\cap B(v)$ that is not contained in $V_{\rm{heavy}}$, with key $\tilde \delta_B(u,w)+\tilde \delta_B(v,w)$.
\end{enumerate}

\noindent\Delete{$u,v$}:\\
\indent \Update{$u,v,\infty$}\\

\noindent\Update{$u,v,w(u,v)$}:
\begin{enumerate}[(1)]
     \item \textbf{Distances from $A$.} For all $s\in A$, run \Update{$u,v,w(u,v)$} on the SSSP approximation denoted by $\delta_A(s,\cdot)$. 
    \item \textbf{Distances from $V_{\rm{heavy}}$.} For all $s\in V_{\rm{heavy}}$, run \Update{$u,v,w(u,v)$} on the SSSP approximation denoted by $\delta_H(s,\cdot)$, and update $Q^{\rm{heavy-pivot}}_{\cdot}$ accordingly. 
    \item \textbf{Update bunches and heaps.} Maintain the bunches and clusters, and update the heaps accordingly. See below for the details.  \label{step:update3_moved}
\end{enumerate}

\noindent\Query{$u,v$}:\\
    Output $\hat d(u, v)$ to be the minimum of 
    \begin{enumerate}[(a)]\setcounter{enumi}{1}
         \item $\min\{ \dLNbunch(u,\tilde p(u))+\delta_A(v,\tilde p(u)),\delta_A(u,\tilde p(v))+\dLNbunch(v,\tilde p(v))\}$;\label{step:pivotW}
        \item $\min\{ \delta_H(u,q(u))+\delta_H(v,q(u)),\delta_H(u,q(v))+\delta_H(v,q(v))\}$;\label{step:quvitW}
        \item The minimum entry of $Q^{\rm{overlap}}_{u,v}$, which gives $\min\{\tilde \delta_B(u,w)+\tilde \delta_B(v,w) : w\in B(u)\cap B(v) \text{ and }w\in V_{\rm{light}}\}$. \label{step:w_light} 
    \end{enumerate}
    
\paragraph{Updating bunches and heaps.}
For completeness, we recall how we update the bunches and clusters. We now maintain different heaps. Again, we maintain the approximate bunches $\tilde B(v)$ for all $v\in V$ by Lemma~\ref{lm:approx_bunches_on_G}. From here maintain the additional approximate bunches $B(v)$ for all $v\in V$ as given in Definition~\ref{def:app_app_bunch}, i.e., 
If there is a change to $\dLNbunch(v',\tilde p(v'))$, do the following two things:
    \begin{itemize}
        \item We check if $\dLNbunch^{(t)}(v',\tilde p^{(t)}(v'))> (1+\epsilon_3)r^{(t-1)}(v')$, and if so we set $r^{(t)}(v')=\dLNbunch^{(t)}(v',p^{(t)}(v'))$ and we set $B^{(t)}(v):=\tilde B^{(t)}(v)$. If not, we set $r^{(t)}(v')=r^{(t-1)}(v')$ and $B^{(t)}(v)=B^{(t-1)}(v)$.
    \end{itemize}
 Let $u'\in V$ be a node for which there is a change to the bunch involving a node $w\in V$. Such a change can be: \ref{case:smaller bunch_copy} $w$ can leave the bunch of $u'$ \ref{case:distance in bunch changed_copy} the distance estimate of $w$ to $u'$ can change while $w$ stays within the bunch, or \ref{case:bigger bunch_copy} $w$ can join the bunch of $u'$. We formalize this as follows, using $\Delta$ to denote which case we are in.
    \begin{enumerate}[(i)]
        \item $w\in B^{(t-1)}(u')$ and $w\notin B^{(t)}(u')$, we set $\Delta=\infty$ to indicate that $w$ is no longer in $B(u')$. \label{case:smaller bunch_copy}
        \item $w\in B^{(t-1)}(u')$ and $w\in B^{(t)}(u')$ with an increased distance estimate $\tilde \delta_B^{(t)}(u',w)$, we set $\Delta=\tilde \delta_B^{(t)}(u',w)$ to be the new distance. \label{case:distance in bunch changed_copy}
        \item \label{case:bigger bunch_copy} $w\notin B^{(t-1)}(u')$ and $w\in B^{(t)}(u')$, we set $\Delta=\tilde \delta_B^{(t)}(u',w)$ to be the distance. 
    \end{enumerate}

    Next, we do the following. 
    If we are in Case~\ref{case:bigger bunch_copy}, we check the cluster size of $w$: if $|C^{(t)}(w)|\geq\tau$, then we add $w$ to $V_{\rm{heavy}}$ and we initialize a SSSP approximation, denoted by $\delta_H(w,\delta)$. Moreover, for each $u'$ in $C^{(t-1)}(w)$ we do the following:
\begin{itemize}
    \item \UpdateOverlapDS{$w,u',\infty$} to remove the corresponding entries $\tilde\delta_B(u',w)+\tilde\delta_B(v',w)$ from the data structures $Q^{\rm{overlap}}_{u',v'}$ for all $v'\in \text{Set}^{\rm{overlap}}_{w,u'}$. Here $\text{Set}^{\rm{overlap}}_{w,u'}$ stores exactly those $v'$ such that there is an entry for $w$ in $Q^{\rm{overlap}}_{u',v'}$. The details on $\text{Set}^{\rm{overlap}}$ can be found together with the definition and total update time of \UpdateOverlapDS are given in Lemma~\ref{lm:min_heap_w}.
\end{itemize}

If we are in Case~\ref{case:bigger bunch_copy} but $|C^{(t)}(w)|<\tau$, or we are in Case~\ref{case:smaller bunch_copy} or~\ref{case:distance in bunch changed_copy}, then we do the following
\begin{itemize}
    \item \UpdateOverlapDS{$w,u',\Delta$} to update the corresponding entries $\tilde\delta_B(u',w)+\tilde\delta_B(v',w)$ in the data structures $Q^{\rm{overlap}}_{u',v'}$ for all $v'\in \text{Set}^{\rm{overlap}}_{w,u'}$. 
\end{itemize}

\paragraph{Distances from $ V_{\rm{heavy}} $.}
First, we note that a node is heavy if it is contained in at least $\tau$ bunches, or equivalently, if $|C(w)|>\tau$. Hence by Lemma~\ref{lm:updatebunches}, we can easily keep track when nodes become heavy and should join $V_{\rm{heavy}}$. 

Next, we use the shortest path algorithm from Lemma~\ref{lm:decr_MSSP}, which we combine with Lemma~\ref{lm:size_heavy}, which states that $V_{\rm{heavy}}$ might be growing over time, but its size is always bounded by $ \tilde O (\tfrac{n }{p\tau}\log_{1+\epsilon_3}(nW)) $. 
Moreover, we maintain the heavy pivots $q(v)$ by maintaining a Min-Heap $Q^{\rm{heavy~pivot}}_v$ for each node $v$, with as entries all nodes $u\in V_{\rm{heavy}}$ with key $\delta_H(u,v)$. Throughout, we will use that we can extract the minimum entry of a Min-Heap in constant time, that we can insert entries in $\log(n)$ time, and that we can edit or delete entries in $\log(n)$ time, given that we know either the key or have a pointer to their location in the Min-Heap. 

This clearly does not incur more than a constant factor in the running time, so together we obtain the following lemma. 

\begin{lemma} \label{lm:heavy_MSSP}
    With high probability, we can maintain $(1+\epsilon_2)$-SSSP from each node in $ V_{\rm{heavy}} $, providing us with $\delta_H(s,v)$ for $s\in V_{\rm{heavy}}$ and $v\in V$. Further, for all $v\in V$, we can maintain heavy pivots $q(v)$ if $d(v,q(v))\leq d$. We can do this in total update time $ \tilde O (\tfrac{m n }{p\tau}\log_{1+\epsilon_3}(nW)) $ if $m=n^{1+\Omega(1)}$ and $\tfrac{n}{p\tau}=n^{\Omega(1)}$, and $ \tilde O (\tfrac{m n^{1+o(1)} }{p\tau}\log_{1+\epsilon_3}(nW)) $ otherwise. 
\end{lemma}

\paragraph{Overlap data structure.}
Now we provide the details of how we maintain the data structure $Q^{\rm{overlap}}_{u,v}$ for all $u,v\in V$, which has entries for all node $w\in B(u)\cap B(v)\cap V_{\rm{light}}$.


\begin{lemma}\label{lm:min_heap_w}
   With high probability, we can maintain the data structures $Q^{\rm{overlap}}_{u,v}$ for all $u,v\in V$ in total update time $\tilde O (\tfrac{n\tau }{p} \cdot\log_{1+\epsilon_3}(nW)\log_{1+\epsilon_4}(nW)) $.
\end{lemma}
\begin{proof}
    Conceptually, we initialize a Min-Heap $Q^{\rm{overlap}}_{u,v}$ for each pair $(u,v)\in V\times V$. In practice, we only create the Min-Heap if the first entry is about to be added. This avoids a $\Theta(n^2)$ initialization time. Now fix $v\in V$, and consider the moment we add a node $w$ to $B(v)$. For this data structure, we only need to consider $w\in V_{\rm{light}}$, hence $w$ is part of at most $ \tau $ bunches. When we add $w$ to to $B(v)$, we know it is part of at most $\tau$ other bunches $B(u)$. For each such tuple $(u,w,v)$, we add an entry to the Min-Heap $Q^{\rm{overlap}}_{u,v}$ for $(u,v)$, which stores $w$ with as key the distance estimate $\delta_w(u,v):=\tilde\delta_B(u,w)+\tilde\delta_B(w,v)$. For the total time, notice that for each $v\in V$, we have at most $\tilde O(\tfrac{1}{p})$ different $w\in B(v)$ for which we store at most $\tau$ distances, hence we have total initialization time $ \tilde O (n \cdot \tfrac{1}{p} \cdot \tau) $.
    
    Next, we consider how to update this data structure. Suppose a value $\tilde\delta_B(u,w)$ changed for some $w\in B(u)$ due to an update to the bunches (Lemma~\ref{lm:updatebunches}), then we would need to change the corresponding entries in $Q^{\rm{overlap}}_{u,v}$. Lemma~\ref{lm:updatebunches} gives us $u$, $w$, and the new distance $\tilde \delta_B(u,w)$, which we need to update in all data structures $Q^{\rm{overlap}}_{u,v}$ for which $w\in B(u) \cap B(v)$. Clearly checking all possible~$v$ and looking through the entire $Q^{\rm{overlap}}_{u,v}$ is too slow. Instead, we keep an additional data structure for each node $w\in B(u)$: a set $\text{Set}_{w,u}^{\rm{overlap}}$ that shows for which $v\in V$ it appears in $Q^{\rm{overlap}}_{u,v}$ and a pointer to where in $Q^{\rm{overlap}}_{u,v}$ it appears. We equip this set with the structure of a self-balancing binary search tree, so we have polylogarithmic insertion and removal times. We update this every time the appearance or location in $Q^{\rm{overlap}}_{u,v}$ changes. Now if $\tilde\delta_B(u,w)$ changes due to an update, we can efficiently change the corresponding places in $Q^{\rm{overlap}}_{u,v}$ as follows. 
    
    We define the function \UpdateOverlapDS{$u,w,\Delta$}, where $w$ is an element of $B^{(t-1)}(u)$ and/or $B^{(t)}$, and $\Delta$ is the new value $\tilde\delta_B^{(t)}(u,w)$, or $\infty$ if $w$ is no longer in $B^{(t)}(u)$. Note that we will only call this function for light $w$.\footnote{If $w$ becomes heavy we remove all the corresponding estimates.} \UpdateOverlapDS does one of three things:
    \begin{enumerate}[(i)]
        \item If $\Delta<\infty$, $w\notin B^{(t-1)}(u)$ and $w\in B^{(t)}(u)$. We add the new estimates:\label{step:new}
        \begin{itemize}
            \item We initialize $\text{Set}_{w,u}^{\rm{overlap}}$ to be empty.
            \item For each $v\in C(w)$ we add $\delta_w(u,v):= \tilde\delta_B(u,w)+\tilde\delta_B(w,v)$ to $Q^{\rm{overlap}}_{u,v}$ and store $v$ together with its location in $Q^{\rm{overlap}}_{u,v}$ as entries in $\text{Set}_{w,u}^{\rm{overlap}}$ and $\text{Set}_{w,v}^{\rm{overlap}}$.
        \end{itemize}
        \item If $\Delta<\infty$, $w\in B^{(t-1)}(u)$ and $w\in B^{(t)}(u)$. We update the estimates:\label{step:existing}
        \begin{itemize}
            \item For each entry of $v$ in $\text{Set}_{w,u}^{\rm{overlap}}$, update $\delta_w(u,v)= \tilde\delta_B(u,w)+\tilde\delta_B(w,v)$ in $Q^{\rm{overlap}}_{u,v}$ and update the entries in $\text{Set}_{w,u}^{\rm{overlap}}$ and $\text{Set}_{w,v}^{\rm{overlap}}$ accordingly.
        \end{itemize}
        \item If $\Delta=\infty$, then $w\in B^{(t-1)}(u)$ and $w\notin B^{(t)}(u)$. We remove the corresponding estimates:\label{step:del}
        \begin{itemize}
            \item For each entry of $v$ in $\text{Set}_{w,u}^{\rm{overlap}}$, remove $\delta_w(u,v)$ from $Q^{\rm{overlap}}_{u,v}$ and remove the corresponding entries from $\text{Set}_{w,u}^{\rm{overlap}}$ and $\text{Set}_{w,v}^{\rm{overlap}}$.
        \end{itemize}
    \end{enumerate}

    For the total update time, notice that in total at most $\tilde O(\tfrac{n}{p}\cdot \log_{1+\epsilon_3}(nW))$ nodes get added to a bunch, see Lemma~\ref{lm:size_heavy}. Since we only call this function when the added node $w$ is light at that time, we have $|C(w)|\leq \tau$ and thus step~\ref{step:new} takes $O(\tfrac{n\tau}{p}\cdot \log_{1+\epsilon_3}(nW))$ total update time. For step~\ref{step:existing}, note that the distance can be updated at most $O(\log_{1+\epsilon_4}(nW))$~times and therefore this step has total update time $O(\tfrac{n\tau}{p}\cdot \log_{1+\epsilon_3}(nW)\log_{1+\epsilon_4}(nW))$. Further nodes get deleted from bunches at most as often as they get added, so the total update time of step~\ref{step:del} is dominated by step~\ref{step:new}. We obtain total update time $\tilde O ( \tfrac{n\tau }{p} \cdot \log_{1+\epsilon_3}(nW)\log_{1+\epsilon_4}(nW)) $.
\end{proof}

\paragraph{Update time.}
We summarize the running times in the following lemma. 

\begin{lemma}\label{lm:mult_2W_time}
    With high probability, the algorithm \DynAPSPtwoW has total update time $ \tilde O ((pnm + \tfrac{m n }{p\tau}+\tfrac{n\tau}{ p})\log^2(nW)) $ if $m=n^{1+\Omega(1)}$, $p=n^{\Omega(1)-1}$, $\tfrac{n}{p\tau}=n^{\Omega(1)}$, and $\tau = O(n^{1-\rho})$. Otherwise we have total running time $ \tilde O ((pnm +\tfrac{m n }{p\tau} +\tfrac{n\tau}{ p})n^{o(1)}\log^2(nW)) $.
\end{lemma}
\begin{proof}
    We list the total update times as follows:
\begin{itemize}
    \item \textbf{Sampling.} $O(n)$
    \item \textbf{Distances from $A$.}
    \begin{itemize}
	\item $\tilde{O}(pnm)$ if $m=n^{1+\Omega(1)}$ and $pn=n^{\Omega(1)}$
	\item $\tilde O(pn^{1+o(1)}m)$ otherwise.
    \end{itemize}
    \item \textbf{Bunches and clusters.} 
    \begin{itemize}
	\item $ \tilde O (\tfrac{mn^\rho}{p}) $, for some arbitrarily small constant $\rho$, if $m=n^{1+\Omega(1)}$, 
	\item $\tilde O(\tfrac{mn^{o(1)}}{p})$ otherwise.
    \end{itemize}
    \item \textbf{Distances from $ V_{\rm{heavy}} $.}
    \begin{itemize}
	\item $ \tilde O (\tfrac{m n }{p\tau}\log_{1+\epsilon_3}(nW)) $ if $m=n^{1+\Omega(1)}$ and $\tfrac{n}{p\tau}=n^{\Omega(1)}$
	\item $ \tilde O (\tfrac{m n^{1+o(1)} }{p\tau}\log_{1+\epsilon_3}(nW)) $ otherwise.
    \end{itemize}
   \item \textbf{Overlap data structure.} $\tilde O (\tfrac{n\tau }{p} \cdot\log_{1+\epsilon_3}(nW)\log_{1+\epsilon_4}(nW)) $.
\end{itemize}

This means for the case that $m=n^{1+\Omega(1)}$, $p=n^{\Omega(1)-1}$, and $\tfrac{n}{p\tau}=n^{\Omega(1)}$, we have total update time
$$ \tilde O ((pnm + \tfrac{mn^\rho}{p}+\tfrac{mn^{\rho}}{p}  + \tfrac{m n }{p\tau}+\tfrac{n\tau}{ p})\log_{1+\epsilon_3}(nW)\log_{1+\epsilon_4}(nW)) = \tilde O ((pnm + \tfrac{mn^{\rho}}{p} + \tfrac{m n }{p\tau}+\tfrac{n\tau}{ p})\log^2(nW)), $$
using that we hide constant factors of $1/\epsilon$ in the $\tilde O$-notation. 
For other choices of $m,n,p$ and $\tau$, we obtain 
$$ \tilde O ((pnm + \tfrac{mn^{\rho}}{p} + \tfrac{m n }{p\tau}+\tfrac{n\tau}{ p})n^{o(1)}\log^2(nW)). $$

If we furthermore assume that $n^\rho = O(n/\tau)$, i.e., $\tau = O(n^{1-\rho})$ for some arbitrarily small $\rho$, then we obtain the simplified version
$$ \tilde O ((pnm +   \tfrac{m n }{p\tau}+\tfrac{n\tau}{ p}\log^2(nW)).$$
For other choices of $m,n,p,$ and $\tau$, we obtain 
\begin{align*}
     \tilde O ((pnm + \tfrac{m n }{p\tau}+\tfrac{n\tau}{ p})n^{o(1)}\log^2(nW)). 
\end{align*}
\end{proof}

\thmtwoW*
\begin{proof}
    This result follows from a exchanging the adjacent data structure for the overlap data structure as described in this section. 
    
    By Lemma~\ref{lm:mult_2W_time} we have terms a running time scaling with $pnm+mn/(p\tau)+n\tau/p$. To balance out these terms, we set $mn/(p\tau)=n\tau/p$, hence $\tau=\sqrt{m}$. We set $pnm = mn/(p\tau)$, hence $p=m^{-1/4}$. Note that for $m=n^{1+o(1)}$, this choice for $\rho$ and $\tau$ also satisfies the additional constraints of Lemma~\ref{lm:mult_2W_time}, so in this case we obtain total update time $\tilde O( (pnm+mn/(p\tau)+n\tau/p)\log^2(nW))= \tilde O (nm^{3/4}\log^2(nW))$. If $m=n^{1+o(1)}$ we obtain $\tilde O (n^{1+o(1)}m^{3/4}\log^2(nW))$.
    
    \paragraph{Stretch analysis.} It remains to show that $\hat d(u,v)$ is a $(2+\epsilon,W_{u,v})$-approximation of the distance $d(u,v)$. We do a similar case distinction as in Lemma~\ref{lm:mult_stretch}. Case 1 is exactly the same, since we maintain the relevant estimate through the pivots. For the remainder, we know there is no $w\in \pi$ such that $w\in (B(u)\cup B(v))$, hence the two bunches are adjacent. We know create cases depending on whether the bunches overlap, and if so, if the node in the overlap is heavy or light. 
    
    \begin{itemize}
        \item[Case 2:] $\pi\cap B(u)\cap B(v) =\emptyset$ (the right case in Figure~\ref{fig:pivots_allcases}).\\
        Hence $\pi$ consists of a path from $u$ to $u'\in B(u)$, the edge $\{u',v'\}$ with $v'\in B(v)$, and a path from $v'$ to $v$.
        Since $v'\notin B(u)$, we have $d(u,v')\geq   r(u) \geq \tfrac{\dLNbunch(u,\tilde p(u))}{1+\epsilon_3}$, so $\dLNbunch(u,\tilde p(u))\leq (1+\epsilon_3)d(u,v')$. Similarly, we have $\dLNbunch(v,\tilde p(v))\leq (1+\epsilon_3)d(u',v)$. Combining this, we obtain $\dLNbunch(u,\tilde p(u))+ \dLNbunch(v,\tilde p(v))\leq (1+\epsilon_3)(d(u',v)+d(u,v')) = (1+\epsilon_3)(d(u,v)+w(u',v'))$. Without loss of generality, assume that $\dLNbunch(u,\tilde p(u))\leq  \tfrac{1+\epsilon_3}{2}(d(u,v)+w(u',v'))$.  By Query Step~\ref{step:pivotW} we have:
        \begin{align*} 
            \hat{d}(u,v) &\leq \dLNbunch(u,\tilde p(u))+\delta_A(v,\tilde p(u)) \\
            &\leq \dLNbunch(u,\tilde p(u)) +(1+\epsilon_1)d(v,\tilde p(u))\\
            &\leq \dLNbunch(u,\tilde p(u))+(1+\epsilon_1)(d(u,\tilde p(u))+d(u,v))\\
            &\leq \dLNbunch(u,\tilde p(u))+(1+\epsilon_1)(\dLNbunch(u,\tilde p(u))+d(u,v))\\
            &\leq  (2+\epsilon_1)\tfrac{1+\epsilon_3}{2}(d(u,v)+w(u',v'))+(1+\epsilon_1)d(u,v)\\
            &\leq (2+\tfrac{3}{2}\epsilon_1+\tfrac{1}{2}\epsilon_1\epsilon_3+\epsilon_3)d(u,v)+(1+\tfrac{1}{2}\epsilon_1+\tfrac{1}{2}\epsilon_1\epsilon_3+\epsilon_3)w(u',v') \\
            &\leq (2+\epsilon/2)d(u,v)+(1+\epsilon/2)W_{u,v}\\
            &\leq (2+\epsilon)d(u,v)+W_{u,v}.
        \end{align*}
        \item[Case 3:] 
            There exists $w\in \pi$ such that $w\in (B(u)\cap B(v))$ and $w\in V_{\rm{light}}$ (the middle case in Figure~\ref{fig:pivots_allcases}).\\
            By Query Step~\ref{step:w_light}, we have 
            \begin{align*}
                \hat d(u,v) &\leq \tilde\delta_B(u,w)+\tilde\delta_B(v,w)\\
                &\leq (1+\epsilon_4)(\delta_B(u,w)+\delta_B(v,w))\\
                &\leq (1+\epsilon_4)(1+\epsilon_2)(d(u,w)+d(v,w))\\
                &= (1+\epsilon_2+\epsilon_2\epsilon_4+\epsilon_4)d(u,v)\\
                &\leq (1+\epsilon)d(u,v).
            \end{align*}
            \item[Case 4:] 
            There exists $w\in \pi$ such that $w\in (B(u)\cap B(v))$ and $w\in V_{\rm{heavy}}$ (See Figure~\ref{fig:heavy_pivots2W}).\\
            If $w=q(u)$ then by Query Step~\ref{step:quvitW}, we have $\hat d(u,v)\leq \delta_H(u,q(u))+\delta_H(v,q(u))\leq (1+\epsilon_1)(d(u,q(u))+d(v,q(u))=(1+\epsilon_1)d(u,v)\leq (1+\epsilon)d(u,v)$. If $w=q(v)$, the argument is analogous. 
            
            If $q(u)\neq w \neq q(v)$, then by definition of $q(u)$ and $q(v)$ we obtain $\delta_H(u,q(u))\leq \delta_H(u,w)$ and $\delta_H(v,q(v))\leq \delta_H(v,w)$. So we have $\delta_H(u,q(u))+\delta_H(v,q(v))\leq (1+\epsilon_1)d(u,v)$, hence either $\delta_H(u,q(u))\leq (1+\epsilon)\tfrac{d(u,v)}{2}$ or $\delta_H(v,q(v))\leq (1+\epsilon_1)\tfrac{d(u,v)}{2}$. Without loss of generality, assume $\delta_H(u,q(u))\leq (1+\epsilon_1) \tfrac{d(u,v)}{2}$. Combining this with the triangle inequality, we obtain
            \begin{align*}
                \hat d(u,v) &\leq \delta_H(u,q(u))+\delta_H(q(u),v)\\
                &\leq \delta_H(u,q(u))+(1+\epsilon_1)d(q(u),v)\\
                &\leq \delta_H(u,q(u))+(1+\epsilon_1)(d(u,q(u))+d(u,v))\\
                &\leq (2+\epsilon_1)\delta_H(u,q(u))+(1+\epsilon_1)d(u,v)\\
                &\leq (2+\tfrac{5}{2}\epsilon_1+\tfrac{1}{2}\epsilon_1^2)d(u,v)\\
                &\leq (2+\epsilon)d(u,v).
            \end{align*}  
    \end{itemize}

    \paragraph{Query time.}
    A query $\Query{u,v}$ consists of the minimum over three values. The first two can be computed in constant time, as the components are all maintained explicitly. The last one is maintained explicitly. This gives us constant query time in total.
\end{proof}

In particular, this means we get $(2+\epsilon,1)$-approximate APSP for unweighted graphs. 

\begin{figure}[htbp!]
    \centering
    \includegraphics[width=0.35\textwidth]{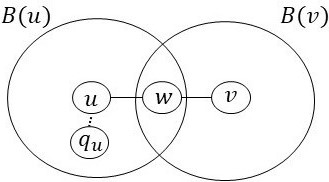}
    \caption{Possible scenario for the overlap between the bunches $B(u)$ and $B(v)$ and the shortest path $\pi$ from $u$ to $v$ where the estimate can be obtained using heavy pivots $q_u$ and $q_v$.}
    \label{fig:heavy_pivots2W}
    \end{figure}


\subsection{Reduction and Unweighted \texorpdfstring{$(2+\epsilon)$}{(2+epsilon)}-APSP}\label{sec:reduction} 
Intuitively, we split every edge into $k+1$ edges to form $G'$. Then we compute $(2,k)$-APSP on $G'$, divide the result by $k+1$, and round down to the nearest integer. Because the additive error $k$ is not enough to reach a new node in $G$, we can simply ignore this part of the path. In the following theorem we make this formal. 

\thmreduction*
\begin{proof}
    We construct the auxiliary graph $G'=(V',E')$ as follows.
    \begin{align*}
        V' &:= V \cup \{ \{u,v\}_i : \{u,v\}\in E, i\in\{1, 2, \dots, k\}\}; \\
        E' &:= \bigcup_{\{u,v\}\in E} \{ \{u, \{u,v\}_1\}, \{\{u,v\}_1,\{u,v\}_2\}, \dots, \{\{u,v\}_{k-1},\{u,v\}_k\}, \{\{u,v\}_k, v\} \}.
    \end{align*}
    Clearly this graph has $n+km$ vertices and $(k+1)m$ edges. We compute $(a,k)$-APSP on $G'$ in $\tau_{\mathcal A}(n+km,(k+1)m)$ time using algorithm $\mathcal{A}$, and denote the output by $\delta_{G'}$.
    We set $\delta(u,v) := \lfloor \tfrac{\delta_{G'}(u,v)}{k+1}\rfloor$ for all $u,v\in V$. 

    \textbf{Claim.} $\delta$ gives $(a+(k+2)\epsilon,0)$-approximate APSP on $G$. 

    Let $u,v\in V$. We remark that $d_{G'}(u,v)=(k+1)d_G(u,v)$ for all $u,v\in V$. Which we will use twice. 
    
    First, we show that $d_G(u,v)\leq \delta(u,v)$.  This follows from $\delta(u,v) = \lfloor \tfrac{\delta_{G'}(u,v)}{k+1}\rfloor \geq \lfloor \tfrac{d_{G'}(u,v)}{k+1}\rfloor = \lfloor \tfrac{(k+1)d_G(u,v)}{k+1}\rfloor=d_G(u,v)$. 

    Second, we show that $\delta(u,v) \leq (a+(k+2)\epsilon)\cdot d_G(u,v)$. This follows from
    $$\delta(u,v) = \lfloor \tfrac{\delta_{G'}(u,v)}{k+1}\rfloor \leq \lfloor \tfrac{(a+\epsilon)\cdot d_{G'}(u,v)+k}{k+1}\rfloor = \lfloor \tfrac{(a+\epsilon)\cdot (k+1)d_{G}(u,v)+k}{k+1}\rfloor= \lfloor (a+\epsilon)\cdot d_G(u,v) +\tfrac{k}{k+1}\rfloor.$$

    Now if $\epsilon=0$, then $a\cdot d_G(u,v)$ is an integer, hence $\lfloor a\cdot d_G(u,v) +\tfrac{k}{k+1}\rfloor = a\cdot d_G(u,v) $. If $\epsilon>0$, then $\epsilon d_G(u,v)+ \tfrac{k}{k+1}$ can be bigger than $1$ -- if it is not, we are done by the same argument. Now we note that 
    \begin{align*}
        (k+1)\epsilon\cdot d_G(u,v) \geq (k+1)\left(1-\tfrac{k}{k+1}\right) =1.
    \end{align*}
    Hence we obtain 
    \begin{align*}
        \lfloor (a+\epsilon)\cdot d_G(u,v) +\tfrac{k}{k+1}\rfloor &\leq (a+\epsilon)\cdot d_G(u,v) +1 \\
        &\leq (a+\epsilon)\cdot d_G(u,v) + (k+1)\epsilon\cdot d_G(u,v)\\
        &= (a+(k+2)\cdot \epsilon)d_G(u,v).
    \end{align*}

    For the claim about dynamic algorithms, note that any update to an edge in $G$, corresponds to updating $k+1$ edges in $G'$. The query consists of computing $\delta(u,v) = \lfloor \tfrac{\delta_{G'}(u,v)}{k+1}\rfloor$, which requires a query to the distance in $G'$ plus constant time operations. 
\end{proof}

In particular, this can give fast algorithms for sparse graphs.
For the decremental case, we obtain the following result. 
\thmtwoUnw*
\begin{proof}
    We apply the theorem above with the $(2+\epsilon,1)$-APSP of this paper, Theorem 1.2. This takes $\tilde O( nm^{3/4})$ time if $m=n^{1+\Omega(1)}$, and $\tilde O( n^{1+o(1)}m^{3/4})$ time otherwise.
\end{proof}


\newpage
\section{Approximate APSP with Additive Factors}\label{sc:additive}
In this section we give our second result, which is maintaining APSP on unweighted graphs with a mixed additive and multiplicative stretch. We start by reviewing useful tools in \Cref{sec:additive_tools}.  We then provide a decremental $(1+\epsilon, 2)$-APSP algorithm as a warm-up in \Cref{sec:additive_warm_up}. Finally in \Cref{sec:add_general} we generalize our techniques to maintain a $(1+\epsilon, 2(k-1))$-APSP algorithm efficiently.

\subsection{Decremental Tools}  \label{sec:additive_tools}

We start with describing several useful tools we use in our algorithms.

\remove{
\paragraph{Even-Shiloach tree (ES-tree).} An ES-tree is a central data structure for maintaining distances in a decremental setting. It maintains the exact distances from a source $s$ under edge deletions, giving the following.

\begin{lemma}(\cite{ES81,henzinger1995fully,king1999fully})\label{lemma_ES}
There is a data structure called ES-tree that, given a weighted directed graph $G$ undergoing deletions and edge weight increases, a root node $s$, and a depth parameter $d$, maintains, for every node $v$ a value $\ell_s(v)$ such that $\ell_s(v) = d_G(s,v)$ if $d_G(s,v) \leq d$ and $\ell_s(v) = \infty$ if $d_G(s,v) > d$.
It has constant query time and a total update time of $O(md+n)$, where $m$ is the number of edges in $G$.
\end{lemma}
}

\paragraph{Even-Shiloach tree (ES-tree).} An ES-tree is a central data structure for maintaining distances in a decremental setting. It maintains the exact distances from a source $s$ under edge deletions, giving the following.

\begin{lemma}[\cite{ES81,henzinger1995fully,king1999fully}]\label{lemma_ES}
There is a data structure called ES-tree that, given a weighted directed graph $G$ undergoing deletions and edge weight increases, a root node $s$, and a depth parameter $d$, maintains, for every node $v$ a value $\ell_s(v)$ such that $\ell_s(v) = d_G(s,v)$ if $d_G(s,v) \leq d$ and $\ell_s(v) = \infty$ if $d_G(s,v) > d$.
It has constant query time and a total update time of $O(md+n)$, where $m$ is the number of edges in $G$.
\end{lemma}

\paragraph{Monotone Even-Shiloach tree (Monotone ES-tree).}
A monotone ES-tree is a generalization of ES-tree studied in \cite{HKN2016,henzinger2014decremental} that handles edge insertions. This is a data structure that maintains distances from a source $s$ in a graph undergoing deletions, insertions and edge weight increases. Note that when edges are inserted to the graph, it can potentially decrease some distances in the graph. However, the data structure ignores such decreases, and maintains distance estimates $\ell_s(v)$ from $s$ to any other node, that can only \emph{increase} during the algorithm. In particular, these distance estimates are not necessarily the correct distances in the graph, and when using this data structure the goal is to prove that these estimates still provide a good approximation for the current distances. We call $\ell_s(v)$ the level of $v$ in the monotone ES-tree rooted at $s$.

We next describe the main properties of monotone ES-trees that we use in our algorithms. A pseudo-code appears in Appendix~\ref{sec:EStrees}. The following lemma follows from \cite{HKN2016} and from the pseudo-code.

\begin{lemma}\label{obs:simple_observations_monotone_ES}
The following holds for a monotone ES-tree of a graph $H$ rooted at $s$, that is maintained up to distance $d$:
\begin{enumerate}[label=(\arabic{*})]
\item At the beginning of the algorithm, $\ell_s(v)=d_H(s,v)$ for all nodes $v$ where $d_H(s,v) \leq d$.\label{obs:beginning}
\item It always holds that $d_H(s,v) \leq \ell_s(v)$.\label{item:easy_side}
\item \label{item: observation one} The level $\ell_s(v)$ of a node $v$ never decreases.
\item When the level $\ell_s(v)$ increases to a value that is at most $d$, we have $\ell_s(v) = \min\limits_{x\in N(v)} \{\ell_s(x)+w_{H}(x,v) \}$.\label{item:change}
\end{enumerate}
\end{lemma}

\begin{lemma}[\cite{HKN2016}]\label{lm:ES-tree}
    For every $d \geq 1$, the total update time of a monotone ES-tree up to maximum level $d$ on a graph $H$ undergoing edge deletions, edge insertions, and edge weight increases is $O(\mathcal E(H)\cdot d+\mathcal W(H)+n)$, where $\mathcal E(H)$ is the total number of edges ever contained in $H$ and $\mathcal W(H)$ is the total number of updates to $H$.
\end{lemma}

\paragraph{Handling large distances.} 

In \cite{HKN14}, the authors show an algorithm that obtains near-additive approximations for decremental APSP, giving the following.

\begin{theorem}[Theorem 1.3 in \cite{HKN14}]\label{thm_large}
For any constant $0 < \epsilon < 1$ and integer $2 \leq k \leq \log{n}$, there is a $(1+\epsilon, 2(1+2/\epsilon)^{k-2})$-approximation algorithm for decremental APSP with constant worst-case query time and expected total update time of $\tilde{O}(n^{2+1/k} (37/\epsilon)^{k-1})$.
\end{theorem}

This will help us in handling \emph{large distances}, because of the following reason. Note that if we have a $(1+\epsilon,\beta)$-approximation algorithm for APSP, then for any pair of nodes at distance $\Omega(\beta/\epsilon)$ from each other, the near-additive approximation is actually a $(1+O(\epsilon))$-approximation because the additive $\beta$ term becomes negligible. Hence, we can focus our attention on estimating distances between pairs of nodes at distance $O(\beta/\epsilon)$ from each other.

\subsection{Warm-up: \texorpdfstring{$(1+\epsilon,2)$}{(1+epsilon,2)}-APSP }\label{sec:additive_warm_up}

Our goal in this section is to use a simpler algorithm based on \cite{DHZ00} that let us maintain $(1+\epsilon,2)$-APSP in total update time of $O(m^{1/2}n^{3/2+o(1)})$ in unweighted undirected graphs with $m=\Omega(n)$ edges. Such an algorithm would outperform the approach in the \Cref{sec:TZ_2+eps} when $m= \Omega( n^{5/3+o(1)} )$, and it also gives a better \emph{near-additive} approximation. We focus on obtaining an algorithm that can maintain $+2$-additive approximate distances for pairs within distance $d$ in total update time $\tilde{O}(m^{1/2}n^{3/2}d)$. This approach can be combined with Theorem~\ref{thm_large} to get our desired bounds. 

\remove{
The static version of the data structure is as follows:

\begin{itemize}
    \item Set $s_1= (\frac{m}{n})^{1/2}$. Let $V_1$ be the set of \textit{dense nodes}: $V_1:= \{ v \in V: \deg(v) \geq s_1\}$. Also let $E_2$ be the set of \textit{sparse edges}, i.e. edges with at least one endpoint with degree less than $s_1$.
    \item Construct a hitting set $D_1$ of nodes in $V_1$. 
    This means that every node $v \in V_1$ has a neighbor in $D_1$. The size of $D_1$ is $O(n \log{n}/s_1)$.
    \item Let $E^*$ be a set of size $O(n)$ such that for each $v \in V_1$, there exists $u \in D_1 \cap N(v)$ such that $\{u,v\} \in E^*$. 
    \item Store distances $D_1 \times V$ by running a BFS from each $u \in D_1$ on the input graph $G$. 
    \item For each $u \in V \setminus D_1$ store a shortest path tree, denoted by $T_u$  rooted at $u$ by running Dijkstra on $(V, E_2 \cup E^* \cup E^{D_1}_u)$, where $E^{D_1}_u$ is the set of weighted edges corresponding to distances in $(\{u\} \times D_1)$ computed in the previous step. 
\end{itemize}
}

\subsubsection{Dynamic Data Structure} 
In \Cref{sec:add_overview}, we described the static algorithm of \cite{DHZ00}. We next describe the dynamic version of the algorithm. As mentioned above, we first focus on maintaining a data structure for pairs of nodes at distance at most $d$ from each other. The algorithm is a variant of the static algorithm. At a high-level, we use ES-trees to maintain distances from nodes in $D_1$ in the graph $G$, and we use monotone ES-trees to maintain distances from nodes $u \in V \setminus D_1$ in the graph $H_u:=(V, E_2 \cup E^* \cup E^{D_1}_u)$. 
The algorithm is summarized as follows.

\begin{enumerate}
    \item \textbf{Node set $D_1$.} We define a random set $D_1$ by sampling each node to $D_1$ with probability $\frac{c\ln{n}}{s_1}$ for $s_1= (\frac{m}{n})^{1/2}$ and a constant $c \geq 2$ of our choice.
    \item \textbf{Edge set $E^*$.} For each node $v$ that has a neighbor in $D_1$, we add one edge $\{v,u\}$ to $E^*$, where $u$ is a neighbor of $v$ in $D_1$.
    \item \textbf{Edge set $E_2$.} For each node $v$ that does not have a neighbor in $D_1$, we add all the adjacent edges to $v$ to the set $E_2$.
    \item \textbf{Maintaining ES-trees from nodes in $D_1$.} For every $u \in D_1$, we maintain an ES-tree rooted at $u$ up to distance $d+2$ in the graph $G$.
    \item \textbf{Maintaining monotone ES-trees from nodes in $V \setminus D_1$.} For each node $u \in V \setminus D_1$, we maintain a monotone ES-tree rooted at $u$ up to distance $d+2$ in the graph $H_u = (V,E_2 \cup E^* \cup E_u^{D_1})$, where $E_u^{D_1}$ are all the edges of the form $\{u,v\}$ where $v \in D_1$. The weight of an edge $\{u,v\} \in E_u^{D_1}$ is $\ell_v(u)$, the level of $u$ in the ES-tree rooted at $v$ computed in the previous step.
\end{enumerate}


\paragraph{Maintaining the monotone ES-trees.}
Next, we give more details about maintaining the monotone trees from nodes $u \in V \setminus D_1$.
\begin{enumerate}
    \item Initially, we construct a shortest path tree from $u \in V \setminus D_1$ in the original graph $H_u$, up to distance $d+2$.
    \item Each time the distance estimate between $u$ and a node $x \in D_1$ increases, we update the weight of the edge $\{u,x\} \in E^{D_1}_u$. Note that we can maintain these distances exactly (up to distance $d+2$) as we compute them in the ES-tree rooted at $x \in D_1$.\label{step_dom}
    \item Each time that an edge $\{v,x\} \in E^*$ between a node $v \in V \setminus D_1$ and $x \in D_1$ is deleted from the graph $G$, we distinguish between two cases. If $v$ still has a neighbor $w \in D_1$, we first add a new edge $\{v,w\}$ to the set $E^*$ and then delete the edge $\{v,x\}$.
    Otherwise, we add to $E_2$ all the edges incident to $v$ in the current graph $G$, and then delete the edge $\{v,x\}$.\label{step_E*_ES-tree}
    \item If an edge $e \in E_2$ is deleted, we delete it from $H_u$.
\end{enumerate}

There are several challenges in maintaining the monotone ES-trees. First, during the algorithm there are insertions to the graph $H_u$, as new edges are added to $E_2$ and $E^*$, and our goal is to show that our distance estimates still give a good approximation to the actual distances. In the static version of the algorithm, the algorithm computes the exact distances in the graphs $H_u$ and this provides a good approximation for the distances. However, in our case, the monotone ES-trees are not guaranteed to keep the actual distances in $H_u$ because they ignore distance decreases that come from edge insertions, and we should prove that this does not affect the stretch of the algorithm. 

In addition, in the static version of the algorithm it is easy to argue that the graphs $H_u$ are sparse. For example, $E^*$ is clearly a set of at most $O(n)$ edges. In our case, as $H_u$ changes dynamically, we need to prove that the total number of edges added to it is small.
In particular, as each time that an edge in $\{v,x\} \in E^*$ gets deleted, we replace it by another edge, it could be the case that overall in the algorithm we add many different edges to $E^*$ which can affect the total update time of the algorithm. In addition, each time that a node $v$ no longer has a neighbor in $D_1$, we add all its adjacent edges to $E_2$, and we should argue that the number of edges added to $E_2$ is small.
We will discuss these issues in more detail later, and show that we can implement the monotone ES-trees efficiently. 
Next, we provide the stretch analysis of the algorithm.

\subsubsection{Stretch Analysis} \label{sec:add_warmup_stretch} 
For any pair of nodes $u,v$, we denote by $\ell_u(v)$ the level of $v$ in the ES-tree or monotone ES-tree rooted at $u$. We start with the following simple observation. 
\begin{claim}\label{obs:stretch_warmup}
For any $u \in V, x \in D_1$ such that $d_G(u,x) \leq d+2$, we have that $\ell_x(u)=\ell_u(x)=d_G(u,x).$
\end{claim} 

\begin{proof}
As $x \in D_1$, we maintain the exact distances from $x$ up to distance $d+2$ in the ES-tree rooted at $x$. Hence, we have $\ell_x(u)=d_G(u,x)$. If $u \in D_1$, the same argument holds also for $\ell_u(x)$. If $u \in V \setminus D_1$, in Step~\ref{step_dom} of the algorithm for maintaining the ES-tree we make sure that we also maintain the correct distance between $u$ and $x$ in the monotone ES-tree rooted at $u$, hence $\ell_u(x)=d_G(u,x)$. Note that as the edge $\{u,x\}$ always appears in $H_u$, and by part~\ref{item:change} of Lemma~\ref{obs:simple_observations_monotone_ES} we always maintain the exact distance estimate in this case.
\end{proof}

\begin{claim}\label{claim:stretch_warmup}
For any $u,v \in V$ such that $d_G(u,v) \leq d$, we have that $d_G(u,v) \leq \ell_u(v) \leq d_G(u,v) +2.$
\end{claim}

\begin{proof}
If $v \in D_1$, this already follows from Claim~\ref{obs:stretch_warmup}. Hence we focus on the case that $v \in V \setminus D_1$. The left side of the inequality follows as $d_G(u,v) \leq d_{H_u}(u,v) \leq \ell_u(v)$, using part~\ref{item:easy_side} of Lemma~\ref{obs:simple_observations_monotone_ES}, and since the distances in $H_u$ can only increase compared to distances in $G$. This holds as the edges in $H_u$ are either edges in $G$ or have a weight that is at least the weight of the shortest path between the corresponding endpoints in $G$.
We next show that $\ell_u(v) \leq d_G(u,v) +2$.
We prove the claim by double-induction. First, on the time in the algorithm, and second on the distance $d_G(u,v)$.

\paragraph{Base Case.} At the beginning of the algorithm, this follows from the correctness of the static algorithm, we give the details for completeness. Note that at the beginning of the algorithm, from part~\ref{obs:beginning} of Lemma~\ref{obs:simple_observations_monotone_ES}, we have that $\ell_u(v)=d_{H_u}(u,v)$ as long as $d_{H_u}(u,v) \leq d+2$, and we show that this is indeed the case when $d_G(u,v) \leq d$. 

\begin{itemize}
    \item [Case 1:] $v$ has a neighbor in $D_1$. In this case there is an edge $\{v,x\} \in E^*$ such that $x \in D_1$. Note that as $x$ is a neighbor of $v$ it holds that $d_G(u,x) \leq d+1$. As $x \in D_1$ and $d_G(u,x) \leq d+1$, the graph $H_u$ has an edge $\{u,x\}$ with weight $d_G(u,x)$, and it also has the edge $\{v,x\} \in E^*$. Hence, at the beginning of the algorithm, we have $$\ell_u(v) = d_{H_u}(u,v) \leq d_G(u,x) + d_G(v,x) = d_G(u,x)+1.$$ Now, by the triangle inequality $$d_G(u,x) \leq d_G(u,v)+d_G(v,x)=d_G(u,v)+1.$$ Hence, overall, we get that $\ell_u(v) \leq d_G(u,v) +2$. Note that we get the $d_{H_u}(u,v) \leq d+2$, which implies that the distance is indeed maintained correctly in the monotone ES-tree.
    \item [Case 2:] $v$ does not have a neighbor in $D_1$. In this case, all the edges adjacent to $v$ are in $E_2 \subseteq H_u$. Let $\{w,v\}$ be the edge adjacent to $v$ in the shortest path between $u$ and $v$ in $G$. Then, $d_G(u,w) = d_G(u,v)-1$. From the induction hypothesis, we have that $\ell_u(w) \leq d_G(u,w) +2$. Since the edge $\{w,v\}$ exists in $H_u$, we get that $$\ell_u(v) = d_{H_u}(u,v) \leq d_{H_u}(u,w) + d_{H_u}(w,v)= \ell_u(w)+1 \leq d_G(u,w)+3 = d_G(u,v) + 2.$$ As $d_{H_u}(u,v) \leq d+2$, we indeed maintain the value correctly in the monotone ES-tree.
\end{itemize}

\paragraph{Induction Step.} Assume now that the claim holds after $t-1$ edge deletions, and we will prove it holds after $t$ edge deletions. We prove it by induction on the length of the path between $u$ and $v$ in the current graph $G$. In the base case that $u=v$ it trivially holds.
First, if the value $\ell_u(v)$ did not change after the $t$-th edge deletion, then the claim follows from the induction hypothesis (as the distance $d_G(u,v)$ can never decrease). Hence, we focus on the case that $\ell_u(v)$ changed. By parts~\ref{item: observation one} and~\ref{item:change} of Lemma~\ref{obs:simple_observations_monotone_ES}, when the level $\ell_u(v)$ is changed it is necessarily increased, and the new value is equal to $\min\limits_{x\in N(v)} \{\ell_u(x)+w_{H_u}(x,v) \}$. 

\begin{itemize}
    \item [Case 1:] When the value $\ell_u(v)$ increases, $v$ has a neighbor in $D_1$. By Step~\ref{step_E*_ES-tree} of the algorithm for maintaining the ES-tree, it is guaranteed that there is an edge $\{v,x\} \in E^* \subseteq H_u$. Also, $\ell_u(x)=d_G(u,x)$ by Claim~\ref{obs:stretch_warmup} and since $d_G(u,x) \leq d+1$, as $x$ is a neighbor of $v$. Hence, when we update the value $\ell_u(v)$, the new value is at most $\ell_u(x)+w_{H_u}(x,v)=d_G(u,x)+1.$ The exact same proof that we used in the base case, shows that $d_G(u,x) \leq d_G(u,v)+1$. Hence, $\ell_u(v) \leq d_G(u,v) +2$, as needed. Again, as this value is bounded by $d+2$ it is maintained correctly.
    \item [Case 2:] When the value $\ell_u(v)$ increases, $v$ does not have a neighbor in $D_1$. It follows from Step~\ref{step_E*_ES-tree} of the algorithm for maintaining the ES-tree that all the edges adjacent to $v$ are in $H_u$. Let $x$ be the neighbor of $v$ in the shortest $u-v$ path in $G$. By the induction hypothesis, $\ell_u(x) \leq d_G(u,x)+2$. Hence, $$\ell_u(v) \leq \ell_u(x)+w_{H_u}(x,v) \leq d_G(u,x)+3 = d_G(u,v)+2,$$ this value is again at most $d+2$, hence it is maintained correctly, completing the proof.\qedhere
\end{itemize}
\end{proof}




\subsubsection{Time Complexity} \label{sec:warmup_add_time}

To analyze the time complexity of the algorithm, we first bound the total number of edges added to $E_2$ and $E^*$ in the algorithm. 

\begin{claim} \label{claim_E2}
The total number of edges added to $E_2$ during the algorithm is bounded by $O(ns_1)$ with high probability.
\end{claim}

\begin{proof}
By definition, the edges in $E_2$ are the edges adjacent to all nodes that do not have a neighbor in $D_1$. We show that as long as the degree of a node $v$ is $s_1$ then with high probability it has a neighbor in $D_1$. This follows as $D_1$ is obtained by sampling each node with probability $\frac{c \ln{n}}{s_1}$. Hence, as long as the degree of $v$ is at least $s_1$, the probability that none of its neighbors is in $D_1$ is at most $(1-\frac{c\ln{n}}{s_1})^{s_1} \leq e^{-c\ln{n}}=\frac{1}{n^c}$. It follows that with high probability each node adds less than $s_1$ edges to $E_2$, proving the claim.
\end{proof}

\begin{claim} \label{claim_warmup_E*}
The number of edges added to $E^*$ during the algorithm is bounded by $O( (mn)^{1/2}\log{n} )$ in expectation.
\end{claim}

\begin{proof}
Recall that each node is added to the set $D_1$ with probability $\frac{c \ln{n}}{s_1}$. Hence, in expectation each node $v$ has $O(\deg(v) \frac{\log{n}}{s_1})$ neighbors in the set $D_1$, where $\deg(v)$ is the degree of $v$ at the beginning of the algorithm. These edges are candidates to be added to the set $E^*$, and in the worst case all of them are added to $E^*$ during the algorithm. Summing over all nodes, the expected number of edges added to $E^*$ is bounded by 
\begin{align*}
    \sum_{v \in V} O \left( \frac{\deg(v)\log{n}}{s_1} \right) = O \left( \frac{m\log{n}}{(\frac{m}{n})^{1/2}} \right) = O( (mn)^{1/2}\log{n} ).
\end{align*} 
\end{proof}

\begin{lemma}
The total update time of the algorithm is $\tilde{O}(m^{1/2}n^{3/2}d)$ in expectation.
\end{lemma}

\begin{proof}
Computing the set $D_1$ takes $O(n)$ time by sampling each node to the set independently with probability $\frac{c \ln{n}}{s_1}$. 

Maintaining the sets $E^*$ and $E_2$ is done as follows. At the beginning of the algorithm, each node $v$ scans its neighbors until it finds the first neighbor $u$ in the set $D_1$ if exists, and it adds the edge $\{v,u\}$ to the set $E^*$. If there is no such $u$, then $v$ adds all its adjacent edges to the set $E_2$.
During the algorithm, when the edge $\{v,u\}$ is deleted, $v$ continues scanning its neighbors until it finds the next neighbor $w$ in the set $D_1$ if exists. If $w$ exists it adds the edge $\{v,w\}$ to $E^*$, and otherwise it adds all its adjacent edges to $E_2$. Note that in the whole algorithm, $v$ only needs to scan its neighbors once to find all the edges it adds to $E^*$. Hence, overall maintaining the sets $E^*$ and $E_2$ takes $O(m)$ time. We can also keep a list of the new edges added to $E^*$ and $E_2$ after a deletion, to help maintaining the monotone ES-trees.

We maintain the ES-trees from nodes $u \in D_1$ as follows. Every time that an edge is deleted we delete it from all the trees. In addition, every time that a level $\ell_u(v)$ increases for $u \in D_1, v \in V \setminus D_1$, we update the weight of the edge $\{v,u\}$ in the graph $H_v$, and update it in the monotone ES-tree rooted at $v$.
Maintaining the ES-trees from nodes in $D_1$ takes $O(|D_1|(md+n))$ time by Lemma~\ref{lemma_ES}. By the definition of $D_1$, this equals $$O \left( \frac{n \log{n}}{s_1} \left( md \right) + n^2 \right) =O(m^{1/2}n^{3/2} (\log{n}) d)$$ in expectation, when $m=\Omega(n).$

We maintain the monotone ES-trees rooted at nodes $u \in V \setminus D_1$ as follows. After an edge $e$ is deleted, we update the weights of the edges $\{u,v\} \in E_u^{D_1}$ when their weights are increased in the ES-trees rooted at nodes $v \in D_1$, as described above. In addition, we add to $H_u$ all the edges that were added to $E_2$ and $E^*$ after the deletion of the edge $e$ (we kept them in a list as described above), and finally we delete the edge $e$ from the monotone ES-tree. By Lemma~\ref{lm:ES-tree}, the total update time for maintaining the monotone ES-tree rooted at the node $u$ is 
$O(\mathcal E(H_u)\cdot d+\mathcal W(H_u)+n)$, where $\mathcal E(H_u)$ is the total number of edges ever contained in $H_u$ and $\mathcal W(H_u)$ is the total number of updates to $H_u$. There are $O(n)$ edges in $E_u^{D_1}$ and each of them is updated $O(d)$ times, as the distance estimates only increase and can have $O(d)$ values. This adds $O(nd)$ term to the complexity of one tree, and $O(n^2 d)$ summing up over all trees. The edges in $E_2$ and $E^*$ are only added and removed once in the algorithm. The total number of edges in $E_2$ is $O(ns_1)$ with high probability from Claim~\ref{claim_E2}. Summing up over all trees these edges add $O(n^2 s_1 d)= O(n^{3/2} m^{1/2} d)$ term to the complexity. The total number of edges in $E^*$ is bounded by $O( (mn)^{1/2}\log{n} )$ in expectation from Claim~\ref{claim_warmup_E*}. Summing up over all trees they add $O(n^{3/2} (\log{n}) m^{1/2} d)$ term to the complexity, completing the proof.
\end{proof} 

\begin{claim}
The query time of the algorithm is $O(1)$.
\end{claim}

\begin{proof}
To answer a distance query for a pair of nodes $u,v$, we just return the value $\ell_u(v)$. This can be done in $O(1)$ time if we keep for each node $u$ an array of the values $\ell_u(v)$ which is updated each time a distance estimate change.
\end{proof}

In conclusion, we get the following.

\begin{theorem}
Given an undirected unweighted graph $G$, there is a decremental data structure that maintains $+2$-additive approximation for APSP for pairs of nodes at distance at most $d$ from each other. The expected total update time is bounded by $\tilde{O}(m^{1/2}n^{3/2}d)$, and the query time is constant.
\end{theorem}

To get rid of the $d$ term in the complexity, we can combine our approach with Theorem~\ref{thm_large}. This would lead to a $(1+\epsilon,2)$-approximation and replace the $d$ term in the complexity by $n^{o(1)}$ term. We will provide full details of the approach when we describe our general algorithm in the next section (see in particular \Cref{sec_summary_additive}). The high-level idea is to use our algorithm with a parameter $d=n^{o(1)}$, and combine it with Theorem~\ref{thm_large} for computing $(1+\epsilon,\beta)$-approximation for large distances. For an appropriate choice of the parameters, we will get a $(1+\epsilon)$-approximation for the distances between pairs of nodes at distance greater than $d$ from each other, and a $+2$ approximation for pairs of nodes at distance at most $d$ from each other.

\subsection{\texorpdfstring{$(1+\epsilon, 2k-2)$}{(1+epsilon,2k-2)}-APSP}\label{sec:add_general}

With this warm-up we can move on to the more general algorithm based on \cite{DHZ00} that leads to a larger stretch but a smaller update time. Our goal is to obtain a $(1+\epsilon, 2k-2)$-approximation in $\tilde{O}(n^{2-1/k}m^{1/k})$ time. The algorithm works in  unweighted undirected graphs. As before, we start by describing a data structure that maintains the approximate distances for pairs of nodes at distance at most $d$ from each other, and then in \Cref{sec_summary_additive} we combine it with Theorem \ref{thm_large} to handle the general case.

\remove{
Our goal is to match the following static bound (up to $n^{o(1)}$):
\begin{theorem}[\cite{DHZ00}]
For every $2 \leq k \leq O(\log n)$ there is a $\tilde{O}(n^{2-1/k}m^{1/k})$ time algorithm for computing APSP distance estimates $\hat{d}$ such that for all $u,v \in V$: 
\[d_G(u,v) \leq \hat{d} (u,v) \leq d_G (u,v)+2(k-1). \]

\end{theorem}

The summary of the static algorithm is as follows:
\begin{enumerate}
    \item Set degree thresholds $s_1, s_2, ...,s_{k-1}$, where $d_i=(\frac{m}{n})^{1-i/k} (\log n)^{i/k}$. Let $E_i$ be set of edges incident to a node with degree $\leq d_{i-1}$. 
    \item Sample hitting sets $D_1, D_2, ..., D_k$\ttodo{To do: describe the sampling} and let $E_1^*, \dots, E_k^*$ be sets of edges such that for all $1 \leq i < k$ and each node with degree in $(d_{i-1}, d_i]$ has one edge incident to $D_i$ in $E_i^*$.
    \item For $i=1$ to $k$:
    \begin{itemize}
        \item For every $u \in D_i$ compute SSSP on the graph $G_{u,i}:=(V, E_i \cup E_i^* \cup (\{u\} \times V))$ and update the distance estimates $\hat{d}$.
    \end{itemize}
\end{enumerate}

Note that this algorithm ensures that any node that is not in $E_i$ has degree more than $d_{i-1}$ and thus has neighbor in $D_i$.
}

\subsubsection{Dynamic Data Structure} 
The algorithm is summarized as follows.
\begin{enumerate}
    \item \textbf{Node sets $D_i$.} We define a series of sets $D_1,D_2,...,D_k$ such that for $i=1,...,k-1$, the set $D_i$ is a random set obtained by sampling each node with probability $p_i=\frac{c\ln{n}}{s_i}$, for $s_i=(\frac{m}{n})^{1-i/k} (\ln{n})^{i/k}$, and a constant $c \geq 2$ for our choice. If a node is sampled to several different sets we keep it only in the set with minimum index $i$ among them. Finally, we define $D_k = V \setminus \bigcup_{i=1}^{k-1} D_i.$ Note that by definition, each node is in exactly one set $D_i$.\label{step_Di} 
    \item \textbf{Indices $i_v$ for all $v \in V$.} For a node $v \in V$, we denote by $i_v$ the minimum index $i$ such that $v$ has a neighbor in $D_i$.\label{step_Vi} 
    \item \textbf{Edge set $E^*$.} For each node $v$, we add one edge $\{v,x\}$ to $E^*$, where $x$ is a neighbor of $v$ in $D_{i_v}$.\label{step_E*}
    \item \textbf{Edge sets $E_i$.} The edges $E_i$ are all the edges $\{u,v\}$ where $\max\{i_u,i_v\} \geq i$. In other words, it is all the edges adjacent to nodes that do not have a neighbor in $D_{i'}$ for $i' < i$.\label{step_Ei}
    \item \textbf{Maintaining monotone ES-trees.}\label{step_ES} The algorithm maintains monotone ES-trees as follows.
    \begin{itemize}
        \item For every $u \in D_1$, we maintain an ES-tree rooted at $u$ up to distance $d+3k$ in the graph $G$. 
    \end{itemize}
    For $i= 2$ to $k$:
    \begin{itemize}
        \item  For every $u \in D_i$, we maintain a monotone ES-tree rooted at $u$ up to distance $d+3k$ in the graph $H_u^i=(V,E_i \cup E^* \cup E_u^i)$, where $E_u^i$ are all the edges $\{u,v\}$, where $v \in D_{i'}$ for $i' < i$. The weight of the edge $\{u,v\}$ is equal to the value $\ell_v(u)$ in the monotone ES-tree rooted at $v$.
    \end{itemize}
\end{enumerate}

Note that we maintain the monotone ES-trees up to distance $d+3k$. As we show in the stretch analysis, this is needed in order to guarantee that we maintain correct estimates for any pair of nodes at distance at most $d$ from each other. Intuitively, as the distance estimate we maintain is at most $d(u,v)+2k$, we would like to maintain correct estimates for pairs of nodes at distance $d+2k$ from each other. Another $k$ term comes from the inductive argument of the stretch analysis. For more details, see \Cref{sec_additive_stretch}.

In \Cref{sec:add_complexity}, we explain in detail how to maintain the data structures efficiently. When we maintain the graphs $H_u^i$ we make sure that the following holds (for a proof, see \Cref{sec:add_complexity}). 

\begin{claim}\label{claim_Hui}
Let $u \in D_i, v \in V$. Throughout the updates, it always holds that $v$ either has an edge $\{v,w\} \in H_u^i$, where $w \in D_{i'}$ for $i' < i$, or all the edges adjacent to $v$ are in $H_u^i$. 
\end{claim}

We are in the first case of the claim when $i_v <i$, and we are in the second case when $i_v \geq i$.

\subsubsection{Stretch Analysis} \label{sec_additive_stretch}

We need the following claims. First, from part~\ref{obs:beginning} of Lemma~\ref{obs:simple_observations_monotone_ES} we have the following.

\begin{claim} \label{claim_beginning}
At the beginning of the algorithm, for any $u \in D_i$, and any $v$ at distance at most $d+3k$ from $u$ in $H_u^i$, we have $\ell_u(v)=d_{H_u^i}(u,v).$
\end{claim}

We next show the following. 

\begin{claim} \label{claim_dist_smaller_dom}
Let $u \in D_i, w \in D_j$ such that $j < i$, then $\ell_u(w) \leq \ell_w(u)$. 
\end{claim}

\begin{proof}
The claim follows as we always have an edge $\{u,w\}$ with weight $\ell_w(u)$ in the graph $H_u^i$, and the value $\ell_w(u)$ is monotonically increasing during the algorithm. Hence, and from part~\ref{item:change} of Lemma~\ref{obs:simple_observations_monotone_ES},  every time that the estimate $\ell_u(w)$ changes it always changes to a value that is at most $\ell_w(u)$.
\end{proof}


We next prove the stretch guarantee of the algorithm.


\begin{lemma}\label{lem_additive_stretch}
For any node $u \in D_i$, and any node $v \in V$ where $d_G(u,v) \leq d+(k-i)$, throughout the updates we have $d_G(u,v) \leq \ell_u(v) \leq  d_G(u,v)+ 2(i-1)$.
\end{lemma}

\begin{proof}
First, the left inequality holds since $d_G(u,v) \leq d_{H_u^i}(u,v) \leq \ell_u(v)$, where the first inequality follows from the definition of $H_u^i$, and the second follows from part~\ref{item:easy_side} of Lemma~\ref{obs:simple_observations_monotone_ES}.
We prove the right inequality by triple-induction, on the value $i$, on the time in the algorithm, and on the distance $d_G(u,v)$.
For $i=1$, the claim follows as we maintain ES-trees up to distance $d+3k$ from the nodes of $D_1$ in the graph $G$. We next assume that it holds for any $i' < i$, and prove that it holds for $i$. The proof is by induction on the time in the algorithm and on the distance $d_G(u,v)$.

\paragraph{Base Case.} We first show that the claim holds at the beginning of the algorithm. From Claim~\ref{claim_beginning}, at the beginning of the algorithm, for any $u,v$ where $d_{H_u^i}(u,v) \leq d+3k$, we have that $\ell_u(v) = d_{H_u^i}(u,v).$ Hence, our goal is to show that $d_{H_u^i}(u,v) \leq d_G(u,v) + 2(i-1)$. Note that if this holds, then since $d_G(u,v) \leq d+(k-i)$, we get that $d_{H_u^i}(u,v) \leq d+3k$, hence the distance is maintained in the monotone ES-tree as needed.
The proof is by induction on $d_G(u,v)$.

\begin{itemize}
    \item [Case 1:] $i_v < i$. In this case, there is an edge $\{v,w\} \in E^* \subseteq H_u^i$ such that $w \in D_{i'}$ for $i' < i$. Note that since $d_G(u,v) \leq d+(k-i)$, and $w$ is a neighbor of $v$, then we have $d_G(u,w) \leq d +(k-i)+1 = d+(k-(i-1)) \leq d+(k-i')$. Hence, we can use the induction hypothesis on $w$ and $u$. From the induction hypothesis and Claim~\ref{claim_beginning}, we have 
    $\ell_w(u) = d_{H_u^i}(w,u) \leq d_G(w,u)+2(i-2)$ (note that this value is smaller than $d+3k$, so we can indeed use Claim~\ref{claim_beginning}). 
    Hence, we get that $$d_{H_u^i}(u,v) \leq d_{H_u^i}(u,w) + d_{H_u^i}(w,v) \leq d_G(w,u) + 2(i-2)+1.$$ We have, $d_G(w,u) \leq d_G(w,v)+d_G(v,u)=d_G(u,v)+1$.
    Hence, we get, $d_{H_u^i}(u,v) \leq d_G(u,v)+2(i-1)$, as needed.
    \item [Case 2:] $i_v \geq i$. In this case, by definition, all the edges adjacent to $v$ are in $E_i \subseteq H_u^i$. Let $w$ be the neighbor of $v$ on the shortest $u-v$ path. From the induction hypothesis, $d_{H_u^i}(u,w) \leq d_G(u,w)+2(i-1)$. So we have $$d_{H_u^i}(u,v) \leq d_{H_u^i}(u,w)+d_{H_u^i}(w,v) \leq d_G(u,w)+1+2(i-1)=d_G(u,v)+2(i-1).$$
\end{itemize}

\paragraph{Induction Step.} We assume that the claim holds after $t-1$ edge deletions, and show that it holds after $t$ edge deletions. The proof is by induction on the distance $d_G(u,v)$. If $u=v$, the claim trivially holds. If the estimate $\ell_u(v)$ did not change after the $t$'th edge deletion, then the claim holds from the induction hypothesis. Note that an edge deletion can potentially lead to several changes of the value $\ell_u(v)$. It is enough to prove that the stretch guarantee holds in the last time the value $\ell_u(v)$ increases, as the value stays the same afterwards. We divide into cases based on Claim~\ref{claim_Hui}.

\begin{itemize}
    \item [Case 1:] When the value $\ell_u(v)$ increases in the last time, there is an edge $\{v,w\} \in H_u^i$ such that $w \in D_{i'}$ for $i' < i$. 
    From part~\ref{item:change} of Lemma~\ref{obs:simple_observations_monotone_ES}, when $\ell_u(v)$ changes, it gets the value $\min\limits_{x\in N(v)} \{\ell_u(x)+w_{H_u^i}(x,v) \}$, assuming this value is at most $d+3k$, and we will show that this is indeed the case. Hence, $$\ell_u(v) \leq \ell_u(w)+w_{H^i_u}(w,v) \leq \ell_w(u) + 1 \leq d_G(u,w)+2(i-2)+1.$$ Here we use Claim~\ref{claim_dist_smaller_dom} and the fact that $w \in D_{i'}$ for $i' < i$ (note that since $d_G(u,v) \leq d+(k-i)$, and $w$ is a neighbor of $v$, it follows that $d_G(u,w) \leq d+(k-i')$, so we can indeed use the induction hypothesis). It also holds that $d_G(u,w) \leq d_G(u,v)+d_G(v,w)=d_G(u,v)+1$. Hence, overall, we have $\ell_u(v) \leq d_G(u,v) + 2(i-2)+2=d_G(u,v)+2(i-1)$, as needed. Note that this value is indeed smaller than $d+3k$, hence we maintain the correct estimate in the monotone ES-tree.
    \item [Case 2:] When the value $\ell_u(v)$ increases in the last time, all the adjacent edges to $v$ are in the graph $H_u^i$.
    Let $w$ be the neighbor of $v$ on the shortest $u-v$ path. When the distance estimate of $v$ changes we have $$\ell_u(v) \leq \ell_u(w)+d_{H_u^i}(w,v) \leq d_G(u,w) + 2(i-1) + 1 = d_G(u,v)+2(i-1),$$ where we use the induction hypothesis on $w$ and part~\ref{item:change} of Lemma~\ref{obs:simple_observations_monotone_ES}. Note that it follows that the values $\ell_u(v),\ell_u(w)$ are smaller than $d+3k$, hence they are maintained correctly in the monotone ES-tree.\qedhere
\end{itemize}
\end{proof}

\subsubsection{Time Complexity} \label{sec:add_complexity}

We next analyze the time complexity of the algorithm. We start by showing that the first steps of the algorithm can be maintained efficiently. 

\begin{claim}
The first 4 steps of the algorithms can be maintained in $O((m+n)k)$ time.
\end{claim}

\begin{proof}
\textbf{Step~\ref{step_Di}.} Computing the sets $D_i$ takes $O(kn)$ time, by just sampling each node to the set $D_i$ with probability $p_i$. For each node $v$, we can stop the sampling after the first time $v$ is sampled to one of the sets. This is only done once at the beginning of the algorithm. For each node $v$, we keep the index $i$ such that $v \in D_i$. For each $i$, we create a list of all the nodes in the set $D_i$. Creating all these lists takes $O(n)$ time by scanning the nodes.

\textbf{Steps~\ref{step_Vi} and~\ref{step_E*}.} The indices $i_v$ and the set $E^*$ are maintained dynamically in the algorithm, as follows. At the beginning, each node $v$ scans its neighbors to find the minimum index $i$ such that $v$ has a neighbor in $D_i$, it also saves the name of an edge $\{v,u\}$ to the first neighbor of $v$ in $D_{i_v}$, this edge is added to set $E^*$. In the algorithm, if the edge $\{v,u\}$ is deleted, then $v$ goes over its next neighbors, and finds the next neighbor $w \in D_{i_v}$ if exists. If such a neighbor exists, then $\{v,w\}$ is added to $E^*$. Otherwise, the value $i_v$ increases, and $v$ scans its neighbors again to determine the new index $i_v$ and adds the first edge $\{v,w\}$ from $v$ to a neighbor in $D_{i_v}$ to $E^*$. Note that overall in the algorithm each node needs to scan its neighbors $O(k)$ times, because as long as the value $i_v$ stays the same, we just need to go over the neighbors once during the whole algorithm to determine all the edges added to $E^*$. Overall these steps take $O(mk)$ time. 

\textbf{Step~\ref{step_Ei}.} Maintaining the sets $E_i$ is done as follows. Initially, we add to $E_i$ all the edges adjacent to nodes $v$ where $i_v \geq i$. During the algorithm, every time that the value $i_v$ of a node $v$ changes, say from $j$ to $j'$, then for all $j < i \leq j'$, we add the edges adjacent to $v$ to the set $E_i$. We also keep a list of all the newly added edges in $E_i$ after a deletion of an edge, so that we can update them in all the relevant monotone ES-trees. The time complexity for maintaining the set $E_i$ is proportional to the number of edges added to $E_i$ during the algorithm, which can trivially be bounded by $O(m)$, hence this step takes at most $O(mk)$ time.  
\end{proof}

\paragraph{Maintaining the monotone ES-trees.} 

The main challenge in the algorithm is maintaining the monotone ES-trees in Step~\ref{step_ES} of the algorithm. Before analyzing the time complexity of the algorithm, we start by describing the algorithm in more detail. Note that each time an edge is deleted it can lead to many changes, including edge insertions and weight changes in all the trees we maintain. We next explain in detail in which order we make these changes in order to maintain the correctness of the algorithm. We use the algorithms for ES trees and monotone ES trees described in \Cref{sec:additive_tools}.

\subparagraph{ES-trees rooted at nodes of $D_1$.} 
First, for any node $u \in D_1$, we start by creating an ES-tree rooted at $u$ in the graph $G$, maintained up to distance $d+3k$. After creating this tree, for any $v \in D_i$ for $i > 1$, we add an edge $\{u,v\}$ with weight $\ell_u(v)$ to the set of edges $E_v^i$. During the algorithm, every time an edge is deleted, we delete it from the ES-tree and update it accordingly. For any $v$ where $\ell_u(v)$ has increased after the deletion and where $v \in D_i$ for $i > 1$, we update the corresponding edge in $E_v^i$, and save it in a list of edges in $E_v^i$ that were updated after the deletion. This completes the description for $u \in D_1$. 
\subparagraph{Monotone ES-trees rooted at nodes of $D_i$.}
For $u \in D_i$ for larger values of $i$, we maintain the monotone ES-trees according to the order of the index $i$, from 2 to $k$. We work as follows. Initially, we create a monotone ES-tree rooted at $u$ up to distance $d+3k$ in the graph $H_u^i=(V,E_i \cup E^* \cup E_u^i)$, with the edges $E_i,E^*$ at the beginning of the algorithm, and with the edges $E_u^i$ that were obtained when creating the monotone ES-trees in previous levels. Note that all these edges of $E_u^i$ are of the form $\{u,v\}$ where $v \in D_{i'}$ for $i' < i$, so we indeed already created the monotone ES-tree rooted at the relevant nodes $v$.  After creating the tree, for any node $v \in D_{i'}$ for $i' > i$, we add a new edge to the set $E_v^{i'}$ with weight $\ell_u(v)$.
During the algorithm, when an edge $e$ is deleted, we update the monotone ES-trees in several steps as follows, again we work according to the order of the index $i$.
\begin{itemize}
    \item \textit{First step: adding the new edges of $E^*$ and $E_i$ to $H_u^i$.} The deletion of the edge $e$ can lead to adding new edges to $E^*$ or $E_i$ in steps~\ref{step_E*} and~\ref{step_Ei} of the algorithm. As explained in these steps, we maintain the newly added edges, and can add them to a list of newly added edges (to $E^*$ and to $E_i$ separately). When we update the monotone ES-tree rooted at $u \in D_i$, we start by adding all these newly added edges that were added to $E^*$ and $E_i$ to $H_u^i$.
    \item \textit{Second step: updating the edges in $E_u^i$.} Next, there are possibly edges in $\{u,v\} \in E_u^i$ that their weight is increased after the deletion. This is already observed when we update the monotone ES-tree rooted at $v$, and if it happens, we add the edge $\{u,v\}$ with new weight $\ell_v(u)$ to a list of edges that were updated in $E_u^i$. When we update the monotone ES-tree rooted at $u$, we update the weights of all these edges (we already have their list after we finish maintaining the ES-trees in previous levels).
    \item \textit{Third step: deleting the edge $e$.} Lastly, we delete the edge $e$ that was deleted in this step from the monotone ES-tree. 
\end{itemize}

Note that the weight increases and deletion can increase the levels $\ell_u(v)$ of some nodes in the tree. For any $v$ in $D_j$ for $j > i$ such that $\ell_u(v)$ increases, when the value increases, we add an edge $\{v,u\}$ with the new weight $\ell_u(v)$ to the list of newly updated edges in $E_v^j$. This completes the description of the algorithm.
From the above description, we can now prove Claim~\ref{claim_Hui}. 

\begin{proof}[Proof of Claim~\ref{claim_Hui}]
Let $u \in D_i, v \in V$. At the beginning of the algorithm, if $i_v < i$, then there exists an edge $\{v,w\} \in E^* \subseteq H_u^i$ where $w \in D_{i'}$ for $i' < i$. Otherwise, $i_v \geq i$, and all the edges adjacent to $v$ are in $E_i \subseteq H_u^i$. During the algorithm, every time an edge $e$ of the form $\{v,w\} \in E^*$ is deleted, then if $i_v < i$ after the deletion, we add a new edge $\{v,w'\}$ to $E^*$ for $w' \in D_{i'}$ for $i' < i$, and if $i_v \geq i$, we add all the edges adjacent to $v$ to $E_i \subseteq H_u^i$. In the algorithm, we first add these new edges to $H_u^i$, and only then delete the edge $e$, which proves the claim.
\end{proof}

To bound the time complexity, we need the following useful claims that bound the total number of edges added to the sets $E_i$ and $E^*$ during the algorithm.

\begin{claim}\label{claim_Ei}
For $i \geq 2$, the total number of edges added to $E_i$ during the algorithm is bounded by $O(ns_{i-1})$ with high probability.
\end{claim}

\begin{proof}
The edges $E_i$ are all the edges adjacent to nodes $v$ where $i_v \geq i$. We show that at the first time in the algorithm that $i_v \geq i$, the degree of $v$ is at most $s_{i-1}$ with high probability. This implies that each node $v$ adds at most $s_{i-1}$ edges to the set $E_i$, proving the claim. 
To bound the degree of $v$, note that by definition the set $D_{i-1}$ is obtained by sampling each node with probability $\frac{c \ln{n}}{s_{i-1}}$, all these nodes are added to the set $D_{i-1}$, unless they are already part of a set $D_{i'}$ for $i' < i-1$. As long as the degree of $v$ is at least $s_{i-1}$, the probability that none of the neighbors of $v$ is sampled to $D_{i-1}$ is at most $(1-\frac{c\ln{n}}{s_{i-1}})^{s_{i-1}} \leq e^{-c\ln{n}}=\frac{1}{n^c}$. Hence, with high probability $v$ has a neighbor sampled to $D_{i-1}$, implying that $i_v \leq i-1$ as long as $\deg(v) \geq s_{i-1}$. From union bound, this holds for all nodes with high probability. This proves that the total number of edges added to $E_i$ in the algorithm is bounded by $O(ns_{i-1})$. Note that the proof assumes that we have an oblivious adversary and that the order in which edges are deleted does not depend on the random choices of the algorithm.
\end{proof}

\begin{claim}\label{claim_E*}
In expectation, the total number of edges added to the set $E^*$ in the algorithm is bounded by $O(kn^{1-1/k}m^{1/k}(\log{n})^{1-1/k})$.
\end{claim}

\begin{proof}
We analyze the number of edges $v$ adds to the set $E^*$ when $i_v=i$ for all $i$. First, when $i_v=1$, then in the worst case $v$ adds to $E^*$ edges to all its neighbors in $D_1$. By the definition of $D_1$, each node is added to $D_1$ with probability $\frac{c\ln{n}}{s_1}$. Hence, in expectation $v$ has $O(\frac{\deg(v) \log{n}}{s_1})$ neighbors in $D_1$, where $\deg(v)$ is the degree of $v$ at the beginning of the algorithm. When we sum up over all $v$, this adds at most $O(\frac{m \log{n}}{s_1})=O(n^{1-1/k}m^{1/k}(\log{n})^{1-1/k})$ edges in expectation, where the last equality follows from the definition of $s_1$.

When $i_v=i$ for $i > 1$, as we proved in Claim~\ref{claim_Ei}, with high probability the degree of $v$ is at most $s_{i-1}$. Note that each one of the neighbors of $v$ is added to $D_i$ with probability $\frac{c\ln{n}}{s_i}$. Hence, in expectation, $v$ has $O(s_{i-1} \cdot \frac{\ln{n}}{s_i})$ neighbors in $D_i$. In the worst case, it adds edges to all of them during the algorithm. In total, the expected number of edges added by all nodes to $E^*$ in the case that $i_v \geq 2$ is bounded by $\sum_{i=2}^{k} O(\frac{\ n \log{n} s_{i-1}}{s_i})$.
By the definition of $s_i$, we have that $\frac{s_{i-1}}{s_i}=(\frac{m}{n})^{1/k}(\ln{n})^{-1/k}$. Hence, in total the expected number of edges added is bounded by $O(k n^{1-1/k}m^{1/k}(\log{n})^{1-1/k}).$
\end{proof}

Next, we prove the complexity of the algorithm. 

\begin{lemma}\label{lemma_time_additive}
Maintaining the ES-trees takes $O(kn^{2-1/k}m^{1/k}(\log{n})^{1-1/k} d)$ time in expectation.
\end{lemma}

\begin{proof}
Maintaining the ES-trees rooted at nodes in $D_1$ takes $O(|D_1|(md+n))$ time by Lemma~\ref{lemma_ES}. 
Since $|D_1| = O(\frac{n\log{n}}{s_1})$ in expectation, and by the choice of $s_1$, this adds $O(n^{2-1/k}m^{1/k}(\log{n})^{1-1/k} d)$ term to the complexity, when $m=\Omega(n)$. 

From Lemma~\ref{lm:ES-tree}, maintaining a monotone ES-tree rooted at $u \in D_i$ takes $O(\mathcal E(H_u^i)\cdot d+\mathcal W(H_u^i)+n)$ time, where $\mathcal E(H_u^i)$ is the total number of edges ever contained in $H_u^i$ and $\mathcal W(H_u^i)$ is the total number of updates to $H_u^i$. To complete the proof, we should analyze these values. Recall that $H_u^i=(V, E_i \cup E^* \cup E_u^i)$. Note that $E_u^i$ has at most $n-1$ edges, that are updated $O(d)$ times during the algorithm, as we maintain distances up to distance $d+3k$, hence these edges add $O(nd)$ term to the complexity of maintaining one tree, and $O(n^2 d)$ for all $n$ trees. The edges in $E^*$ and $E_i$ are only added and removed once from a tree during the algorithm. From Claim~\ref{claim_Ei}, the total number of edges added to $E_i$ during the algorithm is bounded by $O(ns_{i-1})$ with high probability. Since these edges are part of all the ES-trees rooted at nodes in $D_i$, maintaining these edges adds $O(|D_i| n s_{i-1} d)$ term to the complexity.
This equals $$O \left( \frac{n \log{n}}{s_i} \cdot n s_{i-1} \cdot d \right) = O( n^{2-1/k}m^{1/k}(\log{n})^{1-1/k} d),$$ following the same calculations done in Claim~\ref{claim_E*} (but with an extra $nd$ term). Summing over all $k$ adds another multiplicative $k$ term to the complexity.
From Claim~\ref{claim_E*}, the total number of edges added to $E^*$ is bounded by $O(kn^{1-1/k}m^{1/k}(\log{n})^{1-1/k})$ in expectation. Summing over all $n$ trees, it adds another $O( kn^{2-1/k}m^{1/k}(\log{n})^{1-1/k} d)$ term to the complexity, completing the proof.
\end{proof}

\begin{claim}
The query complexity of the algorithm is constant.
\end{claim}

\begin{proof}
To answer a distance query for a pair of nodes $u,v$, we work as follows. If $u \in D_i, v \in D_j$, and $i<j$, our distance estimate is $\ell_u(v)$, otherwise our distance estimate is $\ell_v(u)$. Note that it is easy to maintain the values $\ell_u(v)$ in a way that the query time is constant. For example, we can maintain for each node $u$ an array with the values $\ell_u(v)$, and update the value in constant time each time it was changed in the algorithm.
\end{proof}


\subsubsection{Putting Everything Together} \label{sec_summary_additive}

In the last sections we showed an algorithm that for any pair of nodes at distance at most $d$ maintains a $+2(k-1)$-additive approximation in $\tilde{O}(n^{2-1/k}m^{1/k}d)$ time. To get our final result, we combine it with Theorem~\ref{thm_large} to handle large distances, proving the following.

\thmadd*

Note that for small values of $k$ and if $m=n^{1+\rho}$ for a constant $\rho>0$, we get a complexity of $\tilde{O}(n^{2-1/k}m^{1/k})$, where in the general case we have an extra $n^{o(1)}$ term.

\begin{proof}
We maintain simultaneously our dynamic data structure to handle short distances of at most $d$, and the data structure from Theorem~\ref{thm_large} to handle large distances of more than $d$, for a parameter~$d$ which we specify next. We run the algorithm from Theorem~\ref{thm_large} with parameters $\epsilon' = \epsilon/2$, and $k'$ to be specified later. Recall that this algorithm gives a $(1+\epsilon', \beta)$-approximation for $\beta= 2(1+2/\epsilon')^{k'-2}$. We define $d=\beta/\epsilon'$.

To answer a distance query for a pair of nodes $u,v$, we return the minimum between the values obtained in both algorithms. As in both algorithms the query time is constant, we get a constant query time. 

For the stretch analysis, if $d(u,v) \leq d$, then by Lemma~\ref{lem_additive_stretch}, we get an additive stretch of at most $+2(k-1)$. If $d(u,v) \geq d = \beta/\epsilon'$, then the distance estimate we get from the $(1+\epsilon',\beta)$-approximation algorithm is at most $$(1+\epsilon')d(u,v) + \beta \leq (1+\epsilon')d(u,v) + \epsilon' d(u,v) = (1+2\epsilon')d(u,v)=(1+\epsilon)d(u,v).$$
Hence, overall we get a $(1+\epsilon,2(k-1))$-approximation, as needed.

We now analyze the time complexity. From Lemma~\ref{lemma_time_additive}, the time complexity of our algorithm is $O(kn^{2-1/k}m^{1/k}(\log{n})^{1-1/k} d)$, and from Theorem~\ref{thm_large} the time complexity of the near-additive approximation algorithm is $\tilde{O}(n^{2+1/k'} (37/\epsilon)^{k'-1})$.
We distinguish between two choices for the value $k'$. First we choose $k' = k/\rho$, where $m=n^{1+\rho}$. In this case, the complexity of both expressions is bounded by $\tilde{O}((n^{2-1/k}m^{1/k}) O(1/\epsilon)^{k/\rho})$, as $d=\beta/\epsilon', \beta= 2(1+2/\epsilon')^{k'-2}$, and $O(n^{2-1/k}m^{1/k})=O(n^{2-1/k}n^{(1+\rho)/k})=O(n^{2+\rho/k})=O(n^{2+1/k'})$.

Second, we choose $k' = \sqrt{\log{n}}$. In this case, $\tilde{O}(n^{2+1/k'} (37/\epsilon)^{k'-1})= O(n^{2+o(1)})$. Similarly, $d= \beta/\epsilon' = O(1/\epsilon')^{k-1}= O(1/\epsilon)^{O(\sqrt{\log{n}})}=n^{o(1)}$. Hence, in total we get a complexity of $O(n^{2-1/k+o(1)}m^{1/k})$, when $m=\Omega(n)$.
\end{proof}

\bibliographystyle{alpha}
\bibliography{refs}

\appendix\newpage

\section{Monotone ES-tree} \label{sec:EStrees}
Algorithm~\ref{alg:monotone_ES_tree} provides pseudocode for the monotone ES-tree algorithm, this is taken from~\cite{henzinger2014decremental}. 
\begin{algorithm}[H]
\caption{Monotone ES-tree}
\label{alg:monotone_ES_tree}

\SetKwFunction{initialize}{Initialize}
\SetKwFunction{delete}{Delete}
\SetKwFunction{increase}{Increase}
\SetKwFunction{insert}{Insert}
\SetKwFunction{updateLevels}{UpdateLevels}

\tcp{Internal data structures:}
\tcp{$ N (u) $: for every node $ u $ a heap $ N (u) $ whose intended use is to store for every neighbor $ v $ of $ u $ in the current graph the value of $ \lev (v) + w_{H} (u, v) $}
\tcp{$ Q $: global heap whose intended use is to store nodes whose levels might need to be updated}

\BlankLine

\procedure{\initialize{}}{
	Compute shortest paths tree from $ s $ in $ H $ up to depth $ L $\;
	\ForEach{$ u \in V $}{
		Set $ \lev (u) = \dist_{H} (s, u) $\;
		\lFor{every edge $ (u, v) $ in $ H $}{
			insert $ v $ into heap $ N(u) $ of $ u $ with key $ \lev(v) + w_{H} (u, v) $
		}
	}
}

\BlankLine

\procedure{\delete{u, v}}{
	\increase{$u$, $v$, $\infty$} 
}

\BlankLine

\procedure{\increase{u, v, $w (u, v)$}}{
	\tcp{Increase weight of edge $ (u, v) $ to $ w(u, v) $}
	Insert $ u $ and $v$ into heap $ Q $ with keys $ \lev(u) $ and $\lev(v)$ respectively\;\label{line:insert u}
	Update key of $ v $ in heap $ N(u) $ to $ \lev(v) + w(u, v) $ and key of $ u $ in heap $ N(v) $ to $ \lev(u) + w(u, v) $\;\label{line:update N after increase}
	$ \updateLevels{} $\;
}

\BlankLine

\procedure{\insert{$u$, $v$, $ w (u, v)$}}{
	\tcp{Increase edge $ (u, v) $ of weight $ w(u, v) $}
	Insert $ v $ into heap $ N(u) $ with key $ \lev (v) + w (u, v) $ and $u$ into heap $N(v)$ with key $ \lev (u) + w_{H}(u, v) $\; 
}

\BlankLine

\procedure{\updateLevels{}}{
	\While{heap $ Q $ is not empty}{
		Take node $ u $ with minimum key $ \lev (u) $ from heap $ Q $ and remove it from $ Q $\;
		$ \lev'(u) \gets \min_{v} (\lev (v) + w_{H} (u, v)) $\;
		\tcp{$\min_{v} (\lev (v) + w_{H} (u, v)) $ can be retrieved from the heap $ N(u) $. $\arg\min_{v} (\lev (v) + w_{H} (u, v)) $ is $u$'s parent in the ES-tree. }
	
		\If{$ \lev'(u) > \lev (u) $}{\label{line:check_for_level_increase}
			$\lev(u)\gets \lev'(u)$\;\label{line:level_increase}
			\lIf{$ \lev' (u) > L $}{
				$ \lev (u) \gets \infty $
			}

		\ForEach{neighbor $ v $ of $ u $}{
				update key of $ u $ in heap $ N(v) $ to $ \lev(u) + w_{H} (u, v) $\;
				insert $ v $ into heap $ Q $ with key $ \lev(v) $ if $Q$ does not already contain $ v $\;
			}
		}
	}
}
\end{algorithm}

\remove{
For the analysis of the monotone ES-tree we will use the following terminology. We say that an edge $(u,v)$ is \emph{stretched} if $\ell(u) > \ell(v) + w_H(u,v)$. We say that a node $u$ is \emph{stretched} if it is incident to an edge $(u,v)$ that is stretched.

\begin{observation}[\cite{HKN2016}]\label{obs:simple_observations_monotone_ES}
The following holds for the monotone ES-tree:
\begin{enumerate}[label=(\arabic{*})]
\item \label{item: observation one} The level of a node never decreases.
\item \label{item: observation two} An edge can only become stretched when it is inserted.
\item \label{item: observation three} As long as a node is stretched, its level does not change. 
\item \label{item: observation four} For every tree edge $ (u, v) $ (where $ v $ is the parent of $ u $), $ \lev (u) \geq \lev (v) + w_{H} (u, v) $.
\end{enumerate}
\end{observation}
Observe that property~\ref{item: observation four} above implies that the returned distance estimate never underestimates the true distance.

\begin{lemma}[\cite{HKN2016}]\label{lm:ES-tree}
    For every $L \geq 1$, the total update time of a monotone ES-tree up to maximum level $L$ on a graph $H$ undergoing edge deletions, edge insertions, and edge weight increases is $O(\mathcal E(H)\cdot L+\mathcal W(H)+n)$, where $\mathcal E(H)$ is the total number of edges ever contained in $H$ and $\mathcal W(H)$ is the total number of updates to $H$.
\end{lemma}
}

\section{Maintaining Bunches and Clusters}\label{ap:bunches_and_clusters}
In this section we give a sketch of the algorithm of {\L}\k{a}cki and Nazari~\cite{LN2020} that allows us to maintain bunches and clusters needed for the decremental algorithm in \Cref{sec:TZ_2+eps}. 
One tool used is a decremental algorithm by Roditty and Zwick~\cite{RZ12} that maintains Thorup-Zwick clusters up to a bounded depth. This result can be summarized as follows:

\begin{lemma}[\cite{RZ12}]\label{lem:rz}
Let $0<p\leq 1$ and $d\geq 1$ be parameters. Given a dynamic graph $G=(V,E,w)$, that has edge deletions, weight increases, and edge insertions that do not decrease distances, we can maintain clusters (with sampling probability $p$) and bunches up to depth $d$ in total update time $\tilde O(|E(G)|d/p)$. Moreover, we maintain distances from each node to its cluster and its bunch, and these distances are monotone (non-decreasing).
\end{lemma}

This time dependence on $d$ is too expensive for many applications, and thus \cite{LN2020} improve the time for maintaining \textit{approximate} bunches and clusters by maintaining a hopset, combined with a well-known scaling tool. 

\paragraph{Hopsets.}
Hopsets, originally introduced by Cohen and Zwick~\cite{CohenZ01} in the context of parallel algorithms, are widely used in faster algorithms in decremental settings. We briefly define this concept and sketch its application. 
\begin{definition}
Let $G=(V,E,w)$ be a weighted undirected graph, and let $0<\epsilon <1$ be a parameter. A $(\beta,1+\epsilon)$-hopset is a graph $H=(V,E(H),w_H)$ such that for each $u,v\in V$ we have $d_G(u,v)\leq d^{(\beta)}_{G\cup H}(u,v) \leq (1+\epsilon) d_G(u,v)$, where $d^{(\beta)}_{G\cup H}(u,v)$ denotes the length of the shortest path between $u$ and $v$ in $G\cup H$ that uses at most $\beta$ hops. We say $\beta$ is the \emph{hopbound} of the hopset and $1+\epsilon$ is the \emph{stretch} of the hopset. 
\end{definition}

\begin{theorem}[\cite{LN2020}] \label{thm:hopset}
Let $0 < \rho < 1/2, 0< \epsilon <1$ be parameters. We can maintain a $(\beta,1+\epsilon)$-hopset $H$ for $G=(V,E,w)$ with hopbound $\beta =  O(\tfrac{\log(nW)}{\epsilon})^{2/\rho+1}$ in total update time $\tilde O (\tfrac{\beta\log^2(nW)}{\epsilon}mn^{\rho})$. Moreover, we can assume that the distances incurred by the additional edges of the hopset are monotone. Further we have that $H$ is of size $\tilde O(n^{1+\rho})$. 
\end{theorem}

By setting $\rho=\frac{\log \log n}{\sqrt{\log n}}$, we can maintain a hopset with hopbound $2^{\tilde{O}(\sqrt{\log n})}$ in $2^{\tilde{O}(\sqrt{\log n})}$ amortized time. We often choose this parameter setting when the graph is sparse for balancing the hopset maintenance time with the shortest path computation time.
By setting $\rho$ to a constant, we can maintain a hopset of polylogarithmic hopbound in time $O(mn^{\rho})$. We often use this parameter setting when we are maintaining distances from many sources and the graph is slightly denser (e.g.~$m=n^{1+\Omega(1)}$).

\paragraph{Scaling.}
The following scaling technique of Klein and Subramanian~\cite{KS97} changes a bounded hop distance problem into a bounded depth distance problem on a \textit{scaled graph}. Similar ideas have been widely used in dynamic algorithms, see e.g., \cite{Bernstein09,BR11,henzinger2014decremental,LN2020}. As we will see, the bounded depth works well together with the cluster and bunch maintaining algorithm of Roditty and Zwick~\cite{RZ12}
\begin{lemma}[\cite{KS97}]\label{lem:rounding}
Let $G=(V,E,w)$ be a weighted undirected graph. Let $R \geq 0$, $d \geq 1$, $\epsilon > 0$ be parameters.
We define a \emph{scaled graph} to be a graph $\textsc{Scale}(G, R, \epsilon, d)$ $:=  (V, E, \hat{w})$, such that $\hat{w}(e) = \lceil \tfrac{w(e)d}{R \epsilon} \rceil$.
Then for any path $\pi$ in $G$ such that $\pi$ has at most $d$ hops and weight $R \leq w(\pi) \leq 2R$ we have: 1) $\hat{w}(\pi) \leq \lceil 2d/\epsilon \rceil$, and 2) $w(\pi) \leq \tfrac{\epsilon R}{d} \cdot \hat{w}(\pi) \leq (1+\epsilon) w(\pi)$.
\end{lemma}

Hence given a hopset, we can `cover' the entire graph by scaled graphs, by considering separate scaled graphs $G_i$ for the distance scale $[2^i, 2^{i+1}]$ for $i=1,\dots, \log(nW)$. In order to combine this with a hopset algorithm, we need to scaled the graph after adding edges corresponding to a sequence of $2^i$-\textit{restricted hopsets}, i.e.~a graph that satisfies the hopset condition only for pairs within distance at most $2^i$. Slightly more formally, we define a scaled graph
as follows: $G^j := \textsc{Scale}(G \cup \bigcup_{r=0}^{j} H_{r}, 2^j, \epsilon, O(\beta)$, where $H_{r}$ is a $2^r$-restricted hopset. The final hopsets of \cite{LN2020} are the union of all these restricted hopsetes, however for maintaining distances efficiently we need to run the algorithm in Lemma \ref{lem:rz} on all these scaled graphs simultaneously and take the minimum estimate over all the scales. Importantly, to handle edge insertions they use monotone ES trees which implies that the distances obtained from this algorithm is monotone. This is important for us in bounding the running time of our decremental data structures.

The following lemma shows that using the specific hopsets by \cite{LN2020}, we can maintain approximate single-source distances. 


\lmbunch*
\begin{proof}[Proof sketch.]
First, using the algorithm of \cite{LN2020}, we can maintain a $(1+\epsilon, \beta)$-hopset $H$ of size $\tilde{O}(m+n^{1+\rho})$ in update time $\tilde{O}(mn^{\rho} \beta)$ with hopbound $\beta =\tilde{O}(\frac{1}{\rho\epsilon})^{2\rho}$, which is consisted of a sequence of $2^i$-restricted hopsets $H_1,...,H_{\log nW}$. We then use the framework of \cite{LN2020} to maintain \textit{approximate} clusters and hence approximate bunches on the scaled graphs. Since for all $u,v \in V$ in $G \cup H$ there is a path with $\beta$-hops that with length at most $(1+\epsilon)d_G(u,v)$, they show that it is enough to maintain $O(\beta/\epsilon)$-hop bounded bunches/clusters on scaled graphs $G_1,...,G_{\log n W}$. We consider $(4\beta/\epsilon)$-\textit{bounded bunches} $B_{G_i}(v)= \{ w: d_{G_i}(w,v) < d_{G_i}(v, p_i(v))\}$ where $p_i(v)$ is the closest sampled node in $G_i$ and $d_{G_i}(v, p_i(v)) \leq 4\beta/\epsilon$. The scaling ensures that for each node $v$ there exists $1 \leq i \leq \log nW$ such that $p_i(v)$ satisfies the property $d_{G_i}(v, p_i(v)) \leq 4\beta/\epsilon$. Let $d$ be the distances obtained by scaling back the estimates corresponding to $B_{G_i}(v) \cup p_i(v)$. There exists $i$ such that  $\tilde{p}(v) =p_i(v)$ satisfies $d_G(v,p(v)) \leq d(v,\tilde{p}(v)) \leq (1+\epsilon) d_G(v,p(v))$. Also by choosing the appropriate error parameters (to account for the scaling and hopset errors) in Theorem 14 of \cite{LN2020} we can ensure that for any $u \in \tilde{B}(v)$ we have $d_G(u,v) \leq d(u,v) \leq (1+\epsilon) d_G(u,v)$.
Consider any $u \in \hat{B}(v)$. We have $d(u,v) \leq (1+\epsilon)d_G(u,v) \leq d_G(v,p(v)) \leq d(v, \tilde{p}(v))$.
Hence we can see that the bunch $B_{G_i}(v) \subseteq \hat{B}(v)$ and thus $\tilde{B}(v):=\cup_{1 \leq i \leq \log nW} B_{G_i}$ contains all the nodes in $\hat{B}(v)$.

Finally since the algorithm of \cite{LN2020} (that is based on \cite{RZ12}) maintains bunches of size $O(\log n/p)$ on $O(\log nW)$ scaled graphs, we also have the bound on the size of the bunch.
\end{proof}


\end{document}